\documentclass[journal,onecolumn,12pt]{IEEEtran}%
\usepackage{amssymb}
\usepackage{amsfonts}
\usepackage{amsmath}
\usepackage{epic}
\usepackage{geometry}
\usepackage{doublespace}%
\usepackage{graphicx}
\newtheorem{theorem}{Theorem}

\newtheorem{corollary}{Corollary}

\newtheorem{definition}{Definition}
\newtheorem{example}{Example}

\newtheorem{lemma}{Lemma}

\newtheorem{proposition}{Proposition}
\newtheorem{remark}{Remark}

\begin{document}

\title{{\LARGE Optimization of Lyapunov Invariants\vspace*{-0.35in} in Verification
of Software Systems}\vspace*{-0.2in}}
\author{{\small Mardavij Roozbehani,~\IEEEmembership{{\small Member,~IEEE}}, Alexandre
Megretski ,~\IEEEmembership{{\small Member,~IEEE}}, and Eric Feron}%
~\IEEEmembership{{\small Member,~IEEE}}\vspace*{-0.75in} \thanks{Mardavij
Roozbehani and Alexandre Megretski are with the Laboratory for Information and
Decision Systems (LIDS), Massachusetts Institute of Technology, Cambridge, MA.
E-mails: \{mardavij,ameg\}@mit.edu. Eric Feron is with the School of Aerospace
Engineering, Georgia Institute of Technology, Atlanta, GA. E-mail:
feron@gatech.edu.\vspace*{-0.07in}}\vspace*{-0.2in}}
\maketitle

\begin{abstract}
\vspace*{-0.1in}The paper proposes a control-theoretic framework for
verification of numerical software systems, and puts forward software
verification as an important application of control and systems theory. The
idea is to transfer Lyapunov functions and the associated computational
techniques from control systems analysis and convex optimization to
verification of various software safety and performance specifications. These
include but are not limited to absence of overflow, absence of
division-by-zero, termination in finite time, presence of dead-code, and
certain user-specified assertions. Central to this framework are Lyapunov
invariants. These are properly constructed functions of the program variables,
and satisfy certain properties\textbf{---}resembling those of Lyapunov
functions\textbf{---}along the execution trace. The search for the invariants
can be formulated as a convex optimization problem. If the associated
optimization problem is feasible, the result is a certificate for the
specification.\vspace*{-0.15in}

\end{abstract}

\markboth{Extended Version, August 2011}{Shell \MakeLowercase{\textit{et al.}}: Bare Demo of
IEEEtran.cls for Journals}

\begin{keywords}
\vspace*{-0.1in}Software Verification, Lyapunov Invariants, Convex
Optimization.\vspace*{-0.2in}
\end{keywords}

\section{Introduction\vspace*{-0.04in}}

\PARstart{S}{oftware} in safety-critical systems implement complex algorithms
and feedback laws that control the interaction of physical devices with their
environments. Examples of such systems are abundant in aerospace, automotive,
and medical applications. The range of theoretical and practical issues that
arise in analysis, design, and implementation of safety-critical software
systems is extensive, see, e.g., \cite{Kopetz2001}, \cite{Murthy2001} , and
\cite{Heck2003}. While safety-critical software must satisfy various resource
allocation, timing, scheduling, and fault tolerance constraints, the foremost
requirement is that it must be free of run-time errors.\vspace*{-0.15in}

\subsection{Overview of Existing Methods\vspace*{-0.1in}}

\subsubsection{Formal Methods}

Formal verification methods are model-based techniques \cite{Pel01},
\cite{Nielson}, \cite{MitraThesis} for proving or disproving that a
mathematical model of a software (or hardware) satisfies a given
\textit{specification,} i.e., a mathematical expression of a desired behavior.
The approach adopted in this paper too, falls under this category. Herein, we
briefly review \textit{model checking} and \textit{abstract}
\textit{interpretation}.

\paragraph{Model Checking}

In \textit{model checking} \cite{ClarkBook} the system is modeled as a finite
state transition system and the specifications are expressed in some form of
logic formulae, e.g., temporal or propositional logic. The verification
problem then reduces to a graph search, and symbolic algorithms are used to
perform an exhaustive exploration of all possible states. Model checking has
proven to be a powerful technique for verification of circuits
\cite{Clarkware}, security and communication protocols \cite{Marrero},
\cite{Naumovich} and stochastic processes \cite{Baier}. Nevertheless, when the
program has non-integer variables, or when the state space is continuous,
model checking is not directly applicable. In such cases, combinations of
various abstraction techniques and model checking have been proposed
\cite{Alur2002, Dams1996, Tiwari2002}; scalability, however, remains a challenge.

\paragraph{Abstract Interpretation}

is a theory for formal approximation of the \textit{operational semantics} of
computer programs in a systematic way \cite{Cousot1977}. Construction of
abstract models involves abstraction of domains\textbf{---}typically in the
form of a combination of sign, interval, polyhedral, and congruence
abstractions of sets of data\textbf{---}and functions. A system of fixed-point
equations is then generated by symbolic forward/backward executions of the
abstract model. An iterative equation solving procedure, e.g., Newton's
method, is used for solving the nonlinear system of equations, the solution of
which results in an inductive invariant assertion, which is then used for
checking the specifications. In practice, to guarantee finite convergence of
the iterates, narrowing (outer approximation) operators are constructed to
estimate the solution, followed by widening (inner approximation) to improve
the estimate \cite{Cousot2001}. This compromise can be a source of
conservatism in analysis \cite{Adge11}. Nevertheless, these methods have been
used in practice for verification of limited properties of embedded software
of commercial aircraft \cite{ASTREE}.

Alternative formal methods can be found in the computer science literature
mostly under \textit{deductive verification }\cite{Manna1995}, \textit{type
inference }\cite{Pierece}\textit{,} and \textit{data flow analysis
}\cite{Heckt1977}. These methods share extensive similarities in that a notion
of program abstraction and symbolic execution or constraint propagation is
present in all of them. Further details and discussions of the methodologies
can be found in \cite{Cousot2001}, and \cite{Nielson}.\vspace*{-0.02in}

\subsubsection{System Theoretic Methods\vspace*{-0.02in}}

While software analysis has been\vspace*{-0.02in} the subject of an extensive
body of research in computer science, treatment of the topic in the control
\vspace*{-0.02in}systems community has been less systematic. The
relevant\vspace*{-0.02in} results in the systems and control literature can be
found in the field of hybrid systems\vspace*{-0.02in} \cite{Branicky}. Much of
the available techniques for safety verification of hybrid\vspace*{-0.02in}
systems are explicitly or implicitly based on computation of the reachable
sets, either exactly\vspace*{-0.02in} or approximately. These include but are
not limited to techniques based on quantifier \vspace*{-0.02in}elimination
\cite{Lafferriere2001}, ellipsoidal calculus \cite{Krzhanski1996}, and
mathematical programming\vspace*{-0.02in} \cite{Bemporad2000}. Alternative
approaches aim at establishing properties of hybrid systems \vspace
*{-0.02in}through barrier certificates \cite{Prajna2005}, numerical
computation\vspace*{-0.02in} of Lyapunov functions \cite{Branicky1998,
Johansson1998}, or by combined use of bisimulation mechanisms and Lyapunov
techniques \cite{GirardPappsVerification, Laferriere1999, Tiwari2002,
Alur2002}.\vspace*{-0.02in}

Inspired by the concept\vspace*{-0.02in} of Lyapunov functions in stability
analysis of nonlinear dynamical systems \cite{Khalil}, in this paper we
propose Lyapunov invariants for analysis of computer programs.\vspace
*{-0.02in} While Lyapunov functions and similar concepts have been used in
verification of stability or \vspace*{-0.01in}temporal properties of system
level descriptions of hybrid systems \cite{Prajna2007}, \cite{Branicky1998},
\cite{Johansson1998}, to the \vspace*{-0.02in}best of our knowledge, this
paper is the first to present a systematic\vspace*{-0.02in} framework based on
Lyapunov invariance and convex optimization for verification of a broad range
of code-level specifications for computer programs\footnote{This paper
constitutes the synthesis and extension of ideas and computational techniques
expressed in the workshop \cite{FeronWorkshop}, and subsequent conference
papers \cite{RMF-ACC05}, \cite{RFM-HSCC05} and \cite{RoozMegFer06}. Some of
the ideas presented in \cite{FeronWorkshop}, \cite{RMF-ACC05}, and
\cite{RFM-HSCC05} were independently reported in \cite{FeronGiveAway}.
\vspace*{-0.18in}}.\vspace*{-0.02in} Accordingly, it is in the systematic
integration of new ideas and some well-known\vspace*{-0.02in} tools within a
unified software analysis framework that we see the\vspace*{-0.02in} main
contribution of our work, and not in carrying through the proofs\vspace
*{-0.01in} of the underlying theorems and propositions. The introduction and
development of such framework\vspace*{-0.02in} provides an opportunity for the
field of \textit{control} to systematically address\vspace*{-0.02in} a problem
of great practical significance and interest to both computer science and
engineering communities. The framework can be summarized as follows:\vspace
*{-0.02in}

\begin{enumerate}
\item {\normalsize Dynamical system interpretation\vspace*{-0.02in} and
modeling (Section \ref{Modeling}). We introduce generic dynamical system
representations of programs, along with specific modeling languages\vspace
*{-0.02in} which include Mixed-Integer Linear Models (MILM), Graph Models, and
MIL-over-Graph Hybrid Models (MIL-GHM).}\vspace*{-0.02in}

\item {\normalsize Lyapunov invariants as behavior\vspace*{-0.01in}
certificates for computer programs (Section \ref{Chapter:LyapunovInvs}).
Analogous to a Lyapunov function, a Lyapunov invariant is a real-valued
function\vspace*{-0.01in} of the program variables, and satisfies a
\textit{difference inequality}\ along the trace of the program. It is shown
that such functions can be formulated for verification of various
specifications.\vspace*{-0.06in}}

\item {\normalsize A computational procedure for\vspace*{-0.01in} finding the
Lyapunov invariants (Section \ref{Chapter:Computation}). The procedure is
standard and constitutes\ these steps: (i) Restricting the search\vspace
*{-0.01in} space to a linear subspace. (ii) Using convex relaxation techniques
to formulate the search problem as a convex optimization problem, e.g., a
Linear\vspace*{-0.01in} Program (LP) \cite{Bertsimes1997}, Semidefinite
Program (SDP) \cite{Boyd1994, VB}, or a Sum-of-Squares (SOS) program
\cite{Parrilo2001}. (iii) Using convex optimization software for numerical
computation of the certificates.}$\vspace*{-0.16in}$
\end{enumerate}

\section{Dynamical System Interpretation and Modeling of Computer Programs
\label{Modeling}$\vspace*{-0.11in}$}

We interpret computer programs as discrete-time dynamical systems and
introduce generic models that formalize this interpretation. We then introduce
MILMs, Graph Models, and MIL-GHMs as structured cases of the generic models.
The specific modeling languages are used for computational purposes.$\vspace
*{-0.26in}$

\subsection{Generic Models\label{Sec:GenRep}$\vspace*{-0.11in}$}

\subsubsection{Concrete Representation of Computer Programs}

We will consider generic models defined by a finite state space set $X$ with
selected subsets $X_{0}\subseteq X$ of initial states, and $X_{\infty}\subset
X$ of terminal states, and by a set-valued state transition function
$f:X\mapsto2^{X}$, such that $f(x)\subseteq X_{\infty},\forall x\in X_{\infty
}.$ We denote such dynamical systems by $\mathcal{S}(X,f,X_{0},X_{\infty}).$

\begin{definition}
\label{ConcreteRepDef}The dynamical system $\mathcal{S}(X,f,X_{0},X_{\infty})$
is a $\mathcal{C}$-representation of a computer program $\mathcal{P},$ if the
set of all sequences that can be generated by $\mathcal{P}$ is equal to the
set of all sequences $\mathcal{X}=(x(0),x(1),\dots,x(t),\dots)$ of elements
from $X,$ satisfying\vspace*{-0.25in}%
\begin{equation}
x\left(  0\right)  \in X_{0}\subseteq X,\qquad x\left(  t+1\right)  \in
f\left(  x\left(  t\right)  \right)  \text{\qquad}\forall t\in\mathbb{Z}%
_{+}\vspace*{-0.22in}\label{Softa1}%
\end{equation}
The uncertainty in $x(0)$ allows for dependence of the program on different
initial conditions, and the uncertainty in $f$ models\ dependence on
parameters, as well as the ability to respond to real-time inputs.
\end{definition}

\begin{example}
\label{IntegerDiv-Ex}\textbf{Integer Division }(adopted from \cite{Pel01}%
):\ The functionality\vspace*{-0.01in} of Program 1 is to compute the result
of the integer division of $\mathrm{dd}$ (dividend) by $\mathrm{dr}$
(divisor).\vspace*{-0.01in} A $\mathcal{C}$-representation of the program is
displayed alongside. Note that if $\mathrm{dd}\geq0,$ and $\mathrm{dr}\leq0,$
then the\vspace*{-0.01in} program never exits the \textquotedblleft
while\textquotedblright\ loop and the value of $\mathrm{q}$ keeps
increasing\vspace*{-0.01in}, eventually leading to either an overflow or an
erroneous answer. The program terminates\ if $\mathrm{dd}$ and $\mathrm{dr}$
are positive.
\end{example}

\vspace*{-0.55in}

{\small
\begin{gather*}%
\begin{array}
[c]{cl}%
\begin{tabular}
[c]{|l|}\hline
$%
\begin{array}
[c]{l}%
\mathrm{int~IntegerDivision~(~int~dd,int~dr~)}\\
\vspace*{-0.46in}\\
\mathrm{\{int~q=\{0\};~int~r=\{dd\};}\\
\vspace*{-0.46in}\\
\mathrm{{while}\text{ }{(r>=dr)}}\\
\vspace*{-0.46in}\\
\mathrm{{{{\{\hspace{0.13in}q=q+1;}}}}\\
\vspace*{-0.46in}\\
\mathrm{{{\hspace{0.24in}r=r-dr;\}}}}\\
\vspace*{-0.46in}\\
\mathrm{return~r;\}}%
\end{array}
$\\\hline
\end{tabular}
& \hspace*{-0.05in}%
\begin{tabular}
[c]{|l|}\hline
$%
\begin{array}
[c]{l}%
\underline{\mathbb{Z}}=\mathbb{Z\cap}\left[  -32768,32767\right] \\
\vspace*{-0.46in}\\
X=\underline{\mathbb{Z}}^{4}\\
\vspace*{-0.46in}\\
X_{0}=\left\{  (\mathrm{dd},\mathrm{dr},\mathrm{q},\mathrm{r})\in
X~|~\mathrm{q}=0,\text{ }\mathrm{r}=\mathrm{dd}\right\} \\
\vspace*{-0.46in}\\
X_{\infty}=\left\{  (\mathrm{dd},\mathrm{dr},\mathrm{q},\mathrm{r})\in
X~|~\mathrm{r}<\mathrm{dr}\right\} \\
\vspace*{-0.35in}\\
f:(\mathrm{dd},\mathrm{dr},\mathrm{q},\mathrm{r})\mapsto\left\{
\begin{array}
[c]{l}%
\vspace*{-0.44in}\\
(\mathrm{dd},\mathrm{dr},\mathrm{q}+\mathrm{1},\mathrm{r}-\mathrm{dr}),\\
\vspace*{-0.44in}\\
(\mathrm{dd},\mathrm{dr},\mathrm{q},\mathrm{r}),
\end{array}
\right.
\begin{array}
[c]{l}%
\vspace*{-0.44in}\\
(\mathrm{dd},\mathrm{dr},\mathrm{q},\mathrm{r})\in X\backslash X_{\infty}\\
\vspace*{-0.44in}\\
(\mathrm{dd},\mathrm{dr},\mathrm{q},\mathrm{r})\in X_{\infty}%
\end{array}
\end{array}
$\\\hline
\end{tabular}
\end{array}
\\
\text{{Program 1: The Integer Division Program (left) and its Dynamical System
Model (right)\vspace*{-0.2in}}}%
\end{gather*}
}

{\vspace*{-0.62in}}

\subsubsection{Abstract Representation of Computer Programs\vspace*{-0.01in}}

In a $\mathcal{C}$-representation, the elements of the state space $X$ belong
to a finite\vspace*{-0.01in} subset of the set of rational numbers that can be
represented by a fixed number of bits in a specific arithmetic\vspace
*{-0.01in} framework, e.g., fixed-point or floating-point arithmetic. When the
elements of $X$ are\vspace*{-0.01in} non-integers, due to the quantization
effects, the set-valued map $f$ often defines very complicated
dependencies\vspace*{-0.01in} between the elements of $X,$ even for simple
programs involving only elementary arithmetic operations. An abstract\vspace
*{-0.01in} model over-approximates the behavior set in the interest of
tractability. The drawbacks are conservatism\vspace*{-0.01in} of the analysis
and (potentially) undecidability. Nevertheless,\vspace*{-0.01in} abstractions
in the form of formal over-approximations make it possible to formulate
computationally tractable,\vspace*{-0.01in} sufficient conditions for a
verification problem that would otherwise be intractable.

\begin{definition}
\label{Def:abst}Given a program $\mathcal{P}$ and its $\mathcal{C}%
$-representation $\mathcal{S}(X,f,X_{0},X_{\infty})$, we say that
$\overline{\mathcal{S}}(\overline{X},\overline{f},\overline{X}_{0}%
,\overline{X}_{\infty})$ is an $\mathcal{A}$-representation, i.e., an
\emph{abstraction} of $\mathcal{P}$, if $X\subseteq\overline{X}$,
$X_{0}\subseteq\overline{X}_{0}$, and $f(x)\subseteq\overline{f}(x)$\ for all
$x\in X,$ and the following condition holds:\vspace*{-0.07in}%
\begin{equation}
\overline{X}_{\infty}\cap X\subseteq X_{\infty}.\vspace*{-0.07in}%
\label{terminalabstract}%
\end{equation}

\end{definition}

Thus, every trajectory of the actual program is also a trajectory of the
abstract model. The definition of $\overline{X}_{\infty}$ is slightly more
subtle. For proving Finite-Time Termination (FTT), we need to be able to infer
that if all the trajectories of $\overline{\mathcal{S}}$ eventually enter
$\overline{X}_{\infty},$ then all trajectories of $\mathcal{S}$ will
eventually enter $X_{\infty}.$ It is tempting to require that $\overline
{X}_{\infty}\subseteq X_{\infty}$, however, this may not be possible as
$X_{\infty}$ is often a discrete set, while $\overline{X}_{\infty}$ is dense
in the domain of real numbers. The definition of $\overline{X}_{\infty}$ as in
(\ref{terminalabstract}) resolves this issue.

Construction of $\overline{\mathcal{S}}(\overline{X},\overline{f},\overline
{X}_{0},\overline{X}_{\infty})$ from $\mathcal{S}(X,f,X_{0},X_{\infty})$
involves abstraction of each of the elements $X,~f,~X_{0},~X_{\infty}$ in a
way that is consistent with Definition \ref{Def:abst}. Abstraction of the
state space $X$ often involves replacing the domain of \textit{floats} or
integers or a combination of these by the domain of real numbers. Abstraction
of $X_{0}$ or $X_{\infty}$ often involves a combination of domain abstractions
and abstraction of functions that define these sets. Semialgebraic set-valued
abstractions of some commonly used nonlinearities is presented in Appendix I. Also, abstractions of fixed-point and floating point arithmetic operations based on ideas from \cite{MineThesis} and \cite{MinePaper} are discussed in Appendix I. A case study in application of these methods to analysis of a program with floating-point operations is presented in Section \ref{sec:casestudy2}. \vspace*{-0.25in}

\subsection{Specific Models of Computer Programs\label{Section:SpecModels}%
\vspace*{-0.05in}}

Specific modeling languages are particularly useful for automating the proof
process in a computational framework. Here, three specific modeling languages
are proposed: \textit{Mixed-Integer Linear Models (MILM),} \textit{Graph
Models}, and \textit{Mixed-Integer Linear over Graph Hybrid Models (MIL-GHM).}

\subsubsection{Mixed-Integer Linear Model (MILM)}

Proposing MILMs for software modeling and analysis is motivated by the
observation that by imposing linear equality constraints on boolean and
continuous variables over a quasi-hypercube, one can obtain a relatively
compact representation of arbitrary piecewise affine functions defined over
compact polytopic subsets of Euclidean spaces (Proposition \ref{MILM-prop}).
The earliest reference to the statement of universality of MILMs appears to be
\cite{nem}, in which a constructive proof is given for the one-dimensional
case. A constructive proof for the general case is given in Appendix II.
\vspace*{-0.05in}

\begin{proposition}
\label{MILM-prop}\textbf{Universality of Mixed-Integer Linear Models.} Let
$f:X\mapsto\mathbb{R}^{n}$ be a piecewise affine map with a closed graph,
defined on a compact state space $X\subseteq\left[  -1,1\right]  ^{n},$
consisting of a finite union of compact polytopes. That is:\vspace*{-0.28in}%
\[
f\left(  x\right)  \in2A_{i}x+2B_{i}\qquad\text{subject to\qquad}x\in
X_{i},\text{\ }i\in\mathbb{Z}\left(  1,N\right)  \vspace*{-0.15in}%
\]
where, each $X_{i}$ is a compact polytopic set. Then, $f$ can be specified
precisely, by imposing linear equality constraints on a finite number of
binary and continuous variables ranging over compact intervals. Specifically,
there exist matrices $F$ and $H,$ such that the following two sets are
equal:\vspace*{-0.25in}%
\begin{align*}
G_{1}  & =\left\{  \left(  x,f\left(  x\right)  \right)  ~|~x\in X\right\}
\vspace*{-0.1in}\\
G_{2}  & =\{\left(  x,y\right)  ~|~F[%
\begin{array}
[c]{cccc}%
\hspace*{-0.04in}\vspace*{0.04in}x\hspace*{-0.02in} & \hspace*{-0.02in}%
w\hspace*{-0.02in} & \hspace*{-0.02in}v\hspace*{-0.02in} & \hspace
*{-0.02in}1\hspace*{-0.04in}%
\end{array}
]^{^{T}}=y,\text{ }H[%
\begin{array}
[c]{cccc}%
\hspace*{-0.04in}\vspace*{0.04in}x\hspace*{-0.02in} & \hspace*{-0.02in}%
w\hspace*{-0.02in} & \hspace*{-0.02in}v\hspace*{-0.02in} & \hspace
*{-0.02in}1\hspace*{-0.04in}%
\end{array}
]^{^{T}}=0,\text{ }\left(  w,v\right)  \in\left[  -1,1\right]  ^{n_{w}}%
\times\left\{  -1,1\right\}  ^{n_{v}}\}
\end{align*}
\vspace*{-0.5in}
\end{proposition}

Mixed Logical Dynamical Systems (MLDS) with similar structure were considered
in \cite{Bemporad Morari} for analysis of a class of hybrid systems. The main
contribution here is in the application of the model to software analysis. A
MIL model of a computer program is defined via the following
elements:{\small \vspace*{-0.00in}}

\begin{enumerate}
\item The state space $X\subset\left[  -1,1\right]  ^{n}$.\vspace*{-0.00in}

\item Letting $n_{e}=n+n_{w}+n_{v}+1,$ the state transition function
$f:X\mapsto2^{X}$\ is defined by two matrices $F,$~and $H$\ of dimensions
$n$-by-$n_{e}$\ and $n_{H}$-by-$n_{e}$\ respectively, according to:\vspace
*{-0.0in}
\begin{equation}
f(x)\in\left\{  F[%
\begin{array}
[c]{cccc}%
\hspace*{-0.04in}\vspace*{0.04in}x\hspace*{-0.02in} & \hspace*{-0.02in}%
w\hspace*{-0.02in} & \hspace*{-0.02in}v\hspace*{-0.02in} & \hspace
*{-0.02in}1\hspace*{-0.04in}%
\end{array}
]^{^{T}}~|~~H[%
\begin{array}
[c]{cccc}%
\hspace*{-0.04in}\vspace*{0.04in}x\hspace*{-0.02in} & \hspace*{-0.02in}%
w\hspace*{-0.02in} & \hspace*{-0.02in}v\hspace*{-0.02in} & \hspace
*{-0.02in}1\hspace*{-0.04in}%
\end{array}
]^{^{T}}=0,\text{ }\left(  w,v\right)  \in\left[  -1,1\right]  ^{n_{w}}%
\times\left\{  -1,1\right\}  ^{n_{v}}\right\}  .\vspace*{-0.05in}\label{MILM1}%
\end{equation}

\item The set of initial conditions is defined via either of the
following:\vspace*{-0.0in}

\begin{enumerate}
\item If $X_{0}$\ is finite with a small cardinality, then it can be
conveniently specified by extension. We will see in Section
\ref{Chapter:Computation} that per each element of $X_{0},$\ one constraint
needs to be included in the set of constraints of the optimization problem
associated with the verification task.

\item If $X_{0}$\ is not finite, or $\left\vert X_{0}\right\vert $\ is too
large, an abstraction of $X_{0}$\ can be specified by a matrix $H_{0}\in
R^{n_{H_{0}}\times n_{e}}$\ which defines a union of compact polytopes in the
following way:\vspace*{-0.05in}%
\begin{equation}
X_{0}=\{x\in X~|~H_{0}[%
\begin{array}
[c]{cccc}%
\hspace*{-0.04in}\vspace*{0.04in}x\hspace*{-0.02in} & \hspace*{-0.02in}%
w\hspace*{-0.02in} & \hspace*{-0.02in}v\hspace*{-0.02in} & \hspace
*{-0.02in}1\hspace*{-0.04in}%
\end{array}
]^{^{T}}=0,~\left(  w,v\right)  \in\left[  -1,1\right]  ^{n_{w}}\times\left\{
-1,1\right\}  ^{n_{v}}\}.\vspace*{-0.1in}\label{MILM2}%
\end{equation}

\end{enumerate}

\item The set of terminal states $X_{\infty}$\ is defined by\vspace*{-0.1in}%
\begin{equation}
X_{\infty}=\{x\in X~|~H[%
\begin{array}
[c]{cccc}%
\hspace*{-0.04in}\vspace*{0.04in}x\hspace*{-0.02in} & \hspace*{-0.02in}%
w\hspace*{-0.02in} & \hspace*{-0.02in}v\hspace*{-0.02in} & \hspace
*{-0.02in}1\hspace*{-0.04in}%
\end{array}
]^{^{T}}\neq0,~\forall w\in\left[  -1,1\right]  ^{n_{w}},~\forall v\in\left\{
-1,1\right\}  ^{n_{v}}\}.\vspace*{-0.1in}\label{MILM3}%
\end{equation}

\end{enumerate}

Therefore, $\mathcal{S}(X,f,X_{0},X_{\infty})$ is well defined. A compact
description of a MILM of a program is either of the form $\mathcal{S}\left(
F,H,H_{0},n,n_{w},n_{v}\right)  ,$ or of the form $\mathcal{S}\left(
F,H,X_{0},n,n_{w},n_{v}\right)  $. The MILMs can represent a broad range of
computer programs of interest in control applications, including but not
limited to control programs of gain scheduled linear systems in embedded
applications. In addition, generalization of the model to programs with
piecewise affine dynamics subject to quadratic constraints is straightforward.

\begin{example}
A MILM of an abstraction of the {\small IntegerDivision} program (Program 1:
Section \ref{Sec:GenRep})$,$ with all the integer variables replaced with real
variables, is given by $\mathcal{S}\left(  F,H,H_{0},4,3,0\right)  ,$
where$\vspace*{-0.1in}${\small
\[
\hspace*{-0.07in}%
\begin{array}
[b]{lll}%
H_{0}= & H= & F=\\
\left[
\begin{array}
[c]{rrrrrrrr}%
\vspace*{-0.42in} &  &  &  &  &  &  & \\
1 & 0 & 0 & -1 & 0 & 0 & 0 & 0\vspace*{-0.12in}\\
0 & 0 & 1 & 0 & 0 & 0 & 0 & 0\vspace*{-0.12in}\\
0 & -2 & 0 & 0 & 0 & 1 & 0 & 1\vspace*{-0.12in}\\
-2 & 0 & 0 & 0 & 0 & 0 & 1 & 1\vspace*{-0.07in}%
\end{array}
\right]  , & \left[  \hspace*{-0.05in}%
\begin{array}
[c]{rrrrrrrr}%
\vspace*{-0.42in} &  &  &  &  &  &  & \\
0 & 2 & 0 & -2 & 1 & 0 & 0 & 1\vspace*{-0.12in}\\
0 & -2 & 0 & 0 & 0 & 1 & 0 & 1\vspace*{-0.12in}\\
-2 & 0 & 0 & 0 & 0 & 0 & 1 & 1\vspace*{-0.07in}%
\end{array}
\hspace*{-0.05in}\right]  , & \left[
\begin{array}
[c]{rrrrrrrr}%
\vspace*{-0.42in} &  &  &  &  &  &  & \\
1 & 0 & 0 & 0 & 0 & 0 & 0 & 0\vspace*{-0.12in}\\
0 & 1 & 0 & 0 & 0 & 0 & 0 & 0\vspace*{-0.12in}\\
0 & 0 & 1 & 0 & 0 & 0 & 0 & 1/M\vspace*{-0.12in}\\
0 & -1 & 0 & 1 & 0 & 0 & 0 & 0\vspace*{-0.07in}%
\end{array}
\right]
\end{array}
\]
} Here, $M$ is a scaling parameter used for bringing all the variables within
the interval $\left[  -1,1\right]  .\vspace*{-0.05in}$
\end{example}

\subsubsection{Graph Model\label{graph models: section}}

Practical considerations such as universality and strong resemblance to the
natural flow of computer code render graph models an attractive and convenient
model for software analysis. Before we proceed, for convenience, we introduce
the following notation: $P_{r}\left(  i,x\right)  $ denotes the projection
operator defined as $P_{r}\left(  i,x\right)  =x,$ for all $i\in\mathbb{Z\cup
}\left\{  \Join\right\}  ,$ and all $x\in\mathbb{R}^{n}.$

A graph model is defined on a directed graph $G\left(  \mathcal{N}%
,\mathcal{E}\right)  $ with the following elements:\vspace*{-0.05in}

\begin{enumerate}
\item A set of nodes $\mathcal{N}=\{\emptyset\}\cup\{1,\dots,m\}\cup\left\{
\Join\right\}  .$ These can be thought of as line numbers or code locations.
Nodes $\emptyset$ and $\Join$ are starting and terminal nodes, respectively.
The only possible transition from node $\Join$ is the identity transition to
node $\Join.$\vspace*{-0.05in}

\item A set of edges $\mathcal{E}=\left\{  \left(  i,j,k\right)  \text{
}|\text{ }i\in\mathcal{N},\text{ }j\in\mathcal{O}\left(  i\right)  \right\}
,$ where the \textit{outgoing set} $\mathcal{O}\left(  i\right)  $ is the set
of all nodes to which transition from node $i$ is possible in one step.
Definition of the \textit{incoming set} $\mathcal{I}\left(  i\right)  $ is
analogous. The third element in the triplet $\left(  i,j,k\right)  $ is the
index for the $k$th edge between $i$ and $j,$ and $\mathcal{A}_{ji}=\left\{
k~|~\left(  i,j,k\right)  \in\mathcal{E}\right\}  .$\vspace*{-0.05in}

\item A set of program variables $x_{l}\in\Omega\subseteq\mathbb{R},$
$l\in\mathbb{Z}\left(  1,n\right)  .$ Given $\mathcal{N}$ and $n$, the state
space of a graph model is $X=\mathcal{N}\times\Omega^{n}$. The state
$\widetilde{x}=\left(  i,x\right)  $ of a graph model has therefore, two
components: The discrete component $i\in\mathcal{N},$ and the continuous
component $x\in\Omega^{n}\subseteq\mathbb{R}^{n}$.\vspace*{-0.05in}

\item A set of \textit{transition} labels $\overline{T}_{ji}^{k}$ assigned to
every edge $\left(  i,j,k\right)  \in\mathcal{E}$, where $\overline{T}%
_{ji}^{k}$ maps $x$ to the set $\overline{T}_{ji}^{k}x=\{T_{ji}^{k}\left(
x,w,v\right)  ~|~\left(  x,w,v\right)  \in S_{ji}^{k}\},$ where $\left(
w,v\right)  \in\left[  -1,1\right]  ^{n_{w}}\times\left\{  -1,1\right\}
^{n_{v}},$ and $T_{ji}^{k}:\mathbb{R}^{n+n_{w}+n_{v}}\mapsto\mathbb{R}^{n}$ is
a polynomial function and $S_{ji}^{k}$ is a semialgebraic set$.$ If
$\overline{T}_{ji}^{k}$ is a deterministic map, we drop $S_{ji}^{k}$ and
define $\overline{T}_{ji}^{k}\equiv T_{ji}^{k}\left(  x\right)  $%
.\vspace*{-0.05in}

\item A set of \textit{passport }labels $\Pi_{ji}^{k}$ assigned to all edges
$\left(  i,j,k\right)  \in\mathcal{E}$, where $\Pi_{ji}^{k}$ is a
semialgebraic set. A state transition along edge $\left(  i,j,k\right)  $ is
possible if and only if $x\in\Pi_{ji}^{k}.$\vspace*{-0.05in}

\item Semialgebraic invariant sets $X_{i}\subseteq\Omega^{n},$ $i\in
\mathcal{N}$ are assigned to every node on the graph, such that $P_{r}\left(
i,x\right)  \in X_{i}.$ Equivalently, a state $\widetilde{x}=\left(
i,x\right)  $ satisfying $x\in X\backslash X_{i}$ is unreachable.\vspace
*{-0.05in}
\end{enumerate}

Therefore, a graph model is a well-defined specific case of the generic model
$\mathcal{S(}X,f,X_{0},X_{\infty}),$ with $X=\mathcal{N}\times\Omega^{n},$
$X_{0}=\left\{  \emptyset\right\}  \times X_{\emptyset},$ $X_{\infty}=\left\{
\Join\right\}  \times X_{\Join}$ and $f:X\mapsto2^{X}$ defined as:\vspace
*{-0.2in}%
\begin{equation}
f\left(  \widetilde{x}\right)  \equiv f\left(  i,x\right)  =\left\{
(j,\overline{T}_{ji}^{k}x)~|~j\in\mathcal{O}\left(  i\right)  ,\text{ }x\in
\Pi_{ji}^{k}\cap X_{i}\right\}  .\vspace*{-0.25in}\label{GraphUpdate}%
\end{equation}

Conceptually similar models, namely control flow graphs, have been reported in
\cite{Pel01} (and the references therein)\ for software verification, and in
\cite{Alur1995, Brocket} for modeling and verification of hybrid
systems.\vspace*{-0.04in}

\textbf{Remarks}$\vspace*{-0.06in}$

\begin{enumerate}
\item[\textbf{--}] The invariant set of node $\emptyset$ contains all the
available information about the initial conditions of the program variables:
$P_{r}\left(  \emptyset,x\right)  \in X_{\emptyset}.\vspace*{-0.04in}$

\item[\textbf{--}] Multiple edges between nodes enable modeling of logical
"or" or "xor" type conditional transitions. This allows for modeling systems
with nondeterministic discrete transitions.$\vspace*{-0.04in}$

\item[\textbf{--}] The transition label $\overline{T}_{ji}^{k}$\ may represent
a simple update rule which depends on the real-time input. For instance, if
$T=Ax+Bw,$\ and $S=R^{n}\times\left[  -1,1\right]  ,$\ then $x\overset
{\overline{T}}{\mapsto}\left\{  Ax+Bw~|~w\in\left[  -1,1\right]  \right\}
.$\ In other cases, $\overline{T}_{ji}^{k}$\ may represent an abstraction of a
nonlinearity. For instance, the assignment $x\mapsto\sin\left(  x\right)
$\ can be abstracted by $x\overset{\overline{T}}{\mapsto}\left\{  T\left(
x,w\right)  |\left(  x,w\right)  \in S\right\}  $ (see Eqn. (\ref{SinAbst}) in
Appendix I).$\vspace*{-0.03in}$

\item[\textbf{--}] Graph models with state-dependent or time-varying
transitions labels arise, for instance, in modeling computer programs that
involve arithmetic operations with array elements. For instance, consider the
case where the transition functions are parametrized, and the parameters are
drawn from a finite set, e.g., a multidimensional array. The dependence on the
array elements can be random, resulting time-varying transition labels, or it
can be in the form of a functional dependence on other program variables.
Systematic treatment of such models is discussed in Appendix I.
\end{enumerate}

Before we proceed, we introduce the following notation:\ Given a semialgebraic
set $\Pi,$ and a polynomial function $\tau:\mathbb{R}^{n}\mapsto\mathbb{R}%
^{n},$ we denote by $\Pi\left(  \tau\right)  ,$ the set: $\Pi(\tau)=\left\{
x~|~\tau\left(  x\right)  \in\Pi\right\}  .\vspace*{-0.05in}$

\paragraph{Construction of Simple Invariant Sets\label{sec:constabst}}

\vspace*{-0.01in}Simple invariant sets can be included in the model if they
are readily available or easily computable.\vspace*{-0.01in} Even trivial
invariants can simplify the analysis and improve the chances of finding
stronger invariants via convex relaxations, e.g., the $\mathcal{S}$-Procedure
(cf. Section \ref{Chapter:Computation}).\vspace*{-0.05in}

\begin{enumerate}
\item[\textbf{--}] Simple invariant sets may be provided by the programmer.
These can be trivial sets representing simple algebraic relations between
variables, or they can be more complicated relationships\vspace*{-0.01in} that
reflect the programmer's knowledge about the functionality and behavior of the
program.\vspace*{-0.01in}

\item[\textbf{--}] Invariant Propagation: Assuming that $T_{ij}^{k}$ are
deterministic and invertible, the set\vspace*{-0.2in}%
\begin{equation}
X_{i}=%
{\textstyle\bigcup\limits_{j\in\mathcal{I}\left(  i\right)  ,\text{ }%
k\in\mathcal{A}_{ij}}}
\Pi_{ij}^{k}\left(  \left(  T_{ij}^{k}\right)  ^{-1}\right)  \vspace
*{-0.15in}\label{constraint prop}%
\end{equation}
is an invariant set for node $i.$ Furthermore, if the invariant sets $X_{j}$
are strict subsets of $\Omega^{n}$\ for all $j\in\mathcal{I}\left(  i\right)
,$ then (\ref{constraint prop})\ can be improved. Specifically, the
set\vspace*{-0.2in}%
\begin{equation}
X_{i}=%
{\textstyle\bigcup\limits_{j\in\mathcal{I}\left(  i\right)  ,\text{ }%
k\in\mathcal{A}_{ij}}}
\Pi_{ij}^{k}\left(  \left(  T_{ij}^{k}\right)  ^{-1}\right)  \cap X_{j}\left(
\left(  T_{ij}^{k}\right)  ^{-1}\right)  \vspace*{-0.1in}%
\label{invariant prop}%
\end{equation}
is an invariant\vspace*{-0.01in} set for node $i.$ Note that it is sufficient
that the restriction of $T_{ij}^{k}$ to the lower dimensional spaces in the
domains of $\Pi_{ij}^{k}$ and $X_{j}$ be invertible.

\item[\textbf{--}] Preserving Equality Constraints:\ Simple assignments of the
form $T_{ij}^{k}:x_{l}\mapsto f\left(  y_{m}\right)  $ result in invariant
sets of the form $X_{i}=\left\{  x~|~x_{l}-f\left(  y_{m}\right)  =0\right\}
$ at node $i,$ provided that $T_{ij}^{k}$ does not simultaneously update
$y_{m}.$ Formally, let $T_{ij}^{k}$ be such that $(T_{ij}^{k}x)_{l}-x_{l}$ is
non-zero for at most one element $\hat{l}\in\mathbb{Z}\left(  1,n\right)  ,$
and that $(T_{ij}^{k}x)_{\hat{l}}$ is independent of $x_{\hat{l}}.$ Then, the
following set is an invariant set at node $i:$\vspace*{-0.15in}
\[
X_{i}=%
{\textstyle\bigcup\limits_{j\in\mathcal{I}\left(  i\right)  ,\text{ }%
k\in\mathcal{A}_{ij}}}
\left\{  x~|~\left[  T_{ij}^{k}-I\right]  x=0\right\}  \vspace*{-0.15in}%
\]

\end{enumerate}

\subsubsection{Mixed-Integer Linear over Graph Hybrid Model
(MIL-GHM)\label{MGHM: section}}

The MIL-GHMs are graph models in which the effects of several lines and/or
\textit{functions} of code are compactly represented via a MILM. As a result,
the graphs in such models have edges (possibly self-edges) that are labeled
with matrices $F$ and $H$ corresponding to a MILM as the transition and
passport labels. Such models combine the flexibility provided by graph models
and the compactness of MILMs. An example is presented in Section
\ref{sec:casestudy}.\vspace*{-0.18in}

\subsection{Specifications\label{Section:Specifications}\vspace*{-0.08in}}

The specification that can be verified in our framework can generically be
described as unreachability and finite-time termination.\vspace*{-0.05in}
\newpage
\begin{definition}
\label{unreachability}A Program $\mathcal{P}\equiv\mathcal{S}(X,f,X_{0}%
,X_{\infty})$ is said to satisfy the unreachability property with respect to a
subset $X_{-}\subset X,$ if for every trajectory $\mathcal{X}\equiv x\left(
\cdot\right)  $ of (\ref{Softa1}), and every $t\in\mathbb{Z}_{+},$ $x(t)$ does
not belong to $X_{-}.$ A program $\mathcal{P}\equiv\mathcal{S}(X,f,X_{0}%
,X_{\infty})$ is said to \textit{terminate in finite time} if every solution
$\mathcal{X}=x\left(  \cdot\right)  $ of (\ref{Softa1}) satisfies $x(t)\in
X_{\infty}$ for some $t\in\mathbb{Z}_{+}.$\vspace*{-0.05in}
\end{definition}

Several critical specifications associated with runtime errors are special
cases of unreachability.

\subsubsection{Overflow}

Absence of overflow can be characterized as a special case of unreachability
by defining:\vspace*{-0.2in}%
\[
X_{-}=\left\{  x\in X~|~\left\Vert \alpha^{-1}x\right\Vert _{\infty}%
>1,~\alpha=\operatorname{diag}\left\{  \alpha_{i}\right\}  \right\}
\vspace*{-0.2in}%
\]
%
where $\alpha_{i}>0$ is the overflow limit for variable $i.$ We say that
$\alpha\succ0$ specifies the overflow limit.$\vspace*{-0.05in}$

\subsubsection{Out-of-Bounds Array Indexing}

An out-of-bounds array indexing error occurs when a variable exceeding the
length of an array, references an element of the array. Assuming that $x_{l}$
is the corresponding integer index and $L$ is the array length, one must
verify that $x_{l}$ does not exceed $L$ at location $i,$ where referencing
occurs. This can be accomplished by defining $X_{-}=\left\{  \left(
i,x\right)  \in X~|~\left\vert x_{l}\right\vert >L\right\}  $ over a graph
model and proving that $X_{-}$ is unreachable. This is also similar to
\textquotedblleft assertion checking\textquotedblright\ defined next.\vspace
*{-0.05in}

\subsubsection{Program Assertions}

An \textit{assertion} is a mathematical expression whose validity at a
specific location in the code must be verified. It usually indicates the
programmer's expectation from the behavior of the program. We consider
\textit{assertions} that are in the form of semialgebraic set memberships.
Using graph models, this is done as follows:\vspace*{-0.1in}%
\[%
\begin{array}
[c]{llccl}%
\text{at location }i: & \text{assert }x\in A_{i} & \Rightarrow & \text{define}
& X_{-}=\left\{  \left(  i,x\right)  \in X~|~x\in X\backslash A_{i}\right\}
,\\[-0.05in]%
\text{at location }i: & \text{assert }x\notin A_{i} & \Rightarrow &
\text{define} & X_{-}=\left\{  \left(  i,x\right)  \in X~|~x\in A_{i}\right\}
.
\end{array}
\vspace*{-0.1in}%
\]
In particular, safety assertions for division-by-zero or taking the square
root (or logarithm) of positive variables are standard and must be
automatically included in numerical programs (cf. Sec. {\small \ref{Sec:BC}},
Table {\small \ref{Table II}}).\vspace*{-0.03in}

\subsubsection{Program Invariants}

A program invariant is a property that holds throughout the execution of the
program. The property indicates that the variables reside in a semialgebraic
subset $X_{I}\subset X$. Essentially, any method that is used for verifying
unreachability of a subset $X_{-}\subset X,$ can be applied for verifying
invariance of $X_{I}$ by defining $X_{-}=X\backslash X_{I},$ and vice
versa.\vspace*{-0.15in}

\subsection{Implications of the Abstractions\label{Section:ImpAbst}}

For mathematical correctness, we must show that if an $\mathcal{A}%
$-representation of a program satisfies the unreachability and FTT
specifications, then so does the $\mathcal{C}$-representation, i.e., the
actual program. This is established\vspace*{-0.01in} in the following
proposition.\vspace*{-0.03in}

\begin{proposition}
\label{Abstraction}Let $\overline{\mathcal{S}}(\overline{X},\overline
{f},\overline{X}_{0},\overline{X}_{\infty})$ be an $\mathcal{A}$%
-representation of program $\mathcal{P}$ with $\mathcal{C}$-representation
\newline$\mathcal{S}(X,f,X_{0},X_{\infty}).$ Let $X_{-}\subset X$ and
$\overline{X}_{-}\subset\overline{X}$ be such that $X_{-}\subseteq\overline
{X}_{-}.$ Assume that the unreachability property w.r.t. $\overline{X}_{-}$
has been verified for $\overline{\mathcal{S}}$. Then, $\mathcal{P}$ satisfies
the unreachability property w.r.t. $X_{-}.$ Moreover, if the FTT property
holds for $\overline{\mathcal{S}}$, then $\mathcal{P}$ terminates in finite time.
\end{proposition}

\begin{proof}
See Appendix II.
\end{proof}

Since we are not concerned with undecidability issues, and in light of
Proposition \ref{Abstraction}, we will not differentiate between abstract or
concrete representations in the remainder of this paper.

\section{Lyapunov Invariants as Behavior
Certificates\label{Chapter:LyapunovInvs}\vspace*{-0.05in}}

Analogous to a Lyapunov function, a Lyapunov invariant is a real-valued
function of the program variables satisfying a \textit{difference inequality}
along the execution trace.\vspace*{-0.05in}

\begin{definition}
\label{Def:LyapInv}A $(\theta,\mu)$\textit{-Lyapunov invariant} for
$\mathcal{S}(X,f,X_{0},X_{\infty})$ is a function $V:X\mapsto\mathbb{R}$ such
that\vspace*{-0.18in}%
\begin{equation}
V\left(  x_{+}\right)  -\theta V\left(  x\right)  \leq-\mu\text{\qquad}\forall
x\in X,\text{ }x_{+}\in f\left(  x\right)  :x\notin X_{\infty}.\vspace
*{-0.15in}\label{Softa2}%
\end{equation}
where $\left(  \theta,\mu\right)  \in\lbrack0,\infty)\times\lbrack0,\infty)$.
Thus, a Lyapunov invariant satisfies the \textit{difference inequality}
(\ref{Softa2}) along the trajectories of $\mathcal{S}$ until they reach a
terminal state $X_{\infty}$.
\end{definition}

It follows from Definition \ref{Def:LyapInv} that a Lyapunov invariant is not
necessarily nonnegative, or bounded from below, and in general it need not be
monotonically decreasing. While the zero level set of $V$ defines an invariant
set in the sense that $V\left(  x_{k}\right)  \leq0$ implies $V\left(
x_{k+l}\right)  \leq0$, for all $l\geq0,$ monotonicity depends on $\theta$ and
the initial condition.\ For instance, if $V\left(  x_{0}\right)  \leq0,$
$\forall x_{0}\in X_{0},$ then (\ref{Softa2}) implies that $V\left(  x\right)
\leq0$ along the trajectories of $\mathcal{S},$ however, $V\left(  x\right)  $
may not be monotonic if $\theta<1,$ though it will be monotonic for
$\theta\geq1.$ Furthermore, the level sets of a Lyapunov invariant need not be
bounded closed curves.

Proposition \ref{prop:MILMLyap}\ (to follow) formalizes the interpretation of
Definition \ref{Def:LyapInv} for the specific modeling languages. Natural
Lyapunov invariants for graph models are functions of the form$\vspace
*{-0.2in}$%
\begin{equation}
V\left(  \widetilde{x}\right)  \equiv V\left(  i,x\right)  =\sigma_{i}\left(
x\right)  ,\text{\quad}i\in N,\vspace*{-0.2in}\label{nodewiselyap}%
\end{equation}
which assign a polynomial Lyapunov function to every node $i\in\mathcal{N}$ on
the graph $G\left(  \mathcal{N},\mathcal{E}\right)  .\vspace*{-0.05in}$

\begin{proposition}
\label{prop:MILMLyap}Let $\mathcal{S}\left(  F,H,X_{0},n,n_{w},n_{v}\right)  $
and properly labeled graph $G\left(  \mathcal{N},\mathcal{E}\right)  $ be the
MIL and graph models for a computer program $\mathcal{P}.$ The function
$V:\left[  -1,1\right]  ^{n}\mapsto\mathbb{R}$ is a $\left(  \theta
,\mu\right)  $-Lyapunov invariant for $\mathcal{P}$ if it satisfies:\vspace
*{-0.2in}%
\[
V(Fx_{e})-\theta V\left(  x\right)  \leq-\mu,\text{\qquad}\forall\left(
x,x_{e}\right)  \in\left[  -1,1\right]  ^{n}\times\Xi,\vspace*{-0.2in}%
\]
where\vspace*{-0.2in}%
\[
\Xi=\{\left(  x,w,v,1\right)  ~|~H[%
\begin{array}
[c]{cccc}%
\hspace*{-0.04in}\vspace*{0.04in}x\hspace*{-0.02in} & \hspace*{-0.02in}%
w\hspace*{-0.02in} & \hspace*{-0.02in}v\hspace*{-0.02in} & \hspace
*{-0.02in}1\hspace*{-0.04in}%
\end{array}
]^{^{T}}=0,\text{ }\left(  w,v\right)  \in\left[  -1,1\right]  ^{n_{w}}%
\times\left\{  -1,1\right\}  ^{n_{v}}\}.\vspace*{-0.17in}%
\]
The function $V:\mathcal{N\times}\mathbb{R}^{n}\mapsto\mathbb{R},$ satisfying
(\ref{nodewiselyap}) is a $\left(  \theta,\mu\right)  $-Lyapunov invariant for
$\mathcal{P}$ if$\vspace*{-0.2in}$
\begin{equation}
\sigma_{j}(x_{+})-\theta\sigma_{i}\left(  x\right)  \leq-\mu,\text{ }%
\forall\left(  i,j,k\right)  \in\mathcal{E},\text{ }(x,x_{+})\in(X_{i}\cap
\Pi_{ji}^{k})\times\overline{T}_{ji}^{k}x.\vspace*{-0.2in}\label{arcwiselyap0}%
\end{equation}
Note that a generalization of (\ref{Softa2}) allows for $\theta$ and $\mu$ to
depend on the state $x,$ although simultaneous search for $\theta\left(
x\right)  $ and $V\left(  x\right)  $ leads to non-convex conditions, unless
the dependence of $\theta$ on $x$ is fixed a-priori. We allow for dependence
of $\theta$ on the discrete component of the state in the following
way:$\vspace*{-0.18in}$%
\begin{equation}
\sigma_{j}(x_{+})-\theta_{ji}^{k}\sigma_{i}\left(  x\right)  \leq-\mu
_{ji},\text{ }\forall\left(  i,j,k\right)  \in\mathcal{E},\text{ }(x,x_{+}%
)\in(X_{i}\cap\Pi_{ji}^{k})\times\overline{T}_{ji}^{k}x\vspace*{-0.18in}%
\label{arcwiselyap}%
\end{equation}

\end{proposition}

\subsection{Behavior Certificates\label{Sec:BC}\vspace*{-0.1in}}

\subsubsection{Finite-Time Termination (FTT)\ Certificates\vspace*{-0.1in}}

The following proposition is applicable to FTT analysis of both finite and
infinite\ state models.

\begin{proposition}
\label{FTT2}\textbf{Finite-Time Termination.} Consider a program
$\mathcal{P},$ and its dynamical system model $\mathcal{S}(X,f,X_{0}%
,X_{\infty})$. If there exists a $\left(  \theta,\mu\right)  $-Lyapunov
invariant $V:X\mapsto\mathbb{R},$ uniformly bounded on $X\backslash X_{\infty
},$ satisfying (\ref{Softa2}) and the following conditions\vspace*{-0.25in}%
\begin{align}
V\left(  x\right)   & \leq-\eta\leq0,\text{\qquad}\forall x\in X_{0}%
\label{Softa2a1}\\
\mu+\left(  \theta-1\right)  \left\Vert V\right\Vert _{\infty}  &
>0\label{Softa2a2}\\
\max\left(  \mu,\eta\right)   & >0\label{Softa2a3}%
\end{align}%
\[
\vspace*{-0.8in}%
\]
where $\left\Vert V\right\Vert _{\infty}=\sup\limits_{x\in X\backslash
X_{\infty}}V\left(  x\right)  <\infty,$ then $\mathcal{P}$ terminates in
finite time, and an upper-bound on the number of iterations is given
by\vspace*{-0.1in}%
\begin{equation}
T_{u}=\left\{
\begin{array}
[c]{lcl}%
\displaystyle\frac{\log\left(  \mu+\left(  \theta-1\right)  \left\Vert
V\right\Vert _{\infty}\right)  -\log\left(  \mu\right)  }{\log\theta} & , &
\theta\neq1,~\mu>0\vspace*{0.1in}\\
\displaystyle\frac{\log\left(  \left\Vert V\right\Vert _{\infty}\right)
-\log\left(  \eta\right)  }{\log\theta} & , & \theta\neq1,~\mu=0\\
\left\Vert V\right\Vert _{\infty}/\mu & , & \theta=1
\end{array}
\right. \label{Bnd on No. Itrn.}%
\end{equation}%
\[
\vspace*{-0.3in}%
\]

\end{proposition}

\begin{proof}
The proof is presented in Appendix II.
\end{proof}

When the state-space $X$ is finite, or when the Lyapunov invariant $V$ is only
a function of a subset of the variables that assume values in a finite set,
e.g., integer counters, it follows from Proposition \ref{FTT2} that $V$ being
a $\left(  \theta,\mu\right)  $-Lyapunov invariant for any $\theta\geq1$ and
$\mu>0$ is sufficient for certifying FTT, and uniform boundedness of $V$ need
not be established a-priori.

\begin{example}
Consider the {\small IntegerDivision} program presented in Example
\ref{IntegerDiv-Ex}. The function $V:X\mapsto\mathbb{R},$ defined according to
$V:(\mathrm{dd},\mathrm{dr},\mathrm{q},\mathrm{r})\mapsto\mathrm{r}$ is a
$\left(  1,\mathrm{dr}\right)  $-Lyapunov invariant for
{\small IntegerDivision}:{\small \ }at every step, $V$ decreases by
$\mathrm{dr}>0.$ Since $X$ is finite, the program {\small IntegerDivision}
terminates in finite time. This, however, only proves absence of infinite
loops. The program could terminate with an overflow.
\end{example}

\vspace*{0.1in}

\subsubsection{Separating Manifolds and Certificates of Boundedness}

Let $V$ be a Lyapunov invariant satisfying (\ref{Softa2}) with $\theta=1.$ The
level sets of $V,$ defined by $\mathcal{L}_{r}(V)\overset{\text{def}}{=}\{x\in
X:V(x)<r\},$ are invariant with respect to (\ref{Softa1}) in the sense that
$x(t+1)\in\mathcal{L}_{r}(V)$ whenever $x(t)\in\mathcal{L}_{r}(V)$. However,
for $r=0,$ the level sets $\mathcal{L}_{r}(V)$ remain invariant with respect
to (\ref{Softa1}) for any nonnegative $\theta.$ This is an important property
with the implication that $\theta=1$ (i.e., monotonicity) is not necessary for
establishing a separating manifold between the reachable set and the unsafe
regions of the state space (cf. Theorem \ref{BddNess}).

\begin{theorem}
\label{BddNess}\textbf{Lyapunov Invariants as Separating Manifolds.} Let
$\mathcal{V}$ denote the set of all $\left(  \theta,\mu\right)  $-Lyapunov
invariants satisfying (\ref{Softa2}) for\ program $\mathcal{P}\equiv
\mathcal{S}(X,f,X_{0},X_{\infty}).$ Let $I$ be the identity map, and for
$h\in\left\{  f,I\right\}  $ define\vspace*{-0.15in}%
\[
h^{-1}\left(  X_{-}\right)  =\left\{  x\in X~|~h\left(  x\right)  \cap
X_{-}\neq\varnothing\right\}  .\vspace*{-0.2in}%
\]
A subset $X_{-}\subset X,$ where $X_{-}\cap X_{0}=\varnothing$ can never be
reached along the trajectories of $\mathcal{P},$ if there exists
$V\in\mathcal{V}$ satisfying\vspace*{-0.1in}%
\begin{equation}
\underset{x\in X_{0}}{\sup}V(x)~<~\underset{x\in h^{-1}\left(  X_{-}\right)
}{\inf}V\left(  x\right)  \vspace*{-0.12in}\label{Inf G than Sup}%
\end{equation}
and either $\theta=1,$ or one of the following two conditions hold:\vspace
*{-0.13in}%
\begin{align}
\left(  \text{I}\right)  \text{ }\theta & <1\text{\qquad and\qquad}%
\underset{x\in h^{-1}\left(  X_{-}\right)  }{\inf}V(x)>0.\vspace
*{-0.12in}\label{Inf G than Zero}\\
\left(  \text{II}\right)  \text{ }\theta & >1\text{\hspace*{0.33in}%
and\qquad\ }\underset{x\in X_{0}}{\sup}V(x)\leq0.\vspace*{-0.12in}%
\label{Sup L than Zero}%
\end{align}

\end{theorem}

\medskip

\begin{proof}
The proof\ is presented in Appendix II.
\end{proof}

The following corollary\ is based on Theorem \ref{BddNess} and Proposition
\ref{FTT2}\ and presents computationally implementable criteria (cf. Section
\ref{Chapter:Computation}) for simultaneously establishing FTT and absence of overflow.

\begin{corollary}
\label{Bddness and FTT}\textbf{Overflow and FTT Analysis }Consider a program
$\mathcal{P},$ and its dynamical system model $\mathcal{S}(X,f,X_{0}%
,X_{\infty})$. Let $\alpha>0$ be a diagonal matrix specifying the overflow
limit$,$ and let $X_{-}=\{x\in X~|~\left\Vert \alpha^{-1}x\right\Vert
_{\infty}>1\}.$ Let $q\in\mathbb{N}\cup\left\{  \infty\right\}  ,$
$h\in\left\{  f,I\right\}  ,$ and let the function $V:X\mapsto\mathbb{R}$ be a
$\left(  \theta,\mu\right)  $-Lyapunov invariant for $\mathcal{S}$
satisfying\vspace*{-0.2in}%
\begin{align}
V\left(  x\right)   &  \leq0\text{\qquad\qquad\hspace{0.55in}\hspace*{0.03in}%
}\forall x\in X_{0}.\label{One1}\\[-0.1in]
V\left(  x\right)   &  \geq\sup\left\{  \left\Vert \alpha^{-1}h\left(
x\right)  \right\Vert _{q}-1\right\}  \text{\hspace{0.2in}}\forall x\in
X.\vspace*{-0.3in}\label{Three3}%
\end{align}%
\[
\vspace*{-0.7in}%
\]
\vspace*{-0.5in} \newline\noindent Then, an \textit{overflow runtime error}
will not occur during any execution of $\mathcal{P}.$ In addition, if $\mu>0$
and $\mu+\theta>1,$ then, $\mathcal{P}$ terminates in at most $T_{u}$
iterations where $T_{u}=\mu^{-1}$ if $\theta=1,$ and for $\theta\neq1$ we
have:\vspace*{-0.1in}%
\begin{equation}
T_{u}=\frac{\log~\left(  \mu+\left(  \theta-1\right)  \left\Vert V\right\Vert
_{\infty}\right)  -\log\mu}{\log\theta}\leq\frac{\log\left(  \mu
+\theta-1\right)  -\log\mu}{\log\theta}\vspace*{-0.1in}\label{upperboundonTU}%
\end{equation}
where $\left\Vert V\right\Vert _{\infty}=\sup\limits_{x\in X\backslash\left\{
X_{-}\cup X_{\infty}\right\}  }~\left\vert V\left(  x\right)  \right\vert .$%
\[
\vspace*{-0.2in}%
\]

\end{corollary}

\begin{proof}
The proof\ is presented in Appendix II.
\end{proof}

\textbf{Remarks}

\begin{enumerate}
\item[\textbf{--}] Application of Corollary \ref{Bddness and FTT} with $h=f$
typically leads to much less conservative results compared with $h=I$, though
the computational costs are also higher.

\item[\textbf{--}] An alternative, less conservative formulation is obtained
when (\ref{Three3}) is replaced by $n$ constraints in the following way:%
\begin{equation}
V\left(  x\right)  \geq\sup\left\{  \left\vert \alpha_{k}^{-1}h\left(
x\right)  _{k}\right\vert -1\right\}  \text{\qquad}\forall x\in X,\text{ }%
k\in\mathbb{Z}\left(  1,n\right) \label{remarktos}%
\end{equation}
where $\alpha_{k}$ is the overflow limit for scalar variable $x_{k}.$ Since
$\alpha_{k}^{-1}x_{k}$ is a scalar quantity, $\left\Vert \alpha_{k}%
^{-1}h\left(  x\right)  _{k}\right\Vert _{q}=\left\vert \alpha_{k}%
^{-1}h\left(  x\right)  _{k}\right\vert $ and (\ref{remarktos}) can be used in
conjunction with quadratic Lyapunov invariants and the $\mathcal{S}$-Procedure
for convex relaxation to formulate the problem of searching for $V $ as
semidefinite optimization problem (cf. Section \ref{Chapter:Computation}).
However, (\ref{remarktos}) is computationally more expensive than
(\ref{Three3}), as it increases the number of constraints.

\item[\textbf{--}] An even less conservative but computationally more
demanding alternative would be to construct a different Lyapunov invariant for
each variable and rewrite (\ref{remarktos}), (\ref{Three3}), (\ref{One1})
along with (\ref{Softa2}), for $n$ different functions $V_{k}\left(  .\right)
$. For instance:\vspace*{-0.1in}%
\begin{align*}
V_{k}\left(  x\right)   & \leq0,\text{\qquad}\forall x\in X_{0}.\\
V_{k}\left(  x\right)   & \geq\left\vert \alpha_{k}^{-1}h\left(  x\right)
_{k}\right\vert -1\text{\qquad}\forall x_{k}\in X.
\end{align*}

\end{enumerate}

\paragraph{General Unreachability and FTT Analysis over Graph Models}

The results presented so far in this section (Theorem \ref{BddNess}, Corollary
\ref{Bddness and FTT}, and Proposition \ref{FTT2}) are readily applicable to
MILMs. These results will be applied in Section \ref{Chapter:Computation} to
formulate the verification problem as a convex optimization problem. Herein,
we present an adaptation of these results to analysis of graph models.

\begin{definition}
A cycle $\mathcal{C}_{m}$ on a graph $G\left(  \mathcal{N},\mathcal{E}\right)
$ is an ordered list of $m$ triplets $\left(  n_{1},n_{2},k_{1}\right)  ,$
$\left(  n_{2},n_{3},k_{2}\right)  ,...,$\ $\left(  n_{m},n_{m+1}%
,k_{m}\right)  ,$ where $n_{m+1}=n_{1},$ and $\left(  n_{j},n_{j+1}%
,k_{j}\right)  \in\mathcal{E},$ $\forall j\in\mathbb{Z}\left(  1,m\right)  .$
A simple cycle is a cycle with no strict sub-cycles.
\end{definition}

\begin{corollary}
\label{SafetyGraphCor1}\textbf{Unreachability and FTT\ Analysis of Graph
Models}. Consider a program $\mathcal{P}$ and its graph model $G\left(
\mathcal{N},\mathcal{E}\right)  .$ Let $V\left(  i,x\right)  =\sigma
_{i}\left(  x\right)  $ be a Lyapunov invariant for $G\left(  \mathcal{N}%
,\mathcal{E}\right)  ,$ satisfying (\ref{arcwiselyap}) and$\vspace*{-0.2in}$%
\begin{equation}
\sigma_{\emptyset}\left(  x\right)  \leq0,\text{\qquad}\forall x\in
X_{\emptyset}\vspace*{-0.15in}\label{SGC1}%
\end{equation}
and either one of the following two conditions:$\vspace*{-0.1in}$%
\begin{align}
(\text{\textrm{I}})  & :\sigma_{i}\left(  x\right)  >0,\text{\qquad}\forall
x\in X_{i}\cap X_{i-},\text{ }i\in\mathcal{N}\backslash\left\{  \emptyset
\right\} \label{SGC2}\\
(\text{\textrm{II}})  & :\sigma_{i}\left(  x\right)  >0,\text{\qquad}\forall
x\in X_{j}\cap T^{-1}\left(  X_{i-}\right)  ,\text{ }i\in\mathcal{N}%
\backslash\left\{  \emptyset\right\}  ,\text{ }j\in\mathcal{I}\left(
i\right)  ,\text{ }T\in\left\{  \overline{T}_{ij}^{k}\right\} \label{SGC3}%
\end{align}
where\vspace*{-0.15in}%
\[
T^{-1}\left(  X_{i-}\right)  =\left\{  x\in X_{i}|T\left(  x\right)  \cap
X_{i-}\neq\varnothing\right\}
\]
\vspace*{-0.45in} \newline\noindent Then, $\mathcal{P}$ satisfies the
unreachability property w.r.t. the collection of sets $X_{i-},$ $i\in
\mathcal{N}\backslash\left\{  \emptyset\right\}  .$ In addition, if for every
simple cycle $\mathcal{C}\in G,$ we have:$\vspace*{-0.15in}$%
\begin{equation}
\left(  \theta\left(  \mathcal{C}\right)  -1\right)  \left\Vert \sigma\left(
\mathcal{C}\right)  \right\Vert _{\infty}+\mu\left(  \mathcal{C}\right)
>0,\text{ and }\mu\left(  \mathcal{C}\right)  >0,\text{ and }\left\Vert
\sigma\left(  \mathcal{C}\right)  \right\Vert _{\infty}<\infty,\vspace
*{-0.3in}\label{MultiplicativeTheta}%
\end{equation}
where$\vspace*{-0.05in}$
\begin{equation}
\theta\left(  \mathcal{C}\right)  =\underset{\left(  i,j,k\right)
\in\mathcal{C}}{%
{\textstyle\prod}
}\theta_{ij}^{k},\text{\qquad}\mu\left(  \mathcal{C}\right)  =\max_{\left(
i,j,k\right)  \in\mathcal{C}}\mu_{ij}^{k},\text{\qquad}\left\Vert
\sigma\left(  \mathcal{C}\right)  \right\Vert _{\infty}=\max\limits_{\left(
i,.,.\right)  \in\mathcal{C}}~\sup\limits_{x\in X_{i}\backslash X_{i-}%
}\left\vert \sigma_{i}\left(  x\right)  \right\vert \vspace*{-0.0in}%
\label{MultiplicativeThetanDefs}%
\end{equation}
then $\mathcal{P}$ terminates in at most $T_{u}$ iterations where\vspace
*{-0.05in}%
\[
T_{u}=\sum_{\mathcal{C}\in G:\theta\left(  \mathcal{C}\right)  \neq1}%
\frac{\log~\left(  \left(  \theta\left(  \mathcal{C}\right)  -1\right)
\left\Vert \sigma\left(  \mathcal{C}\right)  \right\Vert _{\infty}+\mu\left(
\mathcal{C}\right)  \right)  -\log~\mu\left(  \mathcal{C}\right)  }{\log
\theta\left(  \mathcal{C}\right)  }+\sum_{\mathcal{C}\in G:\theta\left(
\mathcal{C}\right)  =1}\frac{\left\Vert \sigma\left(  \mathcal{C}\right)
\right\Vert _{\infty}}{\mu\left(  \mathcal{C}\right)  }.
\]

\end{corollary}

\bigskip

\begin{proof}
The proof\ is presented in Appendix II.
\end{proof}

For verification against an overflow violation specified by a diagonal matrix
$\alpha>0,$ Corollary \ref{SafetyGraphCor1} is applied with $X_{-}%
=\{x\in\mathbb{R}^{n}~|~\left\Vert \alpha^{-1}x\right\Vert _{\infty}>1\}.$
Hence, (\ref{SGC2}) becomes $\sigma_{i}\left(  x\right)  \geq p\left(
x\right)  (\left\Vert \alpha^{-1}x\right\Vert _{q}-1),$\quad$\forall x\in
X_{i},$ $i\in\mathcal{N}\backslash\left\{  \emptyset\right\}  ,$ where
$p\left(  x\right)  >0$. User-specified assertions, as well as many other
standard safety specifications, such as absence of division-by-zero can be
verified using Corollary \ref{SafetyGraphCor1} (See Table I).

\paragraph*{-- Identification of Dead Code}

Suppose that we wish to verify that a discrete location $i\in\mathcal{N}%
\backslash\left\{  \emptyset\right\}  $ in a graph model $G\left(
\mathcal{N},\mathcal{E}\right)  $ is unreachable. If a function satisfying the
criteria of Corollary \ref{SafetyGraphCor1} with $X_{i-}=\mathbb{R}^{n}$ can
be found, then location $i$ can never be reached. Condition (\ref{SGC2}) then
becomes $\sigma_{i}\left(  x\right)  \geq0$,$~\forall x\in\mathbb{R}^{n}.$
{\small \begin{table}[tbh]
\caption{Application of Corollary \ref{SafetyGraphCor1} to the verification of
various safety specifications.}%
\label{Table II}%
{\small \vspace*{-0.35in}  }
\par
\begin{center}
{\small $%
\begin{tabular}
[c]{clcl}
&  &  & \\
&  &  & \\\cline{4-4}
&  &  & \multicolumn{1}{|l|}{apply Corollary \ref{SafetyGraphCor1}
with:}\\\hline
\multicolumn{1}{|c}{At location $i$:} & \multicolumn{1}{|l}{assert $x\in
X_{a}$} & \multicolumn{1}{|c}{$%
\begin{array}
[c]{c}%
\vspace*{-0.12in}\\
\Rightarrow\\
\vspace*{-0.12in}%
\end{array}
$} & \multicolumn{1}{|l|}{$X_{i-}:=\left\{  x\in\mathbb{R}^{n}~|~x\in
\mathbb{R}^{n}\backslash X_{a}\right\}  $}\\\hline
\multicolumn{1}{|c}{At location $i$:} & \multicolumn{1}{|l}{assert $x\notin
X_{a}$} & \multicolumn{1}{|c}{$%
\begin{array}
[c]{c}%
\vspace*{-0.12in}\\
\Rightarrow\\
\vspace*{-0.12in}%
\end{array}
$} & \multicolumn{1}{|l|}{$X_{i-}:=\left\{  x\in\mathbb{R}^{n}~|~x\in
X_{a}\right\}  $}\\\hline
\multicolumn{1}{|c}{At location $i$:} & \multicolumn{1}{|l}{(expr.)/$x_{o}$} &
\multicolumn{1}{|c}{$%
\begin{array}
[c]{c}%
\vspace*{-0.12in}\\
\Rightarrow\\
\vspace*{-0.12in}%
\end{array}
$} & \multicolumn{1}{|l|}{$X_{i-}:=\left\{  x\in\mathbb{R}^{n}~|~x_{o}%
=0\right\}  $}\\\hline
\multicolumn{1}{|c}{At location $i$:} & \multicolumn{1}{|l}{$\sqrt[2k]{x_{o}}
$} & \multicolumn{1}{|c}{$%
\begin{array}
[c]{c}%
\vspace*{-0.12in}\\
\Rightarrow\\
\vspace*{-0.12in}%
\end{array}
$} & \multicolumn{1}{|l|}{$X_{i-}:=\left\{  x\in\mathbb{R}^{n}~|~x_{o}%
<0\right\}  $}\\\hline
\multicolumn{1}{|c}{At location $i$:} & \multicolumn{1}{|l}{$\log\left(
x_{o}\right)  $} & \multicolumn{1}{|c}{$%
\begin{array}
[c]{c}%
\vspace*{-0.12in}\\
\Rightarrow\\
\vspace*{-0.12in}%
\end{array}
$} & \multicolumn{1}{|l|}{$X_{i-}:=\left\{  x\in\mathbb{R}^{n}~|~x_{o}%
\leq0\right\}  $}\\\hline
\multicolumn{1}{|c}{At location $i$:} & \multicolumn{1}{|l}{dead code} &
\multicolumn{1}{|c}{$%
\begin{array}
[c]{c}%
\vspace*{-0.12in}\\
\Rightarrow\\
\vspace*{-0.12in}%
\end{array}
$} & \multicolumn{1}{|l|}{$X_{i-}:=R^{n}$}\\\hline
\end{tabular}
$  }
\end{center}
\end{table}}

\vspace*{-0.2in}

\begin{example}
Consider the following program\vspace*{-0.05in} {\small
\begin{gather*}%
\begin{array}
[b]{ccc}%
\begin{tabular}
[c]{|l|}\hline
$%
\begin{array}
[c]{ll}
& \mathrm{void~ComputeTurnRate~(void)}\\
& \vspace*{-0.45in}\\
\mathrm{L}0: & \mathrm{\{double~x=\{0\};~double~y=\{\ast PtrToY\};}\\
& \vspace*{-0.45in}\\
\mathrm{L}1: & \mathrm{while~(1)}\\
& \vspace*{-0.45in}\\
\mathrm{L}2: & \mathrm{\{}\hspace{0.18in}~\mathrm{y=\ast PtrToY;}\\
& \vspace*{-0.45in}\\
\mathrm{L}3: & \qquad\mathrm{x=(5\ast sin(y)+1)/3;}\\
& \vspace*{-0.45in}\\
\mathrm{L}4: & \qquad\mathrm{if~x>-1~\{}\\
& \vspace*{-0.45in}\\
\mathrm{L}5: & \qquad\qquad~\mathrm{x=x+1.0472;}\\
& \vspace*{-0.45in}\\
\mathrm{L}6: & \qquad\qquad~\mathrm{TurnRate=y/x;\}}\\
& \vspace*{-0.45in}\\
\mathrm{L}7: & \qquad\mathrm{else~\{}\\
& \vspace*{-0.45in}\\
\mathrm{L}8: & \qquad\qquad~\mathrm{TurnRate=100\ast y/3.1416~\}\}}%
\end{array}
$\\\hline
\end{tabular}
&  &
\raisebox{-1.1341in}{\includegraphics[
height=2.21in,
width=1.92in
]%
{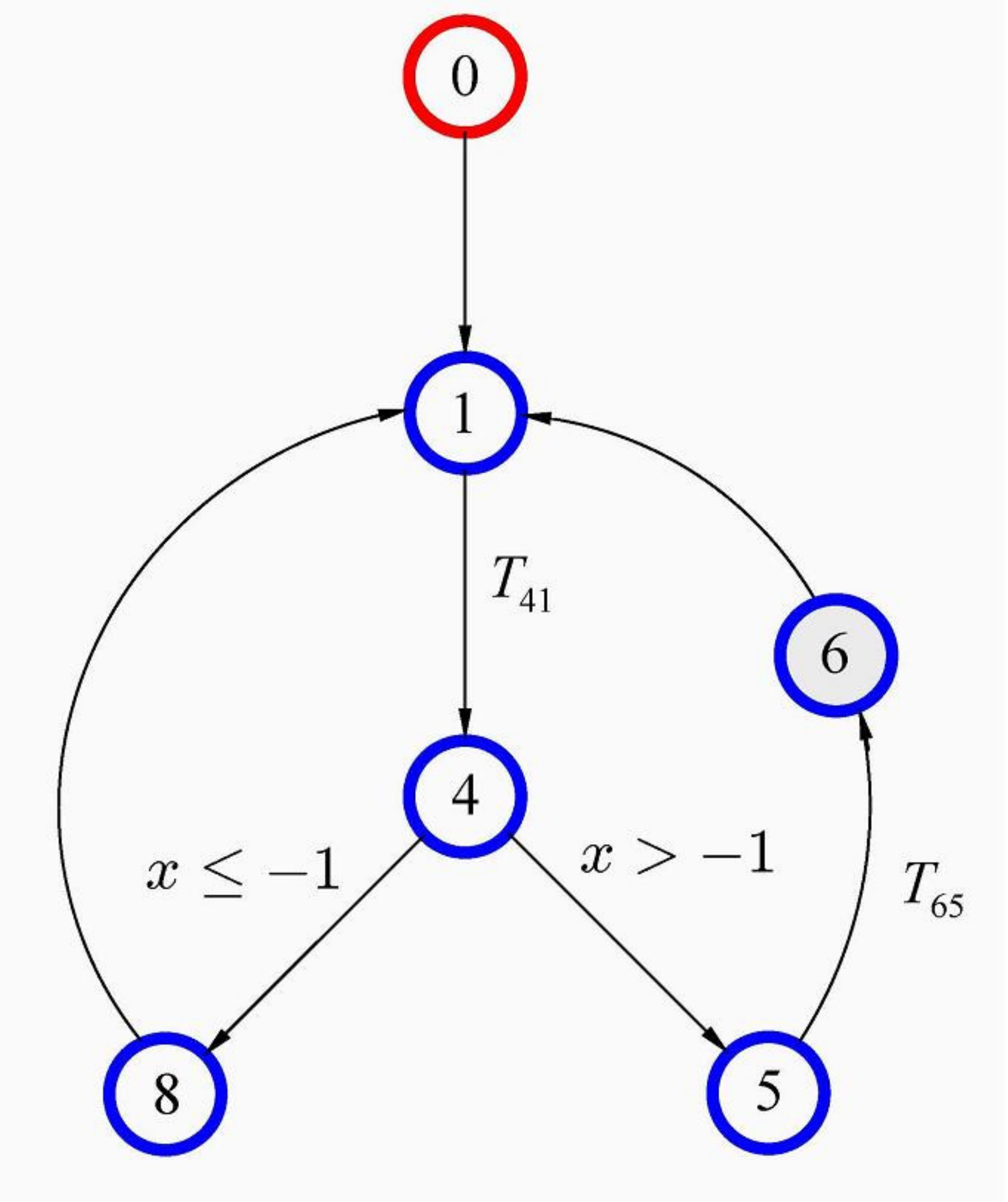}%
}%
\end{array}
\\
\text{Program 3\qquad\hspace{0.5in}\qquad\qquad Graph of an abstraction of
Program 3\hspace*{-1.52in}}%
\end{gather*}
}
Note that $\mathrm{x}$ can be zero right after the assignment $\mathrm{x}%
=(5\sin(\mathrm{y})+1)/3.$ However, at location \textrm{L}$6$, $\mathrm{x}$
cannot be zero and division-by-zero will not occur. The graph model of an
abstraction of Program 3 is shown next to the program and is defined by the
following elements: $T_{65}:\mathrm{x}\mapsto\mathrm{x}+1.0472,$ and
$T_{41}:\mathrm{x}\mapsto\left[  -4/3,2\right]  .$ The other transition labels
are identity. The only non-universal passport labels are $\Pi_{54}$ and
$\Pi_{84}$ as shown in the figure. Define\vspace*{-0.2in}%
\begin{align*}
\sigma_{6}\left(  \mathrm{x}\right)   & =-\mathrm{x}^{2}-100\mathrm{x}%
+1,\text{ }\sigma_{5}\left(  \mathrm{x}\right)  =-(\mathrm{x}+1309/1250)^{2}%
-100\mathrm{x}-2543/25\\[-0.1in]
\sigma_{0}\left(  \mathrm{x}\right)   & =\sigma_{1}\left(  \mathrm{x}\right)
=\sigma_{4}\left(  \mathrm{x}\right)  =\sigma_{8}\left(  \mathrm{x}\right)
=-\mathrm{x}^{2}+2\mathrm{x}-3.\vspace*{-0.1in}%
\end{align*}
\vspace*{-0.55in} \newline\noindent It can be verified that $V\left(
\mathrm{x}\right)  =\sigma_{i}\left(  \mathrm{x}\right)  $ is a $\left(
\theta,1\right)  $-Lyapunov invariant for Program 3 with variable rates:
$\theta_{65}=1,$ and $\theta_{ij}=0$ for all $\left(  i,j\right)  \neq\left(
6,5\right)  $. Since\vspace*{-0.2in}%
\[
-2=\sup_{\mathrm{x}\in X_{0}}\sigma_{0}\left(  \mathrm{x}\right)
<\inf_{\mathrm{x}\in X_{-}}\sigma_{6}\left(  \mathrm{x}\right)  =1\vspace
*{-0.1in}%
\]
the state $\left(  6,\mathrm{x}=0\right)  $ cannot be reached. Hence, a
division by zero will never occur. We will show in the next section how to
find such functions in general.$\vspace*{-0.1in}$
\end{example}

\section{Computation of Lyapunov Invariants \label{Chapter:Computation}%
$\vspace*{-0.05in}$}

It is well known that the main difficulty in using Lyapunov functions in
system analysis is finding them. Naturally, using Lyapunov invariants in
software analysis inherits the same difficulties. However, the recent advances
in hardware and software technology, e.g., semi-definite programming
\cite{GahinetLMILAB}, \cite{Strum1999}, and linear programming software
\cite{ILOG} present an opportunity for new approaches to software verification
based on numerical optimization.\vspace*{-0.25in}

\subsection{Preliminaries\label{Section:Preliminaries}\vspace*{-0.1in}}

\subsubsection{Convex Parameterization of Lyapunov Invariants}

The chances of finding a Lyapunov invariant are increased when (\ref{Softa2})
is only required on a subset of $X\backslash X_{\infty}$. For instance, for
$\theta\leq1,$ it is tempting to replace (\ref{Softa2}) with\vspace*{-0.2in}%
\begin{equation}
V\left(  x_{+}\right)  -\theta V\left(  x\right)  \leq-\mu,\text{ }\forall
x\in X\backslash X_{\infty}:V\left(  x\right)  <1,~x_{+}\in f\left(  x\right)
\vspace*{-0.15in}\label{Softa4}%
\end{equation}
In this formulation $V$ is not required to satisfy (\ref{Softa2}) for those
states which cannot be reached from $X_{0}.$ However, the set of all functions
$V:X\mapsto\mathbb{R}$ satisfying (\ref{Softa4}) is not convex and finding a
solution for (\ref{Softa4}) is typically much harder than (\ref{Softa2}). Such
non-convex formulations are not considered in this paper.

The first step in the search for a function $V:X\mapsto\mathbb{R}$ satisfying
(\ref{Softa2}) is selecting a finite-dimensional linear parameterization of a
candidate function $V$:\vspace*{-0.25in}%
\begin{equation}
V\left(  x\right)  =V_{\tau}\left(  x\right)  =\sum_{k=1}^{n}\tau_{k}%
V_{k}\left(  x\right)  ,\qquad\tau=\left(  \tau_{k}\right)  _{k=1}^{n},\text{
}\tau_{k}\in\mathbb{R},\vspace*{-0.15in}\label{Softa7}%
\end{equation}
where $V_{k}:X\mapsto\mathbb{R}$ are fixed basis functions. Next, for every
$\tau=(\tau_{k})_{k=1}^{N}$ let\vspace*{-0.15in}
\[
\phi(\tau)=\max_{x\in X\backslash X_{\infty},~x_{+}\in f(x)}V_{\tau}%
(x_{+})-\theta V_{\tau}(x),\vspace*{-0.1in}%
\]
(assuming for simplicity that the maximum does exist). Since $\phi\left(
\cdot\right)  $ is a maximum over $x$ of a family of linear functions in
$\tau$, $\phi\left(  \cdot\right)  $ is a convex function. If minimizing
$\phi\left(  \cdot\right)  $ over the unit disk yields a negative minimum, the
optimal $\tau^{\ast}$ defines a valid Lyapunov invariant $V_{\tau^{\ast}}(x)$.
Otherwise, no linear combination (\ref{Softa7}) yields a valid solution for
(\ref{Softa2}).

The success\vspace*{-0.01in} and efficiency of the proposed approach depend on
computability of $\phi\left(  \cdot\right)  $ and its subgradients.\vspace
*{-0.01in} While $\phi\left(  \cdot\right)  $ is convex, the same does not
necessarily hold for $V_{\tau}(x_{+})-\theta V_{\tau}(x)$ as a function of
$x$. In\vspace*{-0.01in} fact, if $X\backslash X_{\infty}$ is non-convex,
which is often the case even for very simple\vspace*{-0.01in} programs,
computation of $\phi\left(  \cdot\right)  $ becomes a non-convex optimization
problem\vspace*{-0.01in} even if $V_{\tau}(x_{+})-V_{\tau}(x)$ is a nice (e.g.
linear or concave and smooth) function\vspace*{-0.01in} of $x.$ To get around
this hurdle, we propose using convex relaxation techniques which essentially
lead to computation of a convex upper bound for $\phi\left(  \tau\right)
$.\vspace*{-0.01in}

\subsubsection{Convex Relaxation Techniques\vspace*{-0.01in}}

Such techniques constitute\vspace*{-0.01in} a broad class of techniques for
constructing finite-dimensional, convex approximations\vspace*{-0.01in} for
difficult non-convex optimization problems. Some of the results most relevant
to the\vspace*{-0.01in} software verification framework presented in this
paper can be found in \cite{Lovasz1991}\ for SDP\vspace*{-0.01in} relaxation
of binary integer programs, \cite{Meg01} and \cite{Nesterov 2000} for SDP
relaxation of quadratic programs, \cite{Yakubovic}\ for $\mathcal{S}%
$-Procedure in robustness analysis,\vspace*{-0.01in} and \cite{ParriloThesis}%
,\cite{Parrilo2001} for sum-of-squares relaxation in polynomial non-negativity
verification. \vspace*{-0.01in}We provide a brief overview of the latter two techniques.

\paragraph{The $\mathcal{S}$-Procedure \label{Section:S-Procedure}}

The $\mathcal{S}$-Procedure is commonly used for construction of Lyapunov
functions for nonlinear dynamical systems. Let functions $\phi_{i}%
:X\mapsto\mathbb{R},$ $i\in\mathbb{Z}\left(  0,m\right)  ,$ and $\psi
_{j}:X\mapsto\mathbb{R},$ $j\in\mathbb{Z}\left(  1,n\right)  $ be given, and
suppose that we are concerned with evaluating the following assertions:\vspace
*{-0.2in}%
\begin{gather}
\text{(I)}\text{: }\phi_{0}\left(  x\right)  >0,\text{ }\forall x\in\left\{
x\in X~|~\phi_{i}\left(  x\right)  \geq0,\text{ }\psi_{j}\left(  x\right)
=0,\text{ }i\in\mathbb{Z}\left(  1,m\right)  ,\text{ }j\in\mathbb{Z}\left(
1,n\right)  \right\} \label{Needs-S-Procedure}\\
\text{(II)}\text{: }\exists\tau_{i}\in\mathbb{R}^{+},\text{ }\exists\mu_{j}%
\in\mathbb{R},\text{ such that }\phi_{0}\left(  x\right)  >\sum_{i=1}^{m}%
\tau_{i}\phi_{i}\left(  x\right)  +\sum_{j=1}^{n}\mu_{j}\psi_{j}\left(
x\right)  .\vspace*{-0.25in}\label{S-Procedure-sufficient}%
\end{gather}

\noindent The implication (II)\ $\rightarrow$ (I) is trivial. The process of
replacing assertion (I) by its \textit{relaxed} version (II) is called the
$\mathcal{S}$-Procedure. Note that condition (II) is convex in decision
variables $\tau_{i}$ and $\mu_{j}.$ The implication (I) $\rightarrow$ (II) is
generally not true and the $\mathcal{S}$-Procedure is called lossless for
special cases where (I) and (II) are equivalent. A well-known such case is
when $m=1,$ $n=0,$ and $\phi_{0},$ $\phi_{1}$ are quadratic functionals$.$ A
comprehensive discussion of the $\mathcal{S}$-Procedure as well as available
results on its losslessness can be found in \cite{Gusev06}. Other variations
of $\mathcal{S}$-Procedure with non-strict inequalities exist as well.

\paragraph{Sum-of-Squares (SOS) Relaxation \label{Section:SOS-Relaxation}}

The SOS\ relaxation technique can be interpreted as the generalized version of
the $\mathcal{S}$-Procedure and is concerned with verification of the
following assertion:\vspace*{-0.15in}%
\begin{equation}
f_{j}\left(  x\right)  \geq0,\text{ }\forall j\in J,\text{\qquad}g_{k}\left(
x\right)  \neq0,\text{ }\forall k\in K,\text{\qquad}h_{l}\left(  x\right)
=0,\text{ }\forall l\in L\Rightarrow-f_{0}\left(  x\right)  \geq
0,\vspace*{-0.15in}\label{SemiAlgFormulae}%
\end{equation}
where $f_{j},g_{k},h_{l}$ are polynomial functions. It is easy to see that the
problem is equivalent to verification of emptiness of a semialgebraic set, a
necessary and sufficient condition for which is given by the
Positivstellensatz Theorem \cite{Boshnack}. In practice, sufficient conditions
in the form of nonnegativity of polynomials are formulated. The non-negativity
conditions are in turn relaxed to SOS conditions. Let $\Sigma\left[
y_{1},\ldots,y_{m}\right]  $ denote the set of SOS polynomials in $m$
variables $y_{1},...,y_{m}$, i.e. the set of polynomials that can be
represented as $p=%
{\textstyle\sum\limits_{i=1}^{t}}
p_{i}^{2},$ $p_{i}\in\mathbb{P}_{m}, $ where $\mathbb{P}_{m}$ is the
polynomial ring of $m$ variables with real coefficients. Then, a sufficient
condition for (\ref{SemiAlgFormulae}) is that there exist SOS\ polynomials
$\tau_{0},\tau_{i},\tau_{ij}\in\Sigma\left[  x\right]  $ and polynomials
$\rho_{l},$ such that\vspace*{-0.1in}%
\[
\tau_{0}+\sum\nolimits_{i}\tau_{i}f_{i}+\sum\nolimits_{i,j}\tau_{ij}f_{i}%
f_{j}+\sum\nolimits_{l}\rho_{l}h_{l}+(\prod g_{k})^{2}=0\vspace*{-0.1in}%
\]
Matlab toolboxes SOSTOOLS \cite{Prajna}, or YALMIP \cite{Lofberg2004} automate
the process of converting an SOS problem to an SDP, which is subsequently
solved by available software packages such as LMILAB \cite{GahinetLMILAB}, or
SeDumi \cite{Strum1999}. Interested readers are referred to \cite{Parrilo2001,
Megretski2003, ParriloThesis, Prajna} for more details.\vspace*{-0.15in}

\subsection{Optimization of Lyapunov Invariants for Mixed-Integer Linear
Models\vspace*{-0.1in}}

Natural Lyapunov invariant candidates for MILMs are quadratic and affine functionals.

\subsubsection{Quadratic Invariants\label{Section:MILM-QI}}

The linear parameterization of the space of quadratic functionals mapping
$\mathbb{R}^{n}$ to $\mathbb{R}$ is given by:\vspace*{-0.1in}{\small
\begin{equation}
\mathcal{V}_{x}^{2}=\left\{  V:\mathbb{R}^{n}\mapsto\mathbb{R~}|~V(x)=\left[
\begin{array}
[c]{c}%
\vspace*{-0.48in}\\
x\vspace*{-0.12in}\\
1\\
\vspace*{-0.43in}%
\end{array}
\right]  ^{T}P\left[
\begin{array}
[c]{c}%
\vspace*{-0.48in}\\
x\vspace*{-0.12in}\\
1\\
\vspace*{-0.43in}%
\end{array}
\right]  \text{, }P\in\mathbb{S}^{n+1}\right\}  ,\vspace*{-0.15in}%
\label{MILP_Quad_fctn}%
\end{equation}
}
where $\mathbb{S}^{n}$ is the set of $n$-by-$n$ symmetric matrices. We have
the following lemma.\vspace{-0.05in}

\begin{lemma}
\label{MIPL_Invariance_Lemma}Consider a program $\mathcal{P}$ and its MILM
$\mathcal{S}\left(  F,H,X_{0},n,n_{w},n_{v}\right)  .$ The program admits a
quadratic $\left(  \theta,\mu\right)  $-Lyapunov invariant $V\in
\mathcal{V}_{x}^{2},$ if there exists a matrix $Y\in\mathbb{R}^{n_{e}\times
n_{H}},$ $n_{e}=n+n_{w}+n_{v}+1,$ a diagonal matrix $D_{v}\in\mathbb{D}%
^{n_{v}},$ a positive semidefinite diagonal matrix $D_{xw}\in\mathbb{D}%
_{+}^{n+n_{w}},$ and a symmetric matrix $P\in\mathbb{S}^{n+1},$ satisfying the
following LMIs:\vspace*{-0.2in} \label{MILP_Invariance_LMI}
\begin{align*}
L_{1}^{T}PL_{1}-\theta L_{2}^{T}PL_{2}  & \preceq\operatorname{He}\left(
YH\right)  +L_{3}^{T}D_{xw}L_{3}+L_{4}^{T}D_{v}L_{4}-\left(  \lambda
+\mu\right)  L_{5}^{T}L_{5}\\[-0.07in]
\lambda & =\operatorname{Trace}D_{xw}+\operatorname{Trace}D_{v}\vspace*{-1in}%
\end{align*}
%
\vspace{-0.6in}\newline where\vspace*{-0.1in} {\small
\[
L_{1}\hspace*{-0.01in}=\hspace*{-0.01in}\left[  \hspace*{-0.04in}%
\begin{array}
[c]{c}%
F\vspace*{-0.1in}\\
L_{5}%
\end{array}
\hspace*{-0.04in}\right]  ,\text{ }L_{2}\hspace*{-0.01in}=\hspace
*{-0.01in}\left[  \hspace*{-0.04in}%
\begin{array}
[c]{lr}%
I_{n} & 0_{n\times(n_{e}-n)}\vspace*{-0.1in}\\
0_{1\times(n_{e}-1)} & 1
\end{array}
\hspace*{-0.04in}\right]  ,\text{ }L_{3}\hspace*{-0.01in}=\hspace
*{-0.01in}\left[  \hspace*{-0.04in}%
\begin{array}
[c]{l}%
I_{n+n_{w}}\vspace*{-0.1in}\\
0_{\left(  n_{v}+1\right)  \times(n+n_{w})}%
\end{array}
\hspace*{-0.04in}\right]  ^{T},\text{ }L_{4}\hspace*{-0.01in}=\hspace
*{-0.01in}\left[  \hspace*{-0.04in}%
\begin{array}
[c]{l}%
\vspace*{-0.48in}\\
0_{(n+n_{w})\times n_{v}}\vspace*{-0.1in}\\
I_{n_{v}}\vspace*{-0.1in}\\
0_{1\times n_{v}}\\
\vspace*{-0.38in}%
\end{array}
\hspace*{-0.04in}\right]  ^{T},\text{ }L_{5}\hspace*{-0.01in}=\hspace
*{-0.01in}\left[  \hspace*{-0.04in}%
\begin{array}
[c]{c}%
0_{(n_{e}-1)\times1}\vspace*{-0.1in}\\
1
\end{array}
\hspace*{-0.04in}\right]  ^{T}\vspace{0.2in}%
\]
}
\end{lemma}

\smallskip

\begin{proof}
The proof is presented in Appendix II
\end{proof}

The following theorem summarizes our results for verification of absence of
overflow and/or FTT for MILMs. The result follows from Lemma
\ref{MIPL_Invariance_Lemma} and Corollary \ref{Bddness and FTT} with $q=2$,
$h=f, $ though the theorem is presented without a detailed proof.\vspace
{-0.05in}

\begin{theorem}
\label{MILP_Correctness_Theorem}\textbf{Optimization-Based
MILM\ Verification.} Let $\alpha:0\prec\alpha\preceq I_{n}$ be a diagonal
positive definite matrix specifying the overflow limit. An overflow runtime
error does not occur during any execution of $\mathcal{P}$ if there exist
matrices $Y_{i}\in\mathbb{R}^{n_{e}\times n_{H}},$ diagonal matrices
$D_{iv}\in\mathbb{D}^{n_{v}},$ positive semidefinite diagonal matrices
$D_{ixw}\in\mathbb{D}_{+}^{n+n_{w}},$ and a symmetric matrix $P\in
\mathbb{S}^{n+1}$ satisfying the following LMIs:\vspace*{-0.2in}
\label{MILP_Thm}
\begin{align}
\lbrack%
\begin{array}
[c]{cc}%
x_{0} & 1
\end{array}
]P[%
\begin{array}
[c]{cc}%
x_{0} & 1
\end{array}
]^{T}  & \leq0,\text{\qquad}\forall x_{0}\in X_{0}\label{MILP_Thm_a}%
\\[-0.06in]
L_{1}^{T}PL_{1}-\theta L_{2}^{T}PL_{2}  & \preceq\operatorname{He}\left(
Y_{1}H\right)  +L_{3}^{T}D_{1xw}L_{3}+L_{4}^{T}D_{1v}L_{4}-\left(  \lambda
_{1}+\mu\right)  L_{5}^{T}L_{5}\label{MILP_Thm_b}\\[-0.06in]
L_{1}^{T}\Lambda L_{1}-L_{2}^{T}PL_{2}  & \preceq\operatorname{He}\left(
Y_{2}H\right)  +L_{3}^{T}D_{2xw}L_{3}+L_{4}^{T}D_{2v}L_{4}-\lambda_{2}%
L_{5}^{T}L_{5}\label{MILP_Thm_c}%
\end{align}%
\[
\vspace*{-0.7in}%
\]
\vspace*{-0.45in} \newline\noindent where $\Lambda=\operatorname{diag}\left\{
\alpha^{-2},-1\right\}  ,$ $\lambda_{i}=\operatorname{Trace}D_{ixw}%
+\operatorname{Trace}D_{iv},$ and $0\preceq D_{ixw},$ $i=1,2.$ In addition, if
$\mu+\theta>1,$ then $\mathcal{P}$ terminates in a most $T_{u}$ steps where
$T_{u}$ is given in (\ref{upperboundonTU}).
\end{theorem}

\subsubsection{Affine Invariants\label{Section:MILM-LI}}

Affine Lyapunov invariants can often establish strong properties, e.g.,
boundedness, for\ variables with simple uncoupled dynamics (e.g. counters) at
a low computational cost. For variables with more complicated dynamics, affine
invariants may simply establish sign-invariance (e.g., $x_{i}\geq0$) or more
generally, upper or lower bounds on some linear combination of certain
variables. As we will observe in Section \ref{sec:casestudy}, establishing
these simple behavioral properties is important as they can be recursively
added to the model (e.g., the matrix $H$ in a MILM, or the invariant sets
$X_{i}$ in a graph model) to improve the chances of success in proving
stronger properties via higher order invariants. The linear parameterization
of the subspace of linear functionals mapping $\mathbb{R}^{n}$ to
$\mathbb{R},\,$is given by:\vspace*{-0.12in}%
\begin{equation}
\mathcal{V}_{x}^{1}=\left\{  V:\mathbb{R}^{n}\mapsto\mathbb{R~}|~V(x)=K^{T}%
\left[  x\quad1\right]  ^{T},~K\in\mathbb{R}^{n+1}\right\}  .\vspace
*{-0.12in}\label{MILP_Linear_fctn}%
\end{equation}
\vspace{-0.5in}\newline It is possible to search for the affine invariants via
semidefinite programming or linear programming.

\begin{proposition}
\label{MIPL_LinearInvariance_Lemma}\textbf{SDP Characterization of Linear
Invariants:} There exists a $\left(  \theta,\mu\right)  $-Lyapunov invariant
$V\in\mathcal{V}_{x}^{1}$ for a program $\mathcal{P}\equiv\mathcal{S}\left(
F,H,X_{0},n,n_{w},n_{v}\right)  ,$ if there exists a matrix $Y\in
\mathbb{R}^{n_{e}\times n_{H}},$ a diagonal matrix $D_{v}\in\mathbb{D}^{n_{v}%
},$ a positive semidefinite diagonal matrix $D_{xw}\in\mathbb{D}_{+}^{\left(
n+n_{w}\right)  \times\left(  n+n_{w}\right)  },$ and a matrix $K\in
\mathbb{R}^{n+1}$ satisfying the following LMI:\vspace*{-0.2in}%
\begin{equation}
\operatorname{He}(L_{1}^{T}KL_{5}-\theta L_{5}^{T}K^{T}L_{2})\prec
\operatorname{He}(YH)+L_{3}^{T}D_{xw}L_{3}+L_{4}^{T}D_{v}L_{4}-\left(
\lambda+\mu\right)  L_{5}^{T}L_{5}\vspace*{-0.2in}\label{MILP-LI-SDP}%
\end{equation}
where $\lambda=\operatorname{Trace}D_{xw}+\operatorname{Trace}D_{v}$ and
$0\preceq D_{xw}.$
\end{proposition}

\begin{proposition}
\label{MIPL_LinearInvariance_Lemma_LP}\textbf{LP Characterization of Linear
Invariants: }There exists a $\left(  \theta,\mu\right)  $-Lyapunov invariant
for a program $\mathcal{P}\equiv\mathcal{S}\left(  F,H,X_{0},n,n_{w}%
,n_{v}\right)  $ in the class $\mathcal{V}_{x}^{1},$ if there exists a matrix
$Y\in\mathbb{R}^{1\times n_{H}},$ and nonnegative matrices $\underline{D}%
_{v},~\overline{D}_{v}\in\mathbb{R}^{1\times n_{v}},$~$\underline{D}%
_{xw},~\overline{D}_{xw}\in\mathbb{R}^{1\times\left(  n+n_{w}\right)  },$ and
a matrix $K\in\mathbb{R}^{n+1}$ satisfying:\vspace*{-0.2in}
\begin{subequations}
\label{MILP-LP-LP}%
\begin{align}
K^{T}L_{1}-\theta K^{T}L_{2}-YH-(\underline{D}_{xw}-\overline{D}_{xw}%
)L_{3}-(\underline{D}_{v}-\overline{D}_{v})L_{4}-\left(  D_{1}+\mu\right)
L_{5}  & =0\vspace*{-0.15in}\label{MILP-LP-LPa}\\
D_{1}+\left(  \overline{D}_{v}+\underline{D}_{v}\right)  \mathbf{1}%
_{r}+\left(  \overline{D}_{xw}+\underline{D}_{xw}\right)  \mathbf{1}%
_{n+n_{w}}  & \leq0\vspace*{-0.15in}\label{MILP-LP-LPb}\\
\overline{D}_{v},~\underline{D}_{v},~\overline{D}_{xw},~\underline{D}_{xw}  &
\geq0\vspace*{-0.15in}\label{MILP-LP-LPc}%
\end{align}%
\end{subequations}
\[
\vspace*{-0.8in}%
\]
where $D_{1}$ is either $0$ or $-1.$ As a special case of (\ref{MILP-LP-LP}%
),\ a subset of all the affine invariants is characterized by the set of all
solutions of the following system of linear equations:\vspace*{-0.2in}
\begin{equation}
K^{T}L_{1}-\theta K^{T}L_{2}+L_{5}=0\label{MILP-LP-LP-simple}%
\end{equation}

\end{proposition}

\begin{remark}
When the objective is to establish properties of the form $Kx\geq a$ for a
fixed $K,$ (e.g., when establishing sign-invariance for certain variables),
matrix $K$ in (\ref{MILP-LI-SDP})$-$(\ref{MILP-LP-LP-simple}) is fixed and
thus one can make $\theta$ a decision variable subject to $\theta\geq0.$
Exploiting this convexity is extremely helpful for successfully establishing
such properties.\vspace*{-0.02in}
\end{remark}

The advantage of using SDP for computation of linear invariants\vspace
*{-0.02in} is that efficient SDP relaxations for treatment of binary variables
exists \cite{Meg01}, \cite{Nesterov 2000}, \cite{Goemans1995},\vspace
*{-0.02in} though the computational costs are typically higher than the
LP-based approaches.\vspace*{-0.02in} In contrast, linear programming
relaxations\vspace*{-0.02in} of the binary constraints are more involved than
the\vspace*{-0.02in} corresponding SDP relaxations. Two extreme remedies can
be readily considered.\vspace*{-0.02in} The first is to relax the binary
constraints and treat the variables as continuous\vspace*{-0.02in} variables.
The second is to consider each of the $2^{n_{v}}$ different
possibilities\vspace*{-0.02in} (one for each vertex of $\left\{  -1,1\right\}
^{n_{v}}$) separately.\vspace*{-0.02in} This approach is practical only if
$n_{v}$ is small. More sophisticated\vspace*{-0.02in} schemes can be developed
based on hierarchical relaxations or convex\vspace*{-0.02in} hull
approximations of binary integer programs \cite{Sherali1994},
\cite{Lovasz1991}.\vspace*{-0.05in}

\subsection{Optimization of Lyapunov Invariants for Graph Models\vspace
*{-0.1in}}

A linear parameterization of the subspace of polynomial functionals with total
degree less than or equal to $d$ is given by:\vspace*{-0.25in}%
\begin{equation}
\mathcal{V}_{x}^{d}=\left\{  V:\mathbb{R}^{n}\mapsto\mathbb{R~}|~V(x)=K^{T}%
Z\left(  x\right)  ,\text{ }K\in\mathbb{R}^{N},\text{ }N=\binom{n+d}%
{d}\right\}  \vspace*{-0.05in}\label{Graph_Poly_fctn}%
\end{equation}
\vspace*{-0.45in}\newline where $Z\left(  x\right)  $ is a vector of length
$\binom{n+d}{d},$ consisting of all monomials of degree less than or equal to
$d$ in $n$ variables $x_{1},...,x_{n}.$ A linear parametrization of Lyapunov
invariants for graph models is defined according to (\ref{nodewiselyap}),
where for every $i\in\mathcal{N},$ we have $\sigma_{i}\left(  \cdot\right)
\in\mathcal{V}_{x}^{d\left(  i\right)  },$ where $d\left(  i\right)  $ is a
selected degree bound for $\sigma_{i}\left(  \cdot\right)  .$ Depending on the
dynamics of the model, the degree bounds $d\left(  i\right)  ,$ and the convex
relaxation technique, the corresponding optimization problem will become a
linear, semidefinite, or SOS optimization problem.

\subsubsection{Node-wise Polynomial Invariants}

We present generic conditions for verification over graph models using
SOS\ programming. For the specific case of verification of \textit{linear
graph models} via quadratic invariants and the $\mathcal{S}$-Procedure for
relaxation of non-convex constraints we present explicit LMIs. The
corresponding Matlab code is available at \cite{MVichSite}. The following
theorem follows from Corollary \ref{SafetyGraphCor1}.

\begin{theorem}
\label{Thm:Graphnodepoly}\textbf{Optimization-Based Graph Model Verification.}
Consider a program $\mathcal{P}$, and its graph model $G\left(  \mathcal{N}%
,\mathcal{E}\right)  .$ Let $V:\Omega^{n}\mapsto\mathbb{R},$ be given by
(\ref{nodewiselyap}), where $\sigma_{i}\left(  \cdot\right)  \in
\mathcal{V}_{x}^{d\left(  i\right)  }.$ Then, the functions $\sigma_{i}\left(
\cdot\right)  ,$ $i\in\mathcal{N}$ define a Lyapunov invariant for
$\mathcal{P},$ if for all $\left(  i,j,k\right)  \in\mathcal{E}$ we
have:\vspace*{-0.2in}%
\begin{equation}
-\sigma_{j}(T_{ji}^{k}\left(  x,w\right)  )+\theta_{ji}^{k}\sigma_{i}\left(
x\right)  -\mu_{ji}^{k}\in\Sigma\left[  x,w\right]  \text{ subject to }\left(
x,w\right)  \in\left(  \left(  X_{i}\cap\Pi_{ji}^{k}\right)  \times
\lbrack-1,1]^{n_{w}}\right)  \cap S_{ji}^{k}\vspace*{-0.2in}\label{D1D1}%
\end{equation}
Furthermore, $\mathcal{P}$ satisfies the unreachability property w.r.t. the
collection of sets $X_{i-},$ $i\in\mathcal{N}\backslash\left\{  \emptyset
\right\}  ,$ if there exist $\varepsilon_{i}\in\left(  0,\infty\right)  ,$
$i\in\mathcal{N}\backslash\left\{  \emptyset\right\}  ,$ such that\vspace
*{-0.2in}%
\begin{align}
-\sigma_{\emptyset}\left(  x\right)   & \in\Sigma\left[  x\right]  \text{
subject to }x\in X_{\emptyset}\vspace*{-0.2in}\label{D1D0}\\
\sigma_{i}\left(  x\right)  -\varepsilon_{i}  & \in\Sigma\left[  x\right]
\text{ subject to }x\in X_{i}\cap X_{i-},\ i\in\mathcal{N}\backslash\left\{
\emptyset\right\}  \vspace*{-0.25in}\label{D1D3}%
\end{align}%
\[
\vspace*{-0.65in}%
\]
\vspace*{-0.5in}\newline As discussed in Section \ref{Section:SOS-Relaxation},
the SOS relaxation techniques can be applied for formulating the search
problem for functions $\sigma_{i}$ satisfying (\ref{D1D1})--(\ref{D1D3}) as a
convex optimization problem. For instance, if\vspace*{-0.2in}
\[
\left(  \left(  X_{i}\cap\Pi_{ji}^{k}\right)  \times\lbrack-1,1]^{n_{w}%
}\right)  \cap S_{ji}^{k}=\left\{  \left(  x,w\right)  ~|~f_{p}\left(
x,w\right)  \geq0,\text{ }h_{l}\left(  x,w\right)  =0\right\}  ,\vspace
*{-0.2in}%
\]
then, (\ref{D1D1}) can be relaxed as an SOS optimization problem of the
following form:\vspace*{-0.2in}
\[
-\sigma_{j}(T_{ji}^{k}\left(  x,w\right)  )+\theta_{ji}^{k}\sigma_{i}\left(
x\right)  -\mu_{ji}^{k}\hspace{-0.01in}-\hspace{-0.03in}\sum\limits_{p}%
\hspace{-0.02in}\tau_{p}f_{p}-\hspace{-0.03in}\sum\limits_{p,q}\hspace
{-0.02in}\tau_{pq}f_{p}f_{q}-\hspace{-0.03in}\sum\limits_{l}\hspace
{-0.02in}\rho_{l}h_{l}\in\Sigma\left[  x,w\right]  ,\text{ s.t. }\tau_{p}%
,\tau_{pq}\in\Sigma\left[  x,w\right]  .\vspace*{-0.2in}%
\]
Software packages such as SOSTOOLS \cite{Prajna} or YALMIP \cite{Lofberg2004}
can then be used for formulating the SOS optimization problems as semidefinite
programs.\vspace*{-0.05in}
\end{theorem}

\subsubsection{Node-wise Quadratic Invariants for Quasi-Linear Graph Models}

In a quasi-linear graph model, all the transition labels are linear. The
passport labels and the invariant sets are defined by linear and/or quadratic
functions. Let $\bar{x}=[x\hspace{0.15in}1]^{T},$ $\bar{n}=n+1,$ and let, for
all\ $\left(  i,j,k\right)  \in\mathcal{E},$ a compact description of the set
$X_{i}\cap\Pi_{ji}^{k}$ be given by:\vspace*{-0.1in}%
\begin{equation}
X_{i}\cap\Pi_{ji}^{k}=\left\{  x\in\mathbb{R}^{n}~|~\bar{x}^{T}Q_{ji,l}%
^{k}\bar{x}\geq0,~\bar{x}^{T}R_{ji,m}^{k}\bar{x}=0,~G_{ji}^{k}\bar{x}%
\geq0,\text{~}H_{ji}^{k}\bar{x}=0\right\}  ,\vspace*{-0.1in}%
\label{eq:CompactPI}%
\end{equation}
where $G_{ji}^{k}\in\mathbb{R}^{g_{ji}^{k}\times\bar{n}},$ $H_{ji}^{k}%
\in\mathbb{R}^{h_{ji}^{k}\times\bar{n}},$ $Q_{ji,s}^{k}\in\mathbb{R}^{\bar
{n}\times\bar{n}},$ $s\in S_{ji}^{k},$ and $R_{ji,m}^{k}\in\mathbb{R}^{\bar
{n}\times\bar{n}},$ $m\in M_{ji}^{k},$ where $S_{ji}^{k}$ and $M_{ji}^{k}$ are
some sets. For all $\left(  i,j,k\right)  \in\mathcal{E}$ let the transitions
be given by $T_{ji}^{k}x=A_{ji}^{k}x+B_{ji}^{k}w+E_{ji}^{k},$ where
$w\in\left[  -1,1\right]  ^{q_{ji}^{k}}.$ Before proceeding to the next
theorem, for convenience, we define the following matrices:\vspace
*{-0.0in}{\small
\begin{align}
F_{ji}^{k} &  =\left[
\begin{array}
[c]{llr}%
\vspace*{-0.42in} &  & \\
A_{ji}^{k}\hspace*{0.04in}\vspace*{-0.08in} & B_{ji}^{k}\hspace*{0.04in} &
\multicolumn{1}{l}{E_{ji}^{k}}\\
0_{1\times n}\hspace*{0.04in} & 0_{1\times q_{ji}^{k}}\hspace*{0.04in} &
\multicolumn{1}{l}{1}\\
\vspace*{-0.35in} &  &
\end{array}
\right]  ,\qquad\qquad L_{ji}^{k}=\left[
\begin{array}
[c]{ccc}%
\vspace*{-0.42in} &  & \\
\multicolumn{1}{l}{I_{n}\hspace*{0.04in}\vspace*{-0.08in}} &
\multicolumn{1}{l}{0_{n\times q_{ji}^{k}}} & \multicolumn{1}{l}{0_{n\times1}%
}\\
\multicolumn{1}{l}{0_{1\times n}} & \multicolumn{1}{l}{0_{1\times q_{ji}^{k}}}
& \multicolumn{1}{l}{1}\\
\vspace*{-0.35in} &  &
\end{array}
\right]  \label{eq:FandL}\\[0.1in]
W_{ji}^{k} &  =\left[
\begin{array}
[c]{ccc}%
\vspace*{-0.42in} &  & \\
0_{q_{ji}^{k}\times n} & I_{q_{ji}^{k}} & 0_{q_{ji}^{k}\times1}\\
\vspace*{-0.35in} &  &
\end{array}
\right]  ,\qquad\qquad\hspace{0.03in}K_{ji}^{k}=\left[
\begin{array}
[c]{ccc}%
\vspace*{-0.42in} &  & \\
0_{1\times n} & 0_{1\times q_{ji}^{k}} & 1\\
\vspace*{-0.35in} &  &
\end{array}
\right]  .\label{eq:WandK}%
\end{align}
}%

\[
\vspace*{-1.0in}%
\]
\newpage

\begin{theorem}
\label{Thm:GraphModelLMI}\textbf{SDP-Based Verification of Quasi-Linear Graph
Models.} Consider a program $\mathcal{P}$, and its quasi-linear graph model
$G\left(  \mathcal{N},\mathcal{E}\right)  .$ For $j\in\mathcal{N},$ Let
$\alpha_{j}\succ0$ be a diagonal matrix specifying the
overflow limit and let $\Lambda_{j}=\operatorname{diag}\left\{  \alpha
_{j}^{-2},-1\right\}  .$ Let $\theta_{j\emptyset}^{k}=0$ for all $\left(
\emptyset,j,k\right)  \in\mathcal{E}.$ An overflow runtime error does not
occur during any execution of $\mathcal{P},$ if for all $\left(  i,j,k\right)
\in\mathcal{E}$ the following LMIs are satisfied:\vspace*{-0.10in}%
\begin{align}
\mathcal{F}^{T}P_{j}\mathcal{F}-\theta_{ji}^{k}\mathcal{L}^{T}P_{i}%
\mathcal{L}  & \preceq\mathcal{W}^{T}D\mathcal{W}-\left(  \operatorname{Tr}%
\left(  D\right)  +\mu_{ji}^{k}\right)  \mathcal{K}^{T}\mathcal{K}%
+\operatorname{He}\left(  Y\mathcal{HL+}Z\mathcal{GL}\right)  +\lambda
_{ji}^{k}\label{LGMLMI1}\\[-0.05in]
\mathcal{F}^{T}\Lambda_{j}\mathcal{F}^{T}-\mathcal{L}^{T}P_{i}\mathcal{L}  &
\preceq\mathcal{W}^{T}\widetilde{D}\mathcal{W}-\operatorname{Tr}(\widetilde
{D})\mathcal{K}^{T}\mathcal{K}+\operatorname{He}(\widetilde{Y}\mathcal{HL+}%
\widetilde{Z}\mathcal{GL)}+\widetilde{\lambda}_{ji}^{k}\label{LGMLMI2.}%
\end{align}
\[
\vspace*{-0.75in}
\]
where for all $\left(  i,j,k\right)  \in\mathcal{E}$ we have\vspace*{-0.2in}%
\begin{align}
\lambda_{ji}^{k}  & =\sum\limits_{s}\eta_{ji,s}^{k}\mathcal{L}^{T}Q_{ji,s}%
^{k}\mathcal{L}+\sum\limits_{m}\rho_{ji,m}^{k}\mathcal{L}^{T}R_{ji,m}%
^{k}\mathcal{L}\label{eq1:thm:GMLMI}\\[-0.05in]
\widetilde{\lambda}_{ji}^{k}  & =\sum\limits_{s}\widetilde{\eta}_{ji,s}%
^{k}\mathcal{L}^{T}Q_{ji,s}^{k}\mathcal{L}+\sum\limits_{m}\widetilde{\rho
}_{ji,m}^{k}\mathcal{L}^{T}R_{ji,m}^{k}\mathcal{L}\label{eq2:thm:GMLMI}%
\\[-0.05in]
(Y,\widetilde{Y},Z,\widetilde{Z},D,\widetilde{D})  & =(Y_{ij}^{k}%
,\widetilde{Y}_{ij}^{k},Z_{ij}^{k},\widetilde{Z}_{ij}^{k},D_{ij}%
^{k},\widetilde{D}_{ij}^{k})\label{eq3:thm:GMLMI}\\[-0.05in]
\left(  \mathcal{F},\mathcal{G},\mathcal{H},\mathcal{K},\mathcal{L}%
,\mathcal{W}\right)    & =(F_{ji}^{k},G_{ji}^{k},H_{ji}^{k},K_{ji}^{k}%
,L_{ji}^{k},W_{ji}^{k})\label{eq4:thm:GMLMI}%
\end{align}
\[
\vspace*{-0.75in}
\]
where the $\mathcal{S}$-Procedure multipliers in (\ref{eq1:thm:GMLMI}%
)$-$(\ref{eq3:thm:GMLMI}) are constrained as follows:\vspace*{-0.15in}%
\begin{align}
\eta_{ji,s}^{k},~\widetilde{\eta}_{ji,s}^{k}  & \in\mathbb{R}_{-}%
,\qquad\forall\left(  i,j,k\right)  \in\mathcal{E},\text{ }s\in S_{ji}%
^{k}\label{mult1}\\[-0.05in]
\rho_{ji,m}^{k},~\widetilde{\rho}_{ji,m}^{k}  & \in\mathbb{R},\qquad
\forall\left(  i,j,k\right)  \in\mathcal{E},\text{ }m\in M_{ji}^{k}%
\label{mult2}\\[-0.05in]
Y_{ji}^{k},~\widetilde{Y}_{ji}^{k}  & \in\mathbb{R}^{(n+q_{ji}^{k}+1)\times
h_{ji}^{k}},\qquad\forall\left(  i,j,k\right)  \in\mathcal{E}\label{mult3}\\[-0.05in]
Z_{ji}^{k},~\widetilde{Z}_{ji}^{k}  & \in\mathbb{R}_{-}^{(n+q_{ji}%
^{k}+1)\times g_{ji}^{k}},\qquad\forall\left(  i,j,k\right)  \in
\mathcal{E}\label{mult4}\\[-0.05in]
D_{ij}^{k},~\widetilde{D}_{ij}^{k}  & \in\mathbb{D}_{+}^{q_{ji}^{k}\times
q_{ji}^{k}},\qquad\forall\left(  i,j,k\right)  \in\mathcal{E},\label{mult5}%
\end{align}
\[
\vspace*{-0.75in}
\]
and the matrices in (\ref{eq4:thm:GMLMI}) are defined in (\ref{eq:CompactPI}%
)$-$(\ref{eq:WandK}). In addition, if for every simple cycle $\mathcal{C}\in
G,$ we have:$\vspace*{-0.15in}$%
\begin{equation}
\theta\left(  \mathcal{C}\right)  +\mu\left(  \mathcal{C}\right)  >1,\text{
\qquad\ }\mu\left(  \mathcal{C}\right)  >0,\vspace*{-0.15in}%
\end{equation}
with $\theta\left(  \mathcal{C}\right)  $ and $\mu\left(  \mathcal{C}\right)
$ defined as in (\ref{MultiplicativeThetanDefs}), then $\mathcal{P}$
terminates in at most $T_{u}$ iterations where\vspace*{-0.05in}%
\[
T_{u}=\sum_{\mathcal{C}\in G:\theta\left(  \mathcal{C}\right)  \neq1}%
\frac{\log~\left(  \theta\left(  \mathcal{C}\right)  +\mu\left(
\mathcal{C}\right)  -1\right)  -\log~\mu\left(  \mathcal{C}\right)  }%
{\log\theta\left(  \mathcal{C}\right)  }+\sum_{\mathcal{C}\in G:\theta\left(
\mathcal{C}\right)  =1}\frac{1}{\mu\left(  \mathcal{C}\right)  }.
\]

\end{theorem}

\vspace*{0.15in}

\begin{proof}
See Appendix II.
\end{proof}

\textbf{Remarks}

\begin{enumerate}
\item A simpler, computationally less demanding but possibly more conservative
variation of Theorem \ref{Thm:GraphModelLMI} can be formulated by replacing
(\ref{LGMLMI2.}) with\vspace*{-0.2in}
\begin{align}
\Lambda_{i}-P_{i}\preceq0,\text{\qquad}\forall i\in\mathcal{N}\label{LGMLMI2Simple}.
\end{align}
\[
\vspace{-0.75in}
\]
Subsequently, all the associated multipliers (i.e., the variables denoted by
\textit{tilde} characters in (\ref{mult1})$-$(\ref{mult5})) will be eliminated.

\item It is possible to use the same set of multipliers (i.e., the variables
defined in (\ref{mult1})$-$(\ref{mult5})) across all edges $\left(
i,j,k\right)  \in\mathcal{E}.$ This approach would also result in a
computationally less demanding but possibly conservative formulation.

\item To improve the chances of finding stronger invariants, quadratic
equality and inequality constraints can be generated from linear constraints
and added to the set of quadratic constraints in characterization of the
set\ $X_{i}\cap\Pi_{ji}^{k}$ (cf. equation \ref{eq:CompactPI})$.$ Given scalar
linear constraints of the form $G_{ji,r}^{k}\bar{x}\geq0,$ $r\in
\mathbb{Z}\left(  1,g_{ji}^{k}\right)  $ and/or $H_{ji,s}^{k}\bar{x}=0,$
$s\in\mathbb{Z}\left(  1,h_{ji}^{k}\right)  ,$ define\ $R_{s_{1}s_{2}%
}=H_{s_{1}}^{T}H_{s_{2}}+H_{s_{2}}^{T}H_{s_{1}},$ $R_{s_{1}r_{1}}=H_{s_{1}%
}^{T}G_{r_{1}}+G_{r_{1}}^{T}H_{s_{1}},$ and $Q_{r_{1}r_{2}}=G_{r_{1}}%
^{T}G_{r_{2}}+G_{r_{2}}^{T}G_{r_{1}}.$ Then, we have $\bar{x}^{T}R_{s_{1}%
s_{2}}\bar{x}=0,$ $\bar{x}^{T}R_{s_{1}r_{1}}\bar{x}=0,$ and $\bar{x}%
^{T}Q_{r_{1}r_{2}}\bar{x}\geq0,$ for all $r_{1},r_{2}\in\mathbb{Z}\left(
1,g_{ji}^{k}\right)  $ and $s_{1},s_{2}\in\mathbb{Z}\left(  1,h_{ji}%
^{k}\right)  .$ 
\item Replacing \eqref{LGMLMI2Simple} with $z\Lambda_{i}-P_{i}\preceq0,$ where $z>0$ is a decision variable can improve the scaling and numerical conditioning of the associated semidefinite program, resulting in stronger invariants and improved analysis. The statement of the theorem on absence of overflow remains unchanged under this modification. Only the upper bound on the maximum number of iterations needs to be adjusted accordingly since we now have $\left\Vert \sigma\left(  \mathcal{C}\right) \right\Vert _{\infty}\leq z$ instead of $\left\Vert \sigma\left(  \mathcal{C}\right)
\right\Vert _{\infty}\leq1.$ The same scaling arguments apply to \eqref{LGMLMI2.}.
\item MATLAB implementations of Theorem \ref{Thm:GraphModelLMI}\ and a few
variations of it can be found in \cite{MVichSite}.
\end{enumerate}

\bigskip

\section{Case Studies\label{sec:casestudy}\vspace*{-0.1in}}

\subsection{Euclidean Division\label{sec:casestudy1}}

In this section we apply the framework to the analysis of Program 4 displayed
below.\vspace*{-0.11in}{\small
\begin{gather*}%
\begin{tabular}
[c]{|l|}\hline
\hspace*{-0.16in}%
\begin{tabular}
[c]{l}%
$%
\begin{array}
[c]{ll}
& \mathrm{/\ast\ EuclideanDivision.c\ \ast/}\\
& \vspace*{-0.45in}\\
\mathrm{F}0:\hspace*{0.1in} & \mathrm{int~IntegerDivision~(~int~dd,int~dr~)}\\
& \vspace*{-0.45in}\\
\mathrm{F}1:\hspace*{0.1in} & \mathrm{\{int~q=\{0\};~int~r=\{dd\};}\\
& \vspace*{-0.45in}\\
\mathrm{F}2:\hspace*{0.1in} & \mathrm{{while}\text{ }{{(r>=dr)~}{\{}}}\\
& \vspace*{-0.45in}\\
\mathrm{F}3:\hspace*{0.1in} & \mathrm{{{\hspace{0.24in}{q=q+1;}}}}\\
& \vspace*{-0.45in}\\
\mathrm{F}4:\hspace*{0.1in} & \mathrm{{{\hspace{0.24in}r=r-dr;}}}\\
& \vspace*{-0.45in}\\
\mathrm{F\hspace*{-0.04in}}\Join:\hspace*{0.1in} & \mathrm{return~r;\}}%
\end{array}
$\\
$%
\begin{array}
[c]{ll}%
\mathrm{L}0:\hspace*{0.1in} & \mathrm{int~main~(~int~X,int~Y~)~\{}\\
& \vspace*{-0.45in}\\
\mathrm{L}1:\hspace*{0.1in} & \mathrm{int~rem=\{0\};}\\
& \vspace*{-0.45in}\\
\mathrm{L}2:\hspace*{0.1in} & \mathrm{while~(Y~>~0)~\{}\\
& \vspace*{-0.45in}\\
\mathrm{L}3:\hspace*{0.1in} & \mathrm{{{\hspace{0.24in}}%
{rem=IntegerDivision~(X~,~Y);}}}\\
& \vspace*{-0.45in}\\
\mathrm{L}4:\hspace*{0.1in} & \mathrm{{{\hspace{0.24in}X=Y;}}}\\
& \vspace*{-0.45in}\\
\mathrm{L}5:\hspace*{0.1in} & \mathrm{{{\hspace{0.24in}Y=rem;}}}\\
& \vspace*{-0.45in}\\
\mathrm{L\hspace*{-0.04in}}\Join:\hspace*{0.1in} & \mathrm{return~X;\}\}}%
\end{array}
$%
\end{tabular}
\\\hline
\end{tabular}
\raisebox{-1.5056in}{\includegraphics[
height=3.2in,
width=4.1in
]%
{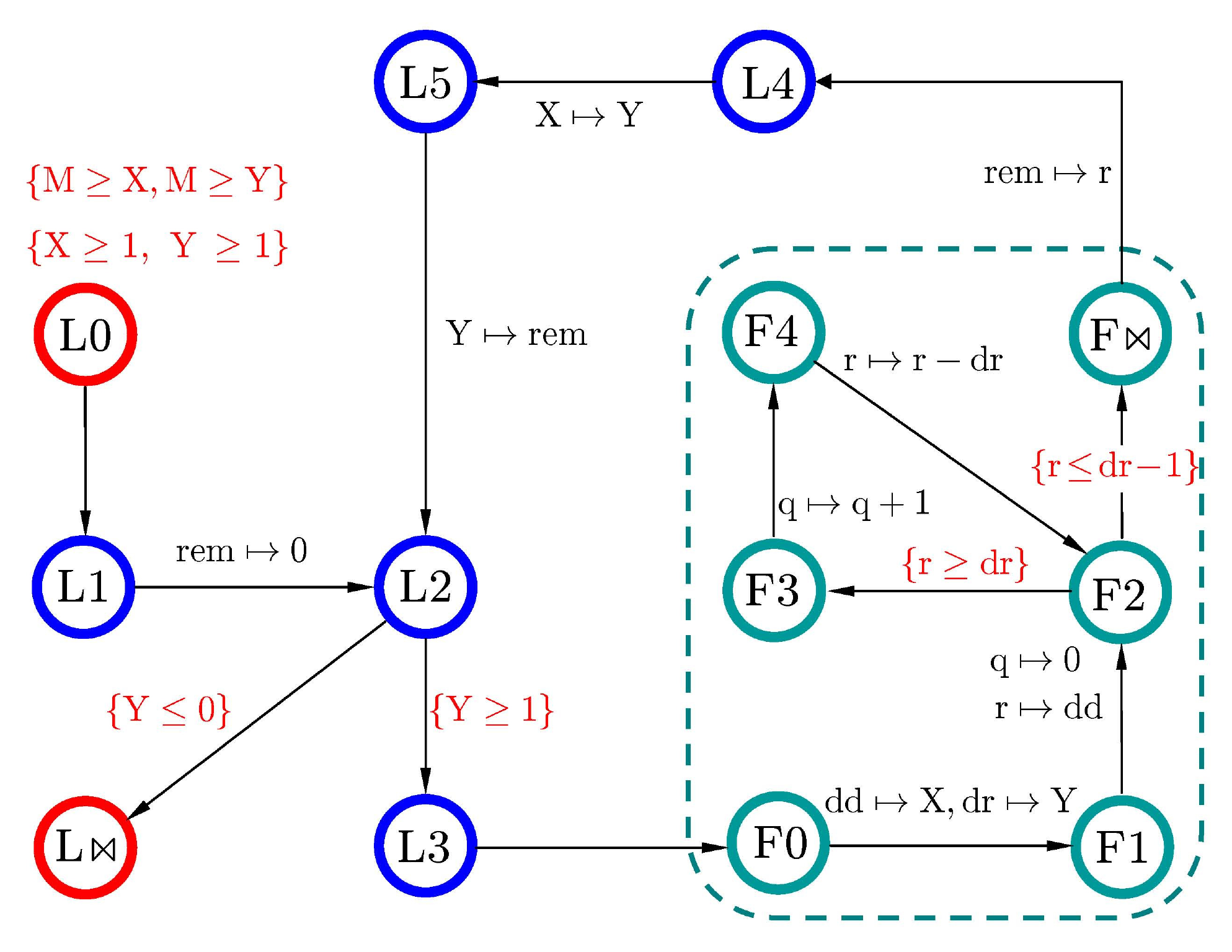}%
}
\\[0.01in]%
\begin{array}
[c]{c}%
\begin{array}
[c]{cccccc}
&  &  & \text{Program 4: Euclidean Division and its Graph Model} &
\vspace*{-0.35in} &
\end{array}
\end{array}
\end{gather*}
}Program 4 takes two positive integers $\mathrm{X}\in\left[  1,\mathrm{M}%
\right]  $ and $\mathrm{Y}\in\left[  1,\mathrm{M}\right]  $ as the input and
returns their greatest common divisor by implementing the Euclidean Division
algorithm. Note that the \textsc{Main} function in Program 4 uses the
\textsc{IntegerDivision} program (Program 1).\vspace*{-0.05in}

\subsubsection{Global Analysis\vspace*{-0.05in}}

A global model can be constructed by embedding the dynamics of the
\textsc{IntegerDivision} program within the dynamics of \textsc{Main}. A
labeled graph model is shown alongside the text of the program. This model has
a state space $X=\mathcal{N}\times\left[  -\mathrm{M},\mathrm{M}\right]
^{7},$ where $\mathcal{N}$ is the set of nodes as shown in the graph, and the
global state $\mathrm{x=}\left[  \mathrm{X,\ Y,\ rem,\ dd,\ dr,\ q,\ r}%
\right]  $ is an element of the hypercube $\left[  -\mathrm{M},\mathrm{M}%
\right]  ^{7}.$ A \textit{reduced} graph model can be obtained by combining
the effects of consecutive transitions and relabeling the reduced graph model
accordingly. While analysis of the full graph model is possible, working with
a \textit{reduced} model is computationally advantageous. Furthermore, mapping
the properties of the reduced graph model to the original model is
algorithmic. Interested readers may consult \cite{RoozbehaniHSCC} for further
elaboration on this topic. For the graph model of Program 4, a reduced model
can be obtained by first eliminating nodes $\mathrm{F}_{\Join},$
$\mathrm{L}_{4},$ $\mathrm{L}_{5},$ $\mathrm{L}_{3},$ $\mathrm{F}_{0},$
$\mathrm{F}_{1},$ $\mathrm{F}_{3},$ $\mathrm{F}_{4},$ and $\mathrm{L}_{1},$
(Figure \ref{fig:redmodel} Left) and composing the transition and passport
labels. Node $\mathrm{L}_{2}$ can be eliminated as well to obtain a further
reduced model with only three nodes: $\mathrm{F}_{2},$ $\mathrm{L}_{0},$
$\mathrm{L}_{\Join}.$ (Figure \ref{fig:redmodel} Right).
\begin{figure}
[tbh]
\begin{center}
\includegraphics[
height=1.5in,
width=5in
]%
{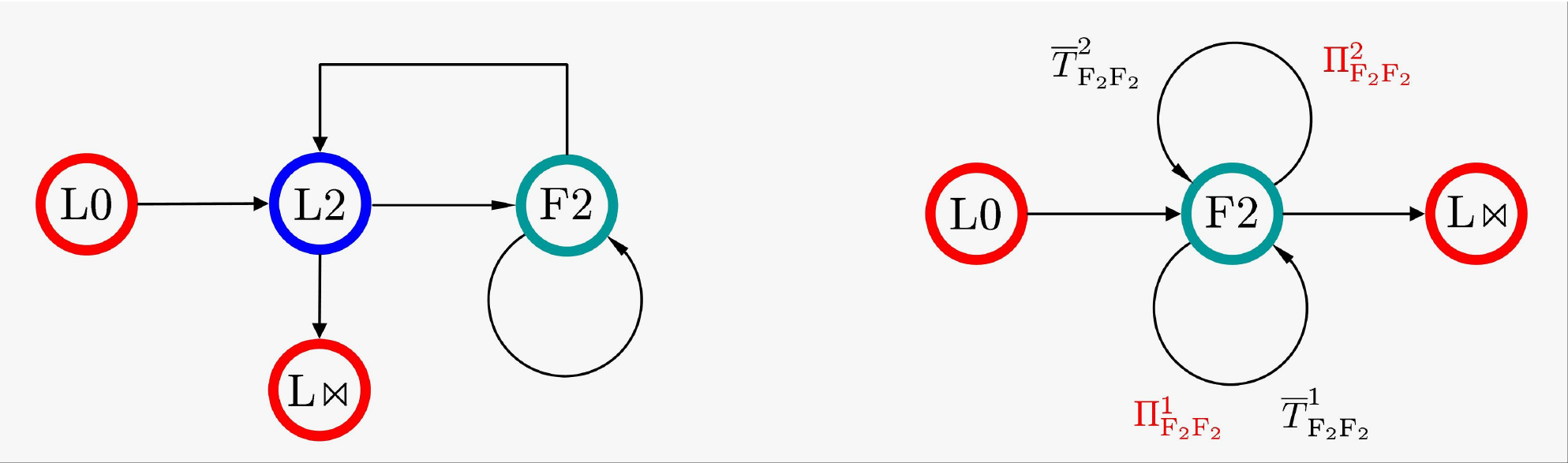}%
\caption{Two reduced models of the graph model of Program 4.}%
\label{fig:redmodel}%
\end{center}
\vspace{-0.25in}
\end{figure}
This is the model
that we will analyze. The passport and transition labels associated with the
reduced model are as follows:%
\vspace{-0.15in}
\begin{align*}
&
\begin{array}
[c]{rllcrrll}%
\overline{T}_{\mathrm{F}2\mathrm{F}2}^{1} & \hspace*{-0.06in}:\hspace
*{-0.06in} & \mathrm{x}\mapsto\left[
\mathrm{X,\ Y,\ rem,\ dd,\ dr,\ q+1,\ r-dr}\right]  & \vspace*{-0.05in} &  &
\overline{T}_{\mathrm{F}2\mathrm{F}2}^{2} & \hspace*{-0.06in}:\hspace
*{-0.06in} & \mathrm{x}\mapsto\left[  \mathrm{Y,\ r,\ r,\ Y,\ r,\ 0,\ Y}%
\right] \\
\overline{T}_{\mathrm{L}0\mathrm{F}2} & \hspace*{-0.06in}:\hspace*{-0.06in} &
\mathrm{x}\mapsto\left[  \mathrm{X,\ Y,\ 0,\ X,\ Y,\ 0,\ X}\right]  &
\vspace*{-0.05in} &  & \overline{T}_{\mathrm{F}2\mathrm{L}\Join} &
\hspace*{-0.06in}:\hspace*{-0.06in} & \mathrm{x}\mapsto\left[
\mathrm{Y,\ r,\ r,\ dd,\ dr,\ q,\ r}\right]
\end{array}
\\
&
\begin{array}
[c]{rllcrrlllllll}%
\Pi_{\mathrm{F}2\mathrm{F}2}^{2} & \hspace*{-0.06in}:\hspace*{-0.06in} &
\left\{  \mathrm{x~|~1\leq r\leq dr-1}\right\}  &  &  & \Pi_{\mathrm{F}%
2\mathrm{F}2}^{1} & \hspace*{-0.06in}:\hspace*{-0.06in}\vspace*{-0.05in} &
\left\{  \mathrm{x~|~r\geq dr}\right\}  &  &  & \Pi_{\mathrm{F}2\mathrm{L}%
\Join} & \hspace*{-0.06in}:\hspace*{-0.06in} & \left\{  \mathrm{x~|~r\leq
dr-1,}\text{ }\mathrm{r\leq0}\right\}
\end{array}
\end{align*}
\vspace*{-0.45in}

\noindent Finally, the invariant sets that can be readily included in the
graph model (cf. Section \ref{sec:constabst}) are:\vspace*{-0.25in}%
\[
X_{\mathrm{L}0}=\left\{  \mathrm{x}~|~\mathrm{M}\geq\mathrm{X},~\mathrm{M}%
\geq\mathrm{Y,~X}\geq\mathrm{1},~\mathrm{Y}\geq\mathrm{1}\right\}
,\text{~~}X_{\mathrm{F}2}=\left\{  \mathrm{x}~|~\mathrm{dd}=\mathrm{X}%
,~\mathrm{dr}=\mathrm{Y}\right\}  ,\text{~~}X_{\mathrm{L}\Join}=\left\{
\mathrm{x}~|~\mathrm{Y}\leq\mathrm{0}\right\}  .\vspace*{-0.2in}%
\]
We are interested in generating certificates of termination and absence of
overflow. First, by recursively searching for linear invariants we are able to
establish simple lower bounds on all variables in just two rounds (the
properties established in the each round are added to the model and the next
round of search begins). For instance, the property $\mathrm{X}\geq\mathrm{1}$
is established only after $\mathrm{Y}\geq\mathrm{1}$ is established. These
results, which were obtained by applying the first part of Theorem
\ref{Thm:Graphnodepoly} (equations (\ref{D1D1})-(\ref{D1D0}) only)\ with
linear functionals are summarized in Table\ II.\vspace*{-0.4in}

{\scriptsize
\begin{gather*}
\text{\textsc{TABLE} \textsc{II}}\\%
\begin{tabular}
[c]{|r|c|c|c|c|c|c|c|}\hline
$\text{Property}$ & $\mathrm{q}\geq\mathrm{0}$ & $\mathrm{Y}\geq\mathrm{1}$ &
$\mathrm{\mathrm{dr}}\geq\mathrm{\mathrm{1}}$ & $\mathrm{\mathrm{rem}}%
\geq\mathrm{0}$ & $\mathrm{\mathrm{dd}}\geq\mathrm{\mathrm{1}}$ &
$\mathrm{X}\geq\mathrm{1}$ & $\mathrm{r}\geq\mathrm{0}$\\\hline
$\text{Proven in Round}$ & $\mathrm{\mathrm{I}}$ & $\mathrm{\mathrm{I}}$ &
$\mathrm{\mathrm{I}}$ & $\mathrm{\mathrm{I}}$ & $\mathrm{\mathrm{II}}$ &
$\mathrm{\mathrm{II}}$ & $\mathrm{\mathrm{II}}$\\\hline
$\sigma_{\mathrm{F}_{2}}\left(  \mathrm{x}\right)  =$ & $-\mathrm{q}$ &
$\mathrm{1}-\mathrm{Y}$ & $\mathrm{1}-\mathrm{dr}$ & $-\mathrm{\mathrm{rem}}$
& $\mathrm{1}-\mathrm{dd}$ & $\mathrm{1}-\mathrm{X}$ & $-\mathrm{r}$\\\hline
$\left(  \theta_{\mathrm{F}2\mathrm{F}2}^{1},\mu_{\mathrm{F}2\mathrm{F}2}%
^{1}\right)  $ & $\left(  1,1\right)  $ & $\left(  1,0\right)  $ & $\left(
1,0\right)  $ & $\left(  1,0\right)  $ & $\left(  0,0\right)  $ & $\left(
0,0\right)  $ & $\left(  0,0\right)  $\\\hline
$\left(  \theta_{\mathrm{F}2\mathrm{F}2}^{2},\mu_{\mathrm{F}2\mathrm{F}2}%
^{2}\right)  $ & $\left(  0,0\right)  $ & $\left(  0,0\right)  $ & $\left(
0,0\right)  $ & $\left(  0,0\right)  $ & $\left(  0,0\right)  $ & $\left(
0,0\right)  $ & $\left(  0,0\right)  $\\\hline
\end{tabular}
\end{gather*}
} We then add these properties to the node invariant sets to obtain stronger
invariants that certify FTT and boundedness of all variables in $\left[
-\mathrm{M},\mathrm{M}\right]  $. By applying Theorem \ref{Thm:Graphnodepoly}
and SOS programming using YALMIP \cite{Lofberg2004}, the following invariants
are found\footnote{Different choices of polynomial degrees for the Lyapunov
invariant function and the multipliers, as well as different choices for
$\theta,\mu$ and different rounding schemes lead to different invariants. Note
that rounding is not essential.} (after post-processing, rounding the
coefficients, and reverifying):\vspace*{-0.35in}%

\[%
\begin{array}
[c]{llll}%
\sigma_{1\mathrm{F}_{2}}\left(  \mathrm{x}\right)  =0.4\left(  \mathrm{Y}%
-\mathrm{M}\right)  \left(  2+\mathrm{M}-\mathrm{r}\right)  &  &  &
\sigma_{2\mathrm{F}_{2}}\left(  \mathrm{x}\right)  =\left(  \mathrm{q}%
\times\mathrm{Y}+\mathrm{r}\right)  ^{2}-\mathrm{M}^{2}\\
\sigma_{3\mathrm{F}_{2}}\left(  \mathrm{x}\right)  =\left(  \mathrm{q}%
+\mathrm{r}\right)  ^{2}-\mathrm{M}^{2} &  &  & \sigma_{4\mathrm{F}_{2}%
}\left(  \mathrm{x}\right)  =0.1\left(  \mathrm{Y}-\mathrm{M}+5\mathrm{Y}%
\times\mathrm{M}+\mathrm{Y}^{2}-6\mathrm{M}^{2}\right) \\
\sigma_{5\mathrm{F}_{2}}\left(  \mathrm{x}\right)  =\mathrm{Y}+\mathrm{r}%
-2\mathrm{M}+\mathrm{Y}\times\mathrm{M}-\mathrm{M}^{{\normalsize 2}} &  &  &
\sigma_{6\mathrm{F}_{2}}\left(  \mathrm{x}\right)  =\mathrm{r}\times
\mathrm{Y}+\mathrm{Y}-\mathrm{M}^{2}-\mathrm{M}%
\end{array}
\]
The properties proven by these invariants are summarized in the Table III. The
specifications that the program terminates and that $\mathrm{x}\in\left[
-\mathrm{M},\mathrm{M}\right]  ^{7}$ for all initial conditions \textrm{X,~Y}%
$\in\left[  1,\mathrm{M}\right]  ,$ could not be established in one shot, at
least when trying polynomials of degree $d\leq4.$ For instance, $\sigma
_{1\mathrm{F}_{2}}$ certifies boundedness of all the variables except
$\mathrm{q},$ while $\sigma_{2\mathrm{F}_{2}}$ and $\sigma_{3\mathrm{F}_{2}}$
which certify boundedness of all variables including $\mathrm{q}$ do not
certify FTT. Furthermore, boundedness of some of the variables is established
in round II, relying on boundedness properties proven in round I. Given
$\sigma\left(  \mathrm{x}\right)  \leq0$ (which is found in round I), second
round verification can be done by searching for a strictly positive polynomial
$p\left(  \mathrm{x}\right)  $ and a nonnegative polynomial $q\left(
\mathrm{x}\right)  \geq0$ satisfying:\vspace*{-0.15in}%
\begin{equation}
q\left(  \mathrm{x}\right)  \sigma\left(  \mathrm{x}\right)  -p\left(
\mathrm{x}\right)  (\left(  \overline{T}\mathrm{x}\right)  _{i}^{2}%
-\mathrm{M}^{2})\geq0,\text{\qquad}\overline{T}\in\{\overline{T}%
_{\mathrm{F}2\mathrm{F}2}^{1},\overline{T}_{\mathrm{F}2\mathrm{F}2}%
^{2}\}\vspace*{-0.15in}\label{eq:roundII}%
\end{equation}
where the inequality (\ref{eq:roundII}) is further subject to boundedness
properties established in round I, as well as the usual passport conditions
and basic invariant set conditions.\vspace*{-0.3in}

{\scriptsize
\begin{gather*}
\text{\textsc{TABLE\ III}}\\%
\begin{tabular}
[c]{|r|c|c|c|c|}\hline
Invariant $\sigma_{\mathrm{F}_{2}}\left(  \mathrm{x}\right)  =\hspace
*{-0.02in}$ & $\sigma_{1\mathrm{F}_{2}}\left(  \mathrm{x}\right)  $ &
$\hspace*{-0.02in}\sigma_{2\mathrm{F}_{2}}\left(  \mathrm{x}\right)
,\sigma_{3\mathrm{F}_{2}}\left(  \mathrm{x}\right)  \hspace*{-0.02in}$ &
$\sigma_{4\mathrm{F}_{2}}\left(  \mathrm{x}\right)  $ & $\sigma_{5\mathrm{F}%
_{2}}\left(  \mathrm{x}\right)  ,\sigma_{6\mathrm{F}_{2}}\left(
\mathrm{x}\right)  $\\\hline
$%
\begin{array}
[c]{c}%
\left(  \theta_{\mathrm{F}2\mathrm{F}2}^{1},\mu_{\mathrm{F}2\mathrm{F}2}%
^{1}\right)  \smallskip
\end{array}
\hspace*{-0.02in}$ & $\left(  1,0\right)  $ & $\left(  1,0\right)  $ &
$\left(  1,0\right)  $ & $\left(  1,1\right)  $\\\hline
$%
\begin{array}
[c]{c}%
\left(  \theta_{\mathrm{F}2\mathrm{F}2}^{2},\mu_{\mathrm{F}2\mathrm{F}2}%
^{2}\right)  \smallskip
\end{array}
\hspace*{-0.02in}$ & $\left(  1,0.8\right)  $ & $\left(  0,0\right)  $ &
$\left(  1,0.7\right)  $ & $\left(  1,1\right)  $\\\hline
\multicolumn{1}{|l|}{Round I: $%
\begin{array}
[c]{c}%
\mathrm{x}_{i}^{2}\leq\mathrm{M}^{2}\text{ for }\mathrm{x}_{i}\mathrm{=}%
\smallskip\hspace*{-0.08in}%
\end{array}
\hspace*{-0.02in}$} & $\hspace*{-0.02in}\mathrm{Y,X,r,dr,rem,dd}%
\hspace*{-0.02in}$ & $\hspace*{-0.02in}\mathrm{q,Y,dr,rem}\hspace*{-0.02in}$ &
$\hspace*{-0.02in}\mathrm{Y,X,r,dr,rem,dd}\hspace*{-0.02in}$ & $\hspace
*{-0.02in}\mathrm{Y,dr,rem}\hspace*{-0.02in}$\\\hline
\multicolumn{1}{|l|}{Round II:$%
\begin{array}
[c]{c}%
\mathrm{x}_{i}^{2}\leq\mathrm{M}^{2}\text{ for }\mathrm{x}_{i}\mathrm{=}%
\smallskip\hspace*{-0.08in}%
\end{array}
\hspace*{-0.02in}$} &  & $\mathrm{X,r,dd}$ &  & $\mathrm{X,r,dd}$\\\hline
Certificate for FTT$\hspace*{-0.02in}$ & NO & NO & NO & YES, $T_{u}%
=2\mathrm{M}^{2}$\\\hline
\end{tabular}
\end{gather*}
} In conclusion, $\sigma_{2\mathrm{F}_{2}}\left(  \mathrm{x}\right)  $ or
$\sigma_{3\mathrm{F}_{2}}\left(  \mathrm{x}\right)  $ in conjunction with
$\sigma_{5\mathrm{F}_{2}}\left(  \mathrm{x}\right)  $ or $\sigma
_{6\mathrm{F}_{2}}\left(  \mathrm{x}\right)  $ prove finite-time termination
of the algorithm, as well as boundedness of all variables within $\left[
-\mathrm{M},\mathrm{M}\right]  $ for all initial conditions $\mathrm{X}%
,\mathrm{Y}\in\left[  1,\mathrm{M}\right]  ,$ for any $\mathrm{M}%
\geq1\mathrm{.}$ The provable bound on the number of iterations certified by
$\sigma_{5\mathrm{F}_{2}}\left(  \mathrm{x}\right)  $ and $\sigma
_{6\mathrm{F}_{2}}\left(  \mathrm{x}\right)  $ is $T_{u}=2\mathrm{M}^{2}$
(Corollary \ref{SafetyGraphCor1})$.$ If we settle for more conservative
specifications, e.g., $\mathrm{x}\in\left[  -k\mathrm{M},k\mathrm{M}\right]
^{7}$ for all initial conditions $\mathrm{X},\mathrm{Y}\in\left[
1,\mathrm{M}\right]  $ and sufficiently large $k,$ then it is possible to
prove the properties in one shot. We show this in the next section.

\subsubsection{MIL-GH Model}

For comparison, we also constructed the MIL-GH model associated with the
reduced graph in Figure \ref{fig:redmodel}. The corresponding matrices are
omitted for brevity, but details of the model along with executable Matlab
verification codes can be found in\ \cite{MVichSite}. The verification theorem
used in this analysis is an extension of Theorem
\ref{MILP_Correctness_Theorem} to analysis of MIL-GHM for specific numerical
values of $\mathrm{M},$ though it is certainly possible to perform this
modeling and analysis exercise for parametric bounded values of $\mathrm{M.}$
The analysis using the MIL-GHM is in general more conservative than SOS
optimization over the graph model presented earlier. This can be attributed to
the type of relaxations proposed (similar to those used in Lemma
\ref{MILP_Invariance_LMI}) for analysis of MILMs and MIL-GHMs. The benefits
are simplified analysis at a typically much less computational cost. The
certificate obtained in this way is a single quadratic function (for each
numerical value of $\mathrm{M}$), establishing a bound $\gamma\left(
\mathrm{M}\right)  $ satisfying $\gamma\left(  \mathrm{M}\right)  \geq\left(
\mathrm{X}^{2}+\mathrm{Y}^{2}+\mathrm{rem}^{2}+\mathrm{dd}^{2}+\mathrm{dr}%
^{2}+\mathrm{q}^{2}+\mathrm{r}^{2}\right)  ^{1/2}.$ Table IV summarizes the
results of this analysis which were performed using\ both Sedumi 1\_3 and
LMILAB solvers.\vspace*{-0.3in}

{\scriptsize
\begin{gather*}
\text{\textsc{TABLE\ IV}}\\%
\begin{tabular}
[c]{|r|c|c|c|c|c|}\hline
$\mathrm{M}$ & $10^{2}$ & $10^{3}$ & $10^{4}$ & $10^{5}$ & $10^{6}$\\\hline
Solver: LMILAB \cite{GahinetLMILAB}: $\gamma\left(  \mathrm{M}\right)  $ &
$5.99\mathrm{M}$ & $6.34\mathrm{M}$ & $6.43\mathrm{M}$ & $6.49\mathrm{M}$ &
$7.05\mathrm{M}$\\\hline
Solver: SEDUMI \cite{Strum1999}: $\gamma\left(  \mathrm{M}\right)  $ &
$6.00\mathrm{M}$ & $6.34\mathrm{M}$ & $6.44\mathrm{M}$ & $6.49\mathrm{M}$ &
NAN\\\hline
$%
\begin{array}
[c]{c}%
\left(  \theta_{\mathrm{F}2\mathrm{F}2}^{1},\mu_{\mathrm{F}2\mathrm{F}2}%
^{1}\right)  \smallskip
\end{array}
\hspace*{-0.02in}$ & $\left(  1,10^{-3}\right)  $ & $\left(  1,10^{-3}\right)
$ & $\left(  1,10^{-3}\right)  $ & $\left(  1,10^{-3}\right)  $ & $\left(
1,10^{-3}\right)  $\\\hline
$%
\begin{array}
[c]{c}%
\left(  \theta_{\mathrm{F}2\mathrm{F}2}^{2},\mu_{\mathrm{F}2\mathrm{F}2}%
^{2}\right)  \smallskip
\end{array}
\hspace*{-0.02in}$ & $\left(  1,10^{-3}\right)  $ & $\left(  1,10^{-3}\right)
$ & $\left(  1,10^{-3}\right)  $ & $\left(  1,10^{-3}\right)  $ & $\left(
1,10^{-3}\right)  $\\\hline
Upperbound on iterations$\hspace*{-0.02in}$ & $T_{u}=2$e$4$ & $T_{u}=8$e$4$ &
$T_{u}=8$e$5$ & $T_{u}=1.5$e$7$ & $T_{u}=3$e$9$\\\hline
\end{tabular}
\end{gather*}
}\vspace*{-0.35in}

\subsubsection{Modular Analysis}

The preceding results were obtained by analysis of a global model which was
constructed by embedding the internal dynamics of the program's functions
within the global dynamics of the Main function. In contrast, the idea in
\textit{modular analysis} is to model software as the interconnection of the
program's "building blocks" or "modules", i.e., functions that interact via a
set of \textit{global} variables. The dynamics of the functions are then
abstracted via Input/Output behavioral models, typically constituting equality
and/or inequality constraints relating the input and output variables. In our
framework, the invariant sets of the terminal nodes of the modules (e.g., the
set $X_{\Join}$ associated with the terminal node $\mathrm{F}_{\Join}$ in
Program 4) provide such I/O model. Thus, richer characterization of the
invariant sets of the terminal nodes of the modules are desirable. Correctness
of each sub-module must be established separately, while correctness of the
entire program will be established by verifying the unreachability and
termination properties w.r.t. the global variables, as well as verifying that
a terminal global state will be reached in finite-time. This way, the program
variables that are \textit{private} to each function are abstracted away from
the global dynamics. This approach has the potential to greatly simplify the
analysis and improve the scalability of the proposed framework as analysis of
large size computer programs is undertaken. In this section, we apply the
framework to modular analysis of Program 4. Detailed analysis of the
advantages in terms of improving scalability, and the limitations in terms of
conservatism the analysis is an important and interesting direction of future research.

The first step is to establish correctness of the \textsc{IntegerDivision}
module, for which we obtain\vspace*{-0.25in}%
\[
\sigma_{7\mathrm{F}2}\left(  \mathrm{dd,dr,q,r}\right)  =\left(
\mathrm{q}+\mathrm{r}\right)  ^{2}-\mathrm{M}^{2}\vspace*{-0.2in}%
\]
The function $\sigma_{7\mathrm{F}2}$ is a $\left(  1,0\right)  $-invariant
proving boundedness of the state variables of \textsc{IntegerDivision}.
Subject to boundedness, we obtain the function\vspace*{-0.25in}%
\[
\sigma_{8\mathrm{F}2}\left(  \mathrm{dd,dr,q,r}\right)  =2\mathrm{r}%
-11\mathrm{q}-6\mathrm{Z}\vspace*{-0.2in}%
\]
which is a $\left(  1,1\right)  $-invariant proving termination of
\textsc{IntegerDivision.} The invariant set of node $\mathrm{F}_{\Join}$ can
thus be characterized by\vspace*{-0.18in}%
\[
X_{\Join}=\left\{  \left(  \mathrm{dd,dr,q,r}\right)  \in\left[
0,\mathrm{M}\right]  ^{4}~|~\mathrm{r\leq dr-1}\right\}  \vspace*{-0.15in}%
\]
The next step is construction of a global model. Given $X_{\Join}$, the
assignment at $\mathrm{L}3$:\vspace*{-0.2in}
\[
\mathrm{L}3:\mathrm{rem=IntegerDivision~(X~,~Y)}\vspace*{-0.25in}%
\]
can be abstracted by\vspace*{-0.2in}
\[
\mathrm{rem=W,}\text{ s.t. }\mathrm{W}\in\left[  0,\mathrm{M}\right]
,\text{~}\mathrm{W}\leq\mathrm{Y-1,}\vspace*{-0.2in}%
\]
allowing for construction of a global model with variables $\mathrm{X,Y,~}$and
$\mathrm{rem,}$ and an external state-dependent input $\mathrm{W}\in\left[
0,\mathrm{M}\right]  \mathrm{,}$ satisfying $\mathrm{W}\leq\mathrm{Y-1.}$
Finally, the last step is analysis of the global model. We obtain the function
$\sigma_{9\mathrm{L}2}\left(  \mathrm{X,Y,rem}\right)  =\mathrm{Y}%
\times\mathrm{M}-\mathrm{M}^{2},$ which is $\left(  1,1\right)  $-invariant
proving both FTT and boundedness of all variables within $\left[
\mathrm{M},\mathrm{M}\right]  .$\vspace*{-0.2in}

\subsection{Filtering\label{sec:casestudy2}}

Program 4 has been adopted with some minor modifications from a similar
program analyzed by ASTREE \cite{ASTREE, ASTREE-Web}. The only modification that we have
made is the addition of the real-time input \textrm{w~}$\in\left[
-1,1\right]  $ in line $\mathrm{L}4.$ Note that when $\mathrm{B}=0$ (in line
$\mathrm{L}4$)\ we have the same program reported in \cite{ASTREE-Web}. In Program
4, the function $\mathrm{saturate(.)}$ truncates the real-time input
$\mathrm{{\small \ast PtrToInput}}$ so that \textrm{w~}$\in\left[
-1,1\right]  $ is guaranteed. We first build a graph model of this program.
The variables of the graph model are:

\textrm{Z, Y, E[0], E[1], S[0], S[1], INIT. }

The variable \textrm{INIT} is Boolean which we model by \textrm{v~}%
$\in\left\{  -1,1\right\}  $ in the following way:\vspace*{-0.15in}%
\[
\mathrm{INIT=True\Leftrightarrow~}\text{\textrm{v}}=1,\text{ and
}\mathrm{INIT=False\Leftrightarrow~}\text{\textrm{v}}=-1.\vspace*{-0.15in}%
\]

The graph model of this program is shown in Figure \ref{FigASTREE}.

\begin{gather*}%
\begin{tabular}
[c]{|l|}\hline%
\begin{tabular}
[c]{l}%
\begin{tabular}
[c]{ll}
& {\small \textrm{/* filter.c */\vspace*{-0.18in}}}\\
& {\small \textrm{typedef enum \{FALSE = 0, TRUE = 1\} BOOLEAN;\vspace
*{-0.18in}}}\\
& {\small \textrm{BOOLEAN INIT; {float Y=\textrm{\{0\}}, Z=\{0\};}%
\vspace*{-0.18in}}}\\
{\small \textrm{F}}${\small 0:\hspace*{0.04in}}$ & {\small \textrm{void filter
() \{\vspace*{-0.18in}}}\\
{\small \textrm{F}}${\small 1:\hspace*{0.04in}}$ & {\small \textrm{static
float E[2], S[2];\vspace*{-0.18in}}}\\
{\small \textrm{F}}${\small 2:\hspace*{0.04in}}$ & {\small \textrm{if (INIT)
\{\vspace*{-0.18in}}}\\
{\small \textrm{F}}${\small 3:\hspace*{0.04in}}$ & \qquad{\small \textrm{S[0]
= Z;\vspace*{-0.18in}}}\\
{\small \textrm{F}}${\small 4:\hspace*{0.04in}}$ & \qquad{\small \textrm{Y =
Z;\vspace*{-0.18in}}}\\
{\small \textrm{F}}${\small 5:\hspace*{0.04in}}$ & \qquad{\small \textrm{E[0]
= Z;\vspace*{-0.18in}}}\\
{\small \textrm{F}}${\small 6:\hspace*{0.04in}}$ & {\small \textrm{\} else
\{\vspace*{-0.18in}}}\\
{\small \textrm{F}}${\small 7:\hspace*{0.04in}}$ & \qquad{\small \textrm{Y =
(((((0.5*Z) - (E[0]*0.7)) + (E[1]*0.4)) + (S[0]*1.5)) - (S[1]*0.7));\vspace
*{-0.18in}}}\\
{\small \textrm{F}}${\small 8:\hspace*{0.04in}}$ & \qquad
{\small \textrm{\}\vspace*{-0.18in}}}\\
{\small \textrm{F}}${\small 9:\hspace*{0.04in}}$ & {\small \textrm{E[1] =
E[0];\vspace*{-0.18in}}}\\
{\small \textrm{F}}${\small 10:\hspace*{0.04in}}$ & {\small \textrm{E[0] =
Z;\vspace*{-0.18in}}}\\
{\small \textrm{F}}${\small 11:\hspace*{0.04in}}$ & {\small \textrm{S[1] =
S[0];\vspace*{-0.18in}}}\\
{\small \textrm{F}}${\small 12:\hspace*{0.04in}}$ & {\small \textrm{S[0] =
Y;\vspace*{-0.18in}}}\\
{\small \textrm{F}}${\small 13:\hspace*{0.04in}}$ & {\small \textrm{\}}}%
\end{tabular}
\\
$%
\begin{tabular}
[c]{ll}%
{\small \textrm{L}}${\small 0:\hspace*{0.1in}}$ & {\small \textrm{void main ()
\{\vspace*{-0.18in}}}\\
{\small \textrm{L}}${\small 1:\hspace*{0.1in}}$ & {\small \textrm{Z = 0.2 * Z
+ 5;\vspace*{-0.18in}}}\\
{\small \textrm{L}}${\small 2:\hspace*{0.1in}}$ & {\small \textrm{INIT =
TRUE;\vspace*{-0.18in}}}\\
{\small \textrm{L}}${\small 3:\hspace*{0.1in}}$ & {\small \textrm{while (1) \{
\vspace*{-0.18in}}}\\
& \qquad{\small \textrm{wait (0.001); w = saturate(}}$\mathrm{{\small \ast
PtrToInput);}}$ {\small \textrm{/*updates real-time input*/\vspace*{-0.18in}}%
}\\
{\small \textrm{L}}${\small 4:\hspace*{0.1in}}$ & \qquad{\small \textrm{Z =
0.9 * Z + 35+B*w;\vspace*{-0.18in}}}\\
{\small \textrm{L}}${\small 5:\hspace*{0.1in}}$ & \qquad{\small \textrm{filter
();\vspace*{-0.18in}}}\\
{\small \textrm{L}}${\small 6:\hspace*{0.1in}}$ & \qquad{\small \textrm{INIT =
FALSE;\vspace*{-0.18in}}}\\
{\small \textrm{L}}${\small 7:\hspace*{0.1in}}$ & \qquad
{\small \textrm{\}\vspace*{-0.18in}}}\\
{\small \textrm{L}}${\small \Join}{\small :\hspace*{0.1in}}$ &
{\small \textrm{\}}}%
\end{tabular}
$%
\end{tabular}
\\\hline
\end{tabular}
\\[0.2in]%
\begin{array}
[c]{c}%
\begin{array}
[c]{cccccc}
&  &  & \text{Program 4: Example of filtering in safety-critical software} &
&
\end{array}
\\
\text{Example adapted from reference \cite{ASTREE-Web} with minor
modifications.}\\
\end{array}
\end{gather*}

\begin{figure}
[ptb]
\begin{center}
\includegraphics[
height=3.00in,
width=4.00in
]%
{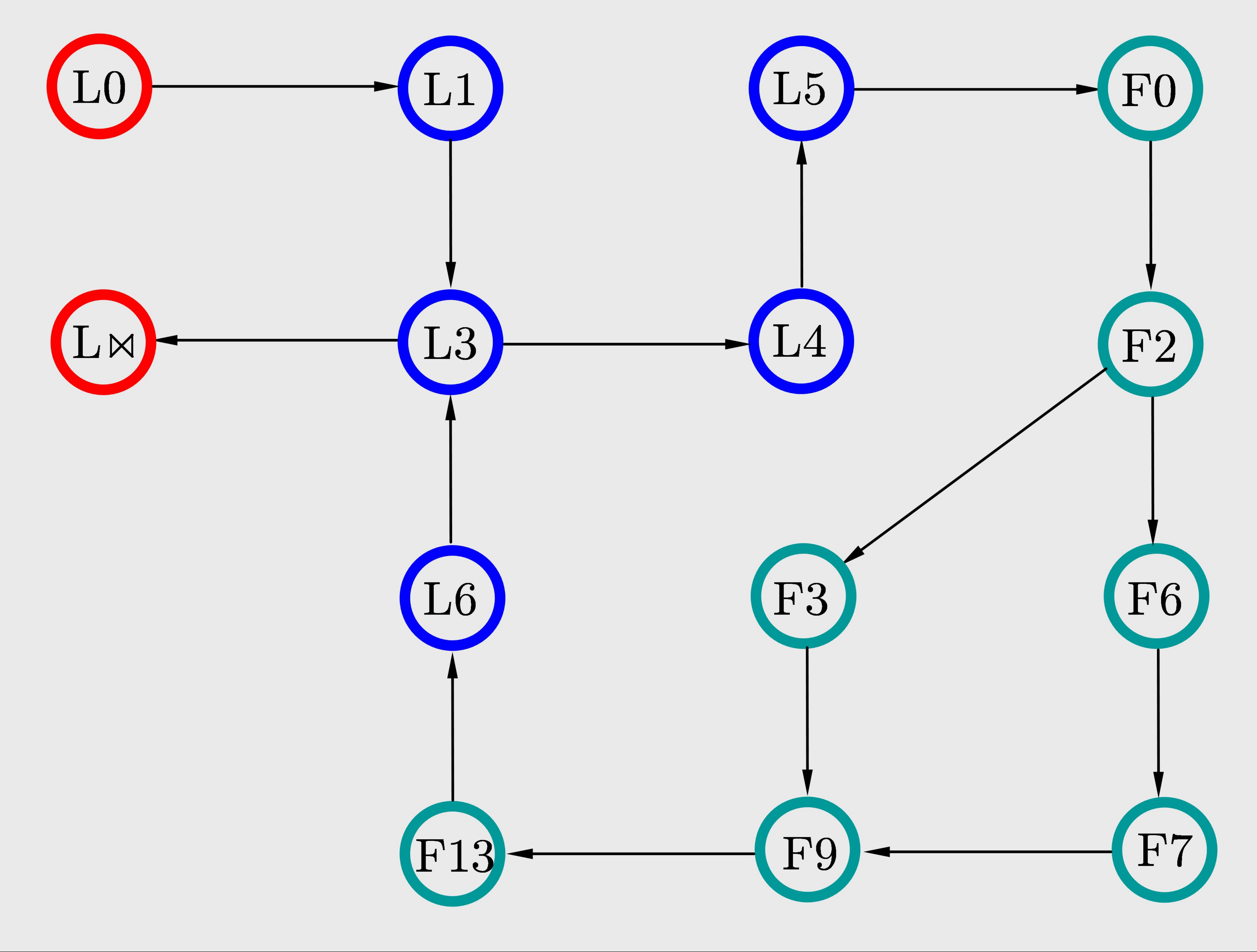}%
\caption{The graph of Program 4.}%
\label{FigASTREE}%
\end{center}
\end{figure}

\bigskip

In order to construct a more compact model, several lines of code
corresponding to consecutive transitions have been combined, and only the
first line corresponding to a series of consecutive transitions has been
assigned a node on the graph. Specifically, the combined lines are:
$\{\mathrm{L}1,$ $\mathrm{L}2\},$ and $\left\{  \mathrm{F}3,\text{ }%
\mathrm{F}4,\text{ }\mathrm{F}5\right\}  ,$ and $\left\{  \mathrm{F}9,\text{
}\mathrm{F}10,\text{ }\mathrm{F}11,\text{ }\mathrm{F}12\right\}  .$ The net
effect of the eliminated lines is captured in the model by taking the
composition of the functions that are coded at these lines.

Letting $\mathrm{x}=$ $\left[  \mathrm{Z,\ Y,\ E[0],\ E[1],\ S[0],\ S[1],\ v}%
\right]  ^{T},$ the state space of the graph model is $\mathcal{N}%
\times\mathbb{R}^{7},$ where $\mathcal{N=}\left\{  \mathrm{L}0,\mathrm{L}%
1,\mathrm{L}3,\mathrm{L}4,\mathrm{L}5,\mathrm{L}6,\mathrm{L}{\small \Join
}\right\}  \cup\left\{  \mathrm{F}0,\mathrm{F}2,\mathrm{F}3,\mathrm{F}%
6,\mathrm{F}7,\mathrm{F}9,\mathrm{F}13\right\}  ,$ as shown in Figure
\ref{FigASTREE}. The initial condition is fixed: $X_{\emptyset}:=\left\{
\mathrm{L}0\right\}  \times\mathrm{x}_{0}$ where $\mathrm{x}_{0}=\left[
0,0,0,0,0,0,1\right]  ^{T}.$ The corresponding passport and transition labels
are as follows:\vspace{-0.1in}%
\begin{align*}
&
\begin{array}
[c]{rllcrrll}%
T_{\mathrm{L}3\mathrm{L}1} & \hspace*{-0.06in}:\hspace*{-0.06in} &
\mathrm{x}\mapsto\left[  0.2\mathrm{Z}%
+5\mathrm{,\ Y,\ E[0],\ E[1],\ S[0],\ S[1],\ }1\right]   & \vspace*{-0.05in} &
&  &  & \\
T_{\mathrm{L}5\mathrm{L}4} & \hspace*{-0.06in}:\hspace*{-0.06in} &
\mathrm{x}\mapsto\left[  0.9\mathrm{Z}%
+35+\mathrm{Bw,\ Y,\ E[0],\ E[1],\ S[0],\ S[1],\ v}\right]   & \vspace
*{-0.05in} &  &  &  & \\
T_{\mathrm{L}3\mathrm{L}6} & \hspace*{-0.06in}:\hspace*{-0.06in} &
\mathrm{x}\mapsto\left[  \mathrm{Z,\ Y,\ E[0],\ E[1],\ S[0],\ S[1],\ -1}%
\right]   &  &  &  &  & \\
T_{\mathrm{F}9\mathrm{F}3} & \hspace*{-0.06in}:\hspace*{-0.06in} &
\mathrm{x}\mapsto\left[  \mathrm{Z,\ Z,\ Z,\ E[1],\ Z,\ S[1],\ v}\right]   &
&  &  &  & \\
T_{\mathrm{F}9\mathrm{F}7} & \hspace*{-0.06in}:\hspace*{-0.06in} &
\mathrm{x}\mapsto\left[  \mathrm{Z,\ }0.5\mathrm{Z}-0.7\mathrm{E[0]}%
+0.4\mathrm{E[1]}+1.5\mathrm{S[0]}%
-0.7\mathrm{S[1],\ E[0],\ E[1],\ S[0],\ S[1],\ v}\right]   &  &  &  &  & \\
T_{\mathrm{F}13\mathrm{F}9} & \hspace*{-0.06in}:\hspace*{-0.06in} &
\mathrm{x}\mapsto\left[  \mathrm{Z,\ Y,\ Z,\ E[0],\ Y,\ S[0],\ v}\right]   &
&  &  &  &
\end{array}
\\[-0.1in]
&
\begin{array}
[c]{rllcrrllllllllllll}%
\Pi_{\mathrm{F}3\mathrm{F}2} & \hspace*{-0.06in}=\hspace*{-0.06in} & \left\{
\mathrm{x~|~v=1}\right\}  , &  &  & \Pi_{\mathrm{F}6\mathrm{F}2} &
\hspace*{-0.06in}=\hspace*{-0.06in}\vspace*{-0.05in} & \left\{
\mathrm{x~|~v=-1}\right\}  , &  &  & \Pi_{\mathrm{L}4\mathrm{L}3} &
\hspace*{-0.06in}=\hspace*{-0.06in}\vspace*{-0.05in} & \mathbb{R}^{7}, &  &  &
\Pi_{\mathrm{L}\Join\mathrm{L}3} & = & \mathbb{\varnothing}.
\end{array}
\end{align*}

\smallskip

The other transitions are identity maps. Theorem \ref{Thm:GraphModelLMI} can
be applied for verification of absence of overflow (note that the program does
not terminate). We apply Theorem \ref{Thm:GraphModelLMI} for two different
values of $\mathrm{B},$ $\mathrm{B}=0$ which corresponds to no real-time
input, and $\mathrm{B}=20.$ For each value of $\mathrm{B} \in \{0,20\}$, absence of overflow is verified using two different methods.  In Method I, we set $\Lambda_{j}=\Lambda=\operatorname{diag}\left\{  \rm{M}^{-2}{\bf{1}}_7,-1\right\},$ and thus, we obtain a certificate for $\mathrm{M}\geq\left\Vert \mathrm{x}\right\Vert _{2}$ using one quadratic invariant. In Method II, for each $i \in \{1,..,7\}$ we set $\Lambda_{j}=\Lambda=\operatorname{diag}\left\{  {\rm{M}^{-2}}{e}_i,-1\right\},$ where $e_i$ is the $i$-th standard unit vector in $\mathbb{R}^7$ (cf. Corollary \ref{Bddness and FTT} and Equation \eqref{remarktos}). As a result, we obtain a much tighter upperbound in the form of $\mathrm{M}\geq\left\Vert \mathrm{x}\right\Vert _{\infty}$ at the expense of more computations. Table V summarizes the results of this analysis.
Analysis of the floating-point operations is based on the abstractions discussed in Appendix I. The corresponding Matlab codes can be found in \cite{MVichSite}. In this analysis, the solver SeDuMi version 1.3 \cite{Strum1999} is used for semidefinite optimization, and Yalmip \cite{Lofberg2004} for compilation. For comparison we point out that the static analyzer Astree reports a bound of $\rm{M}=1327.05$ for the $\mathrm{B}=0$ case \cite{ASTREE-Web}. \vspace*{-0.4in}

{\scriptsize
\begin{gather*}
\text{\textsc{TABLE\ V}}\\%
\begin{tabular}
[c]{|r|r|c|c|c|l|}\hline
& Computations: & infinite-precision & $64$-bit double precision & $32$-bit
single precision & \\\hline
\multicolumn{1}{|l|}{Method I} & \multicolumn{1}{|c|}{$\mathrm{B}=0$} &
$\mathrm{M}=884.95$ & $\mathrm{M}=884.96$ & $\mathrm{M}=884.96$ &
$\mathrm{M}\geq\left\Vert \mathrm{x}\right\Vert _{2}\geq\left\Vert
\mathrm{x}\right\Vert _{\infty}$\\\hline
\multicolumn{1}{|l|}{Method II} & \multicolumn{1}{|c|}{$\mathrm{B}=0$} &
$\mathrm{M}=373.11$ & $\mathrm{M}=373.12$ & $\mathrm{M}=373.12$ &
$\mathrm{M}\geq\left\Vert \mathrm{x}\right\Vert _{\infty}$\\\hline
\multicolumn{1}{|l|}{Method I} & \multicolumn{1}{|c|}{$\mathrm{B}=20$} &
\multicolumn{1}{|c|}{$\mathrm{M}=1400.95$} & \multicolumn{1}{|c|}{$\mathrm{M}%
=1402.81$} & $\mathrm{M}=1402.87$ & $\mathrm{M}\geq\left\Vert \mathrm{x}%
\right\Vert _{2}\geq\left\Vert \mathrm{x}\right\Vert _{\infty}$\\\hline
\multicolumn{1}{|l|}{Method II} & \multicolumn{1}{|c|}{$\mathrm{B}=20$} &
$\mathrm{M}=609.83$ & $\mathrm{M}=609.83$ & $\mathrm{M}=609.83$ &
$\mathrm{M}\geq\left\Vert \mathrm{x}\right\Vert _{\infty}$\\\hline
& $%
\begin{array}
[c]{c}%
\left(  \theta_{\mathrm{ij}},\mu_{\mathrm{ij}}\right)  \smallskip
\end{array}
\hspace*{-0.02in}$ & $\left(  0.98,0\right)  $ & $\left(  0.98,0\right)  $ &
$\left(  0.98,0\right)  $ & \\\hline
\end{tabular}
\end{gather*}
\vspace{-0.4in}
}

\section{Concluding Remarks\vspace*{-0.05in}}

We took a systems-theoretic approach to software analysis, and presented a
framework based on convex optimization of Lyapunov invariants for verification
of a range of important specifications for software systems, including
finite-time termination and absence of run-time errors such as overflow,
out-of-bounds array indexing, division-by-zero, and user-defined program
assertions. The verification problem is reduced to solving a numerical
optimization problem, which when feasible, results in a certificate for the
desired specification.\ The novelty of the framework, and consequently, the
main contributions of this paper are in the systematic transfer of Lyapunov
functions and the associated computational techniques from control systems to
software analysis. The presented work can be extended in several directions.
These include understanding the limitations of modular analysis of programs,
perturbation analysis of the Lyapunov certificates to quantify robustness with
respect to round-off errors, extension to systems with software in closed loop
with hardware, and adaptation of the framework to specific classes of
software.\vspace*{-0.1in}
\newpage
\section{Appendix I\vspace*{-0.1in}}

\subsection{Semialgebraic Set-Valued Abstractions of Commonly-Used
Nonlinearities\label{Sec:abstnon}\vspace*{-0.08in}}

\begin{enumerate}
\item[\textbf{--}] Trigonometric Functions:\vspace*{-0.05in}
\end{enumerate}

Abstraction of trigonometric functions can be obtained by approximation of the
Taylor series expansion followed by representation of the absolute error by a
static bounded uncertainty. For instance, an abstraction of the $\sin\left(
\cdot\right)  $ function can be constructed as follows:\vspace*{-0.1in}%
{\small
\[%
\begin{tabular}
[c]{|l|l|l|}\hline
Abstraction of $\sin\left(  x\right)  $ & $x\in\lbrack-\frac{\pi}{2},\frac
{\pi}{2}]$ & $x\in\lbrack-\pi,\pi]$\\\hline
$\overline{\sin}\left(  x\right)  \in\left\{  x+aw~|~w\in\left[  -1,1\right]
\right\}  $ & $a=0.571$ & $a=3.142$\\\hline
$\overline{\sin}\left(  x\right)  \in\{x-\frac{1}{6}x^{3}+aw~|~w\in\left[
-1,1\right]  \}$ & $a=0.076$ & $a=2.027$\\\hline
\end{tabular}
\]
} Abstraction of $\cos\left(  \cdot\right)  $ is similar. It is also possible
to obtain piecewise linear abstractions by first approximating the function by
a piece-wise linear (PWL) function and then representing the absolute error by
a bounded uncertainty. Section \ref{Section:SpecModels} (Proposition
\ref{MILM-prop}) establishes universality of representation of generic PWL
functions via binary and continuous variables and an algorithmic construction
is given in the proof of the proposition in Appendix II. For instance, if
$x\in\lbrack0,\pi/2]$ then a piecewise linear approximation with absolute
error less than $0.06$ can be constructed in the following way:\vspace
*{-0.15in}
\begin{subequations}
\label{SinAbst}%
\begin{align}
\hspace*{-0.1in}S\hspace{-0.02in} &  =\hspace{-0.02in}\left\{  \left(
x,v,w\right)  \hspace*{0.02in}|\hspace*{0.02in}x=0.2\left[  \left(
1+v\right)  \left(  1+w_{2}\right)  +\left(  1-v\right)  \left(
3+w_{2}\right)  \right]  ,\hspace*{0.02in}\left(  w,v\right)  \in\left[
-1,1\right]  ^{2}\times\left\{  -1,1\right\}  \right\} \label{SinAbsta}\\
\hspace*{-0.1in}\overline{\sin}\left(  x\right)  \hspace{-0.02in} &
\in\hspace{-0.02in}\left\{  Tx_{E}~|~x_{E}\in S\right\}  ,\text{\ }%
T:x_{E}\mapsto0.45\left(  1+v\right)  x+\left(  1-v\right)  \left(
0.2x+0.2\right)  +0.06w_{1}%
\end{align}

\end{subequations}
\begin{enumerate}
\item[\textbf{--}] The Sign Function ($\operatorname{sgn}$) and the Absolute
Value Function ($\operatorname{abs}$):\vspace*{-0.05in}
\end{enumerate}

The sign function ($\operatorname{sgn}(x)=1\mathbb{I}_{[0,\infty)}\left(
x\right)  -1\mathbb{I}_{(-\infty,0)}\left(  x\right)  $) may appear explicitly
as one of the program's functions, or implicitly as a model for conditional
switching. A particular abstraction of $\operatorname{sgn}\left(
\cdot\right)  $ is as follows: $\overline{\operatorname{sgn}}(x)\in\left\{
v~|~xv\geq0,\text{ }v\in\left\{  -1,1\right\}  \right\}  $. Note that
$\operatorname{sgn}\left(  0\right)  $ is equal to $1,$ while the abstraction
is multi-valued at zero$:$ $\overline{\operatorname{sgn}}\left(  0\right)
\in\left\{  -1,1\right\}  .$ The absolute value function can be represented
(precisely) over $\left[  -1,1\right]  $ in the following way:\vspace
*{-0.17in}
\[
\operatorname{abs}\left(  x\right)  =\left\{  xv~|~x=0.5\left(  v+w\right)
,\text{ }\left(  w,v\right)  \in\left[  -1,1\right]  \times\left\{
-1,1\right\}  \right\}  \vspace*{-0.25in}%
\]

\begin{enumerate}
\item[\textbf{--}] Modulo Arithmetic:\vspace*{-0.05in}
\end{enumerate}

Consider the function $\operatorname{mod}:\mathbb{Z}\times\mathbb{Z}%
\rightarrow\mathbb{Z}$ defined in the following way:\vspace*{-0.2in}%
\[
\operatorname{mod}\left(  t,s\right)  =t-ns,\text{ where }n=\lfloor\frac{t}%
{s}\rfloor.\vspace*{-0.2in}%
\]
\vspace*{0.1in}Abstraction of the function $\operatorname{mod}\left(
.,.\right)  $ over bounded subsets of the set of integers is as follows. Let
$t\in\lbrack0,Ms)$, where $M>0$ is an integer. Then:\vspace*{-0.1in}%
\begin{equation}
\operatorname{mod}\left(  t,s\right)  \in\left\{  t-\frac{1}{2}s\sum
\limits_{k=1}^{M-1}\left(  1+v_{k}\right)  ~|~(t-ks)v_{k}\geq0,~v_{k}%
\in\left\{  -1,1\right\}  ,~k=1,...,M-1\right\}  .\vspace*{-0.1in}\label{modA}%
\end{equation}
It follows from (\ref{modA}) that the result of summation modulo $s$ of two
integers $t_{1},t_{2}\in\lbrack0,Ms)$ can be abstracted in the following
way:\vspace*{-0.15in}%
\[
\operatorname{mod}\left(  t_{1}+t_{2},s\right)  \in\left\{  t_{1}+t_{2}%
-\frac{1}{2}s\sum\limits_{k=1}^{2M-1}\left(  1+v_{k}\right)  ~|~(t_{1}%
+t_{2}-ks)v_{k}\geq0,~v_{k}\in\left\{  -1,1\right\}  ,~k=1,...,2M-1\right\}
.\vspace*{-0.2in}%
\]
\vspace*{-0.40in}

\subsection{Floating-Point and Fixed-Point Arithmetic}
This subsection is based on \cite{MineThesis}, Chapter 7. Interested readers are referred to \cite{MineThesis, MinePaper}
for a more comprehensive discussion of computations with floats.
For computations with floating-point numbers, the IEEE 754-1985 norm has
become the hardware standard, and is supported by most popular programming
languages such as C. In this standard, a \textit{float} number is represented
by a triplet $\left(  s,f,e\right)  ,$ where:$\vspace*{-0.1in}$

\begin{itemize}
\item $s\in\left\{  0,1\right\}  $ is the sign bit.$\vspace*{-0.05in}$

\item $f$ is the fractional part, represented by a $p$-bit unsigned integer:
$f_{1}\ldots f_{p},$ $f_{i}\in\{0,1\}.\vspace*{-0.05in}$

\item $e$ is the biased exponent, represented by a $q$-bit unsigned integer:
$e_{1}\ldots e_{q},$ $e_{i}\in\{0,1\}.\vspace*{-0.05in}$
\end{itemize}

A floating point number $z=\left(  s,f,e\right)  $ is then in one of the
following forms:

\begin{itemize}
\item $z=\left(  -1\right)  ^{s}\times2^{e-bias}\times1.f,$ when $1\leq
e\leq2^{q}-2.\vspace*{-0.05in}$

\item $z=\left(  -1\right)  ^{s}\times2^{1-bias}\times0.f,$ when $e=0,$ and
$f\neq0.\vspace*{-0.05in}$

\item $z=+0$ (when $s=1$) and $z=-0$ (when $s=0$) and $e=f=0.\vspace
*{-0.05in}$

\item $z=+\infty$ (when $s=1$) and $z=-\infty$ (when $s=0$) and $e=2^{q}-1, $
$f\neq0.\vspace*{-0.05in}$

\item $z=$ $NaN$ when $e=2^{q}-1,$ $f=0.\vspace*{-0.05in}$
\end{itemize}

The values of $p,$ $q,$ and $bias$ depend on the specific format:

\begin{itemize}
\item If format is 32-bit single precision (f=32), then $bias=127,$ $q=8,$ and
$p=23.\vspace*{-0.05in}$

\item If format is 64-bit double precision (f=64), then $bias=1023,$ $q=11,$
and $p=52.\vspace*{-0.05in}$
\end{itemize}

Other formatting standards such as \textit{long double} or \textit{quadruple}
precision also exist. In floating point computations, the result of performing
the arithmetic operation $\circledast\in\left\{  +,-,\times,\div\right\}  $ on
two \textit{float} variables $x$ and $y$ is stored in a \textit{float}
variable $z:=\operatorname{float}\left(  x\circledast y,\text{f}\right)  $
which is a complicated function of $x,$ $y,$ and the format f. Examples of the
format f include the IEEE 754 with 32-bit single precision, or IEEE 754 with
64-bit double precision$.$ A floating-point operation is equivalent to
performing the operation over the set of real numbers followed by rounding the
result to a \textit{float }\cite{MinePaper, MineThesis}. In IEEE 754 the
possible rounding modes are rounding towards $0$, towards $+\infty$, towards
$-\infty$, and to the nearest (n). The rounding function $\Gamma_{\text{f,m}%
}:\mathbb{R\rightarrow F\cup}\left\{  \epsilon\right\}  $ maps a real number
to a float number or to runtime error $\epsilon$, depending on the format `f'
and the rounding mode `m'. We refer the reader to \cite{MineThesis} for more
details on the rounding function. For our purposes, it is sufficient to say
that for all rounding modes, the following relation holds:\vspace{-0.15in}%
\[
\forall x\in\lbrack-\alpha_{\text{f}},\alpha_{\text{f}}]:\left\vert
\Gamma_{\text{f,m}}\left(  x\right)  -x\right\vert \leq\gamma_{\text{f}%
}\left\vert x\right\vert +\beta_{\text{f}}\vspace*{-0.15in}%
\]
\vspace*{0.1in}where $\gamma_{\text{f}}:=2^{-p},$ and $\alpha_{\text{f}%
}:=\left(  2-2^{-p}\right)  2^{2^{q}-bias-2}$ is the largest non-infinite
number, and $\beta_{\text{f}}:=2^{1-bias-p}$ is the smallest non-zero positive number.

Based on the above discussions, an abstraction of the floating-point
arithmetic operators can be constructed in the following way:\vspace{-0.1in}%
\[%
\begin{array}
[c]{c}%
x+y\in\lbrack-\alpha_{\text{f}},\alpha_{\text{f}}]\Rightarrow
\operatorname{float}\left(  x+y\right)  =z\in\left\{  x+y+\delta
w~|~w\in\left[  -1,1\right]  ,\text{ }\delta=\gamma_{\text{f}}\left(
\left\vert x\right\vert +\left\vert y\right\vert \right)  +\beta_{\text{f}%
}\right\} \\
x-y\in\lbrack-\alpha_{\text{f}},\alpha_{\text{f}}]\Rightarrow
\operatorname{float}\left(  x-y\right)  =z\in\left\{  x-y+\delta
w~|~w\in\left[  -1,1\right]  ,\text{ }\delta=\gamma_{\text{f}}\left(
\left\vert x\right\vert +\left\vert y\right\vert \right)  +\beta_{\text{f}%
}\right\} \\
x\times y\in\lbrack-\alpha_{\text{f}},\alpha_{\text{f}}]\Rightarrow
\operatorname{float}\left(  x\times y\right)  =z\in\left\{  x\times y+\delta
w~|~w\in\left[  -1,1\right]  ,\text{ }\delta=\gamma_{\text{f}}\left(
\left\vert x\right\vert \times\left\vert y\right\vert \right)  +\beta
_{\text{f}}\right\} \\
x\div y\in\lbrack-\alpha_{\text{f}},\alpha_{\text{f}}]\Rightarrow
\operatorname{float}\left(  x\div y\right)  =z\in\left\{  x\div y+\delta
w~|~w\in\left[  -1,1\right]  ,\text{ }\delta=\gamma_{\text{f}}\left(
\left\vert x\right\vert \div\left\vert y\right\vert \right)  +\beta_{\text{f}%
}\right\}
\end{array}
\]
where the constants $\alpha_{\text{f }},$ $\beta_{\text{f}}$ and
$\gamma_{\text{f}}$ are defined as before. The above abstractions can still be
complicated as the magnitude of $\delta$ depends on the operands $x$ and $y. $
If all of the program variables (including the result of the arithmetic
operation) reside in $\left[  -\alpha,\alpha\right]  $ where $\alpha\ll
\alpha_{\text{f }}$ then a simpler but more conservative abstraction can be
constructed in the following way:$\vspace*{-0.2in}$%
\[
x\circledast y\in\lbrack-\alpha,\alpha]\Rightarrow\operatorname{float}\left(
x\circledast y\right)  =z\in\left\{  x\circledast y+\delta w~|~w\in\left[
-1,1\right]  ,\text{ }\delta=\alpha\gamma_{\text{f}}+\beta_{\text{f}}\right\}
\vspace*{-0.15in}%
\]
For instance, if f=32 and $\alpha=10^{6},$ then $\delta=0.12,$
and if f=64 and $\alpha=10^{10},$ then $\delta=2.3\times10^{-6}.$ Abstractions
of arithmetic operations in fixed-point computations is similar to the above.
The magnitude of $\delta$ will depend on the number of bits and the dynamic
range. For instance, in the two's complement format we have: $\delta
=\rho\left(  2^{b}-1\right)  ^{-1},$ where $\rho$ is the dynamic range and $b$
is the number of bits. A case study in software verification based on the above discussed abstractions of floating-point computations has been presented in Section \ref{sec:casestudy2}.

\bigskip

\subsection{Graph models with state-dependent or time-varying edges}

Consider the following fragment from a program with two variables:
$x\in\mathbb{R},$ and $i\in\left\{  1,2\right\}  :\vspace*{-0.1in}$%
\[%
\begin{array}
[c]{c}%
\begin{tabular}
[c]{|l|}\hline
$%
\begin{array}
[c]{cl}%
\text{\textrm{L}}1: & x=A\left[  i\right]  x;\\
\text{\textrm{L}}2: & \text{expression}%
\end{array}
$\\\hline
\end{tabular}
\vspace*{0.1in}%
\end{array}
\]
Then, the state transition from node $1$ to $2$ is defined by $T_{21}:\left(
x,i\right)  \rightarrow\left(  A\left[  ~i~\right]  x,i\right)  ,$ where $A$
is an array of size $2$. An algorithm similar to the one used in construction
of MILMs (cf. Proposition \ref{MILM-prop}) can be used to construct a
transition label using semi-algebraic set-valued maps. For the above example,
this can be done in the following way:$\vspace*{-0.15in}$%
\[%
\begin{array}
[c]{rll}%
\overline{T}_{21}\left(  x,i\right)  & = & \{T_{21}\left(  x,i,v_{1}%
,v_{2}\right)  ~|~\left(  x,i,v_{1},v_{2}\right)  \in S_{21}\}.\\
T_{21}\left(  x,i,v_{1},v_{2}\right)  & = & A\left[  1\right]  v_{1}x+A\left[
2\right]  v_{2}x.\\
S_{21} & = & \left\{  \left(  x,i,v_{1},v_{2}\right)  ~|~v_{1}+v_{2}=1,\text{
}v_{1},v_{2}\in\left\{  0,1\right\}  ,\text{ }v_{1}\left(  i-1\right)
+v_{2}\left(  i-2\right)  =0\right\}  .
\end{array}
\]
In light of Proposition \ref{MILM-prop}, generalization of this technique to
multidimensional arrays with arbitrary finite size is straightforward.
However, the number of binary decision variables grows (linearly) with the
number of elements in the array, which can be undesirable for large arrays. An
alternative approach to constructing graph models with fixed labels is to
attempt to derive a fixed map by computing the net effect of several lines of
code \cite{FeronRooz07}. When applicable, the result is a fixed map which is
obtained by taking the composition of several functions. This is an instance
of a more general concept in computer science, namely, extracting the higher
level semantics, which allows one to define more compact models of the software.

For instance, consider the following code fragment:
\begin{figure}
[ptb]
\begin{center}
\includegraphics[
height=3in,
width=4.65in
]%
{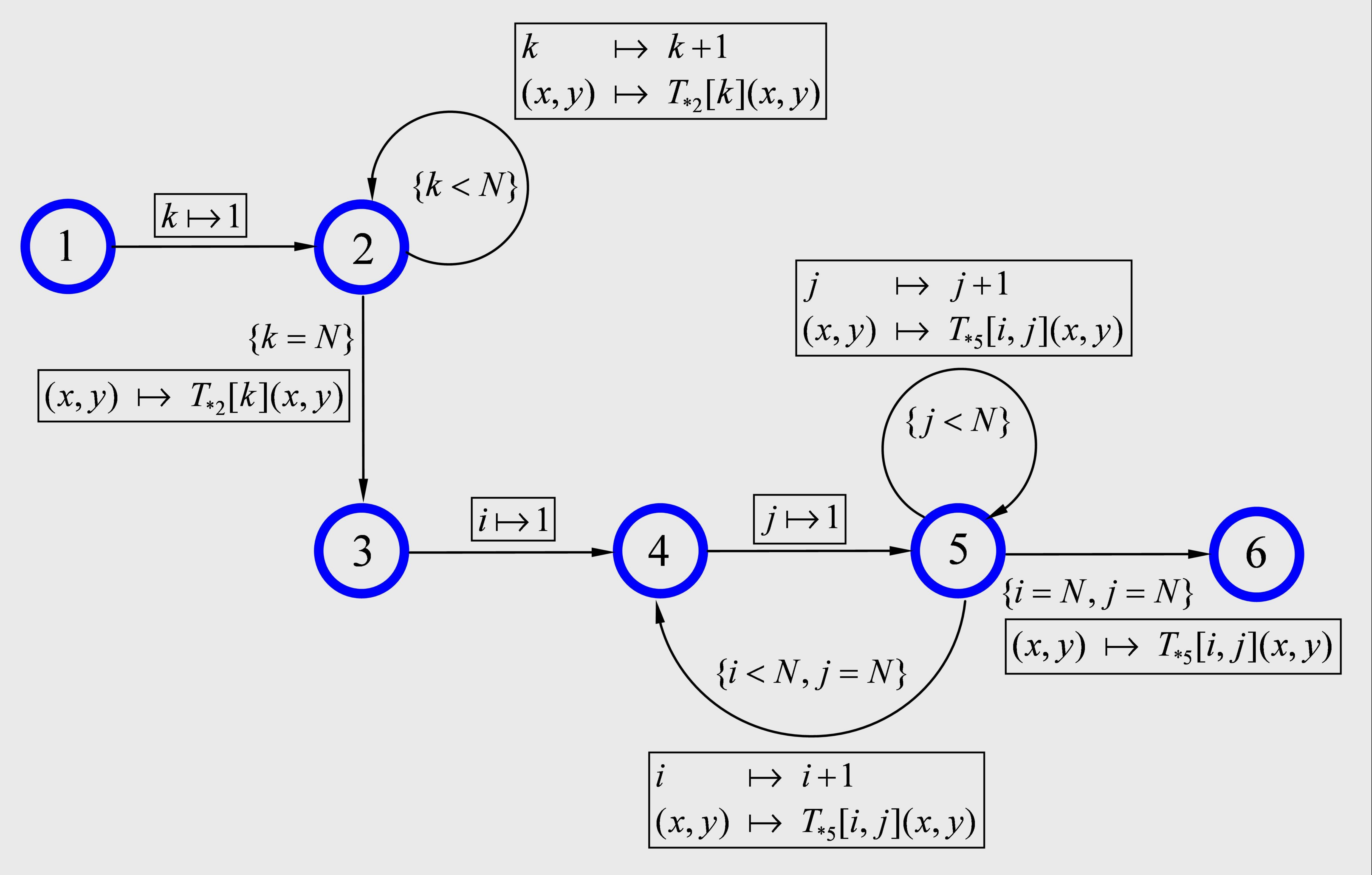}%
\caption{Graph model of a code fragment (Program 5) with state-dependent
labels. The transition labels are shown in boxes and the passport labels are
in brackets. For simplicity, only the non-identity transition labels are
shown.}%
\label{FoorLoops2}%
\end{center}
\end{figure}
\begin{gather*}%
\begin{tabular}
[c]{|l|}\hline
$%
\begin{array}
[c]{ll}%
\begin{array}
[c]{c}%
\mathrm{L}1:
\end{array}
&
\begin{array}
[c]{c}%
\text{{\small for ( k = 1 ; k == N ; k++)\qquad\{}}%
\end{array}
\vspace*{-0.15in}\\%
\begin{array}
[c]{c}%
\mathrm{L}2:
\end{array}
&
\begin{array}
[c]{c}%
\text{{\small y[ k ] = x[ k ]; x[ k ] = 0; \}}}%
\end{array}
\vspace*{-0.15in}\\%
\begin{array}
[c]{c}%
\mathrm{L}3:
\end{array}
&
\begin{array}
[c]{c}%
\text{{\small for ( i = 1 ; i == N ; i++)\qquad\{}}%
\end{array}
\vspace*{-0.15in}\\%
\begin{array}
[c]{c}%
\mathrm{L}4:
\end{array}
&
\begin{array}
[c]{c}%
\text{\qquad{\small for ( j = 1 ; j == N ; j++)\qquad\{}}%
\end{array}
\vspace*{-0.15in}\\%
\begin{array}
[c]{c}%
\mathrm{L}5:
\end{array}
&
\begin{array}
[c]{c}%
\qquad\text{\qquad{\small x[ i ] = x[ i ] + y[ j ] * A[ i ][ j ];}}%
\end{array}
\vspace*{-0.15in}\\%
\begin{array}
[c]{c}%
\mathrm{L}6:
\end{array}
&
\begin{array}
[c]{c}%
\text{{\small \}\}}}%
\end{array}
\end{array}
$\\\hline
\end{tabular}
\\[0.1in]
\text{Program 5: A code fragment from a linear filtering application}%
\end{gather*}
A graph model of this code fragment is shown in Figure \ref{FoorLoops2},
where:$\vspace*{-0.15in}$%
\begin{align*}
y  & =\left[  y\left[  1\right]  \cdots y\left[  N\right]  \right]
^{T},\text{ }x=\left[  x\left[  1\right]  \cdots x\left[  N\right]  \right]
^{T}\\[0.25in]
T_{\ast2}\left[  k\right]   & :\left[
\begin{array}
[c]{c}%
\vspace*{-0.5in}\\
y\\
\vspace*{-0.5in}\\
x\\
\vspace*{-0.4in}%
\end{array}
\right]  \rightarrow\left[
\begin{array}
[c]{cc}%
\vspace*{-0.5in} & \\
I-e_{kk} & e_{kk}\\
\vspace*{-0.5in} & \\
0 & I-e_{kk}\\
\vspace*{-0.4in} &
\end{array}
\right]  \left[
\begin{array}
[c]{c}%
\vspace*{-0.5in}\\
y\\
\vspace*{-0.5in}\\
x\\
\vspace*{-0.4in}%
\end{array}
\right]  ,\\[0.25in]
T_{\ast5}\left[  i,j\right]   & :\left[
\begin{array}
[c]{c}%
\vspace*{-0.5in}\\
y\\
\vspace*{-0.5in}\\
x\\
\vspace*{-0.4in}%
\end{array}
\right]  \rightarrow\left[
\begin{array}
[c]{cc}%
\vspace*{-0.5in} & \\
I & 0\\
\vspace*{-0.5in} & \\
A_{ij}e_{ij} & I\\
\vspace*{-0.4in} &
\end{array}
\right]  \left[
\begin{array}
[c]{c}%
\vspace*{-0.5in}\\
y\\
\vspace*{-0.5in}\\
x\\
\vspace*{-0.4in}%
\end{array}
\right]  ,
\end{align*}
where $e_{ij}$ is an $N\times N$ matrix which is zero everywhere except at the
$\left(  i,j\right)  $-th entry which is $1,$ and $A_{ij}$ denotes the
$\left(  i,j\right)  $-th entry of $A$ ($A_{ij}\equiv A\left[  i\right]
\left[  j\right]  $) which is an $N\times N$ array in the code
\cite{FeronRooz07}. The remaining transition labels correspond to the counter
variables $i,~j,~k,$ and have been specified on the graph in Figure
\ref{FoorLoops2}. From node $1 $ to node $3,$ the net effect is:$\vspace
*{-0.1in}$
\[%
{\textstyle\prod\limits_{k=1}^{N}}
T_{\ast2}\left[  k\right]  =\left[
\begin{array}
[c]{cc}%
\vspace*{-0.5in} & \\
0 & I\\
\vspace*{-0.5in} & \\
0 & 0\\
\vspace*{-0.45in} &
\end{array}
\right]  ,
\]
and from node $3$ to node $6$ the net effect is:$\vspace*{-0.1in}$%
\[%
{\textstyle\prod\limits_{i=1}^{N}}
{\textstyle\prod\limits_{j=1}^{N}}
T_{\ast5}\left[  i,j\right]  =\left[
\begin{array}
[c]{cc}%
\vspace*{-0.5in} & \\
I & 0\\
\vspace*{-0.5in} & \\
A & 0\\
\vspace*{-0.4in} &
\end{array}
\right]  .
\]

\section{Appendix II}

\smallskip

\begin{proof}
[Proof of Proposition \ref{MILM-prop}]The proof is by construction. Let $X=%
{\textstyle\bigcup\limits_{i=1}^{N}}
X_{i},$ where the $X_{i}$'s are compact polytopic sets. Further, assume that
each $X_{i}$ is characterized by a finite set of linear inequality
constraints: $X_{i}=\left\{  x~|~\mathbf{S}_{i}x\leq\mathbf{s}_{i},\text{
}\mathbf{S}_{i}\in\mathbb{R}^{N_{i}\times n},\text{ }\mathbf{s}_{i}%
\in\mathbb{R}^{N_{i}}\right\}  .$ Let $v\in\left\{  -1,1\right\}  ^{N},$ and
consider the following sets:\vspace*{-0.1in}%
\begin{align*}
G_{xv} &  \triangleq\left\{  \left(  x,v\right)  ~|~\sum_{i=1}^{N}%
v_{i}=-N+2,\text{ }\left(  \mathbf{S}_{i}x-\mathbf{s}_{i}\right)  \left(
v_{i}+1\right)  \leq0,\text{ }v_{i}\in\left\{  -1,1\right\}  :i\in
\mathbb{Z}\left(  1,N\right)  \right\}  ,\\
G_{xy} &  \triangleq\left\{  \left(  x,y\right)  ~|~\sum_{i=1}^{N}\left(
1+v_{i}\right)  \left(  A_{i}x+B_{i}\right)  =y,\text{ }\left(  x,v\right)
\in G_{xv}\right\}  .\vspace*{-0.1in}%
\end{align*}
First, we prove that $G_{xy}=G_{1}.$ Define $N$ binary vectors: $\eta^{j}%
\in\left\{  -1,1\right\}  ^{N},$ $j\in\mathbb{Z}\left(  1,N\right)  ,$
according to the following rule: $\eta_{i}^{j}=1$ if and only if $i=j.$ Also
define $\mathbb{I}\left(  x\right)  \overset{\text{def}}{=}\left\{
j\in\mathbb{Z}\left(  1,N\right)  ~|~x\in X_{j}\right\}  .$ Now, let $\left(
x_{0},f\left(  x_{0}\right)  \right)  \in G_{1}.$ Then\vspace*{-0.1in}%
\[
\left(  x_{0},\eta^{j}\right)  \in G_{xv}:\forall j\in\mathbb{I}\left(
x_{0}\right)  .\vspace*{-0.1in}%
\]
Therefore, $\left(  x_{0,}2A_{j}x_{0}+2B_{j}\right)  \in G_{xy}:\forall
j\in\mathbb{I}\left(  x_{0}\right)  ,$ which by supposition, implies that
$\left(  x_{0},f\left(  x_{0}\right)  \right)  \in G_{xy}.$ This proves that
$G_{1}\subseteq G_{xy}.$ Now, let $\left(  x_{0},y_{0}\right)  \in G_{xy}.$
Then, there exists $\widehat{v}\in\left\{  -1,1\right\}  ^{N},$ such that
$\left(  x_{0},\widehat{v}\right)  \in G_{xv}.$ Hence, there exists
$j\in\mathbb{Z}\left(  1,N\right)  ,$ such that $x_{0}\in X_{j},$ and
$y_{0}=2A_{j}x+2B_{j}.$ It follows from the definition of $f$ that $\left(
x_{0},y_{0}\right)  \in G_{1}.$ This proves that $G_{xy}\subseteq G_{1}.$ We
have shown that $G_{1}=G_{xy}.$ It remains to show that $G_{xy}$ has a
representation equivalent to $G_{2}.$ \newline Define new variables
$u_{i}=xv_{i}\in\left[  -1,1\right]  ^{n}.$ Then\vspace*{-0.05in}%
\begin{equation}
G_{xy}\triangleq\left\{
\begin{array}
[c]{r}%
\left(  x,y\right)  ~|~y=\sum_{i=1}^{N}\left(  A_{i}x+A_{i}u_{i}+B_{i}%
+B_{i}v_{i}\right)  ,\text{ }\mathbf{S}_{i}u_{i}+\mathbf{S}_{i}x-\mathbf{s}%
_{i}v_{i}-\mathbf{s}_{i}\leq0\\
\sum_{i=1}^{N}v_{i}=-N+2,\text{ }u_{i}=xv_{i},\text{ }v_{i}\in\left\{
-1,1\right\}  :i\in\mathbb{Z}\left(  1,N\right)
\end{array}
\right\}  .\vspace*{-0.05in}\label{Conditional Polytopic SS}%
\end{equation}
The nonlinear map $\left(  x,v_{i}\right)  \mapsto u_{i}$ can be represented
by an affine transformation involving auxiliary variables $\underline{z}%
_{i}\in\left[  -1,1\right]  ^{n},$ and $\overline{z}_{i}\in\left[
-1,1\right]  ^{n},$ subject to a set of linear constraints, in the following
way:$\vspace*{-0.15in}$%
\begin{equation}
u_{i}=2\underline{z}_{i}-x-v_{i}\boldsymbol{1}_{n}+\boldsymbol{1}%
_{n},\text{\qquad}\underline{z}_{i}\leq v_{i}\boldsymbol{1}_{n},\qquad
\overline{z}_{i}\leq-v_{i}\boldsymbol{1}_{n},\qquad\underline{z}%
_{i}=x-\overline{z}_{i}-\boldsymbol{1}_{n}\vspace*{-0.2in}\label{ZMat}%
\end{equation}
equivalently:$\vspace*{-0.2in}$%
\[
u_{i}=2\underline{z}_{i}-x-\left(  v_{i}-1\right)  \boldsymbol{1}%
_{n},~~\underline{z}_{i}=x-\overline{z}_{i}-\boldsymbol{1}_{n},~~\underline
{z}_{i}=\left(  v_{i}-1\right)  \boldsymbol{1}_{n}+\underline{w}%
_{i},~~\overline{z}_{i}=\left(  -v_{i}-1\right)  \boldsymbol{1}_{n}%
+\overline{w}_{i},\vspace*{-0.15in}%
\]
where $\underline{w}_{i},\overline{w}_{i}\in\left[  -1,1\right]  ^{n}.$ Since
by assumption each $X_{i}$ is bounded, for all $i\in\mathbb{Z}\left(
1,N\right)  $ and all $j\in\mathbb{Z}\left(  1,N_{i}\right)  ,$ the
quantities$\vspace*{-0.15in}$%
\begin{equation}
\underline{R}_{ij}=\min_{x\in X_{i}}\mathbf{S}_{ij}x-\mathbf{s}_{ij}%
\vspace*{-0.15in}\label{LP}%
\end{equation}
exist, are finite and can be computed by solving $N_{i}\times N$ linear
programs given in (\ref{LP}) ($\mathbf{S}_{ij}$ and $\mathbf{s}_{ij}$ denote
the $j$-th row of $\mathbf{S}_{i}$ and $\mathbf{s}_{i}$ respectively). Let
$\underline{R}_{i}=\underset{j\in\mathbb{Z}\left(  1,N_{i}\right)
}{\operatorname*{diag}}\{\underline{R}_{ij}\}.$ Then $G_{xy}$ as defined in
(\ref{Conditional Polytopic SS}) is equivalent to:$\vspace*{-0.1in}$%
\begin{equation}
G_{xy}\triangleq\left\{
\begin{array}
[c]{r}%
\left(  x,y\right)  ~|~y=\sum_{i=1}^{N}\left(  A_{i}x+A_{i}u_{i}+B_{i}%
+B_{i}v_{i}\right)  ,\text{ }\left(  x,u_{i},v_{i}\right)  \in\mathbf{H}%
\end{array}
\right\}  \label{FMat}%
\end{equation}
where $\mathbf{H}\mathbf{=}\{\left(  x,u_{i},v_{i}\right)  ~|~x,u_{i},v_{i}$
satisfy (\ref{HMat}) for some $\underline{z}_{i},\overline{z}_{i}%
,\underline{w}_{i},\overline{w}_{i}\in\left[  -1,1\right]  ^{n},$ and
$w_{i}\in\left[  -1,1\right]  ^{N_{i}}\}:$%
\begin{equation}
\hspace{-0.06in}\mathbf{H}\mathbf{=}\left\{  \hspace{-0.06in}%
\begin{array}
[c]{rcllrll}%
0\hspace{-0.06in} & \hspace{-0.06in}=\hspace{-0.06in} & \hspace{-0.06in}%
\mathbf{S}_{i}u_{i}+\mathbf{S}_{i}x-\mathbf{s}_{i}v_{i}-\mathbf{s}%
_{i}-\underline{R}_{i}\left(  w_{i}+\mathbf{1}_{N_{i}}\right)  \vspace
*{-0.05in} &  & \underline{z}_{i} & \hspace{-0.06in}=\hspace{-0.06in} &
\hspace{-0.06in}x-\overline{z}_{i}-\boldsymbol{1}_{n},\text{ \ }\underline
{z}_{i},\overline{z}_{i}\in\left[  -1,1\right]  ^{n}\\
u_{i}\hspace{-0.06in} & \hspace{-0.06in}=\hspace{-0.06in} & \hspace
{-0.06in}2\underline{z}_{i}-x-\left(  v_{i}-1\right)  \boldsymbol{1}%
_{n}\vspace*{-0.05in} &  & \overline{z}_{i} & \hspace{-0.06in}=\hspace
{-0.06in} & \hspace{-0.06in}\left(  -v_{i}-1\right)  \boldsymbol{1}%
_{n}+\overline{w}_{i},\\
\underline{z}_{i}\hspace{-0.06in} & \hspace{-0.06in}=\hspace{-0.06in} &
\hspace{-0.06in}\left(  v_{i}-1\right)  \boldsymbol{1}_{n}+\underline{w}_{i} &
& \sum_{i=1}^{N}v_{i} & \hspace{-0.06in}=\hspace{-0.06in} & \hspace
{-0.06in}-N+2,\text{ }v_{i}\in\left\{  -1,1\right\}
\end{array}
\hspace{-0.06in}\right\}  \hspace{-0.06in}\label{HMat}%
\end{equation}
The equality constraints that define the Matrix $H$ are precisely the
equalities in (\ref{HMat}), while the equality constraint in (\ref{FMat})
defines the Matrix $F$. This completes the proof.
\end{proof}

\bigskip

\begin{proof}
[Proof of Proposition \ref{Abstraction}]First, consider the unreachability property.
Assume to the contrary, that $\mathcal{P}$ does not satisfy unreachability
w.r.t. $X_{-}$. Then, there exists a solution $\mathcal{X}\equiv x(.)$ of
$\mathcal{S}(X,f,X_{0},X_{\infty}),$ and a positive integer $t_{-}%
\in\mathbb{Z}_{+}$ such that $x\left(  0\right)  \in X_{0}$, and $x\left(
t_{-}\right)  \in X_{-}$. It follows from $X_{0}\subseteq\overline{X}_{0},$
and $f(x)\subseteq\overline{f}(x)$ that $\mathcal{X}\equiv x(.)$ is also a
solution of $\overline{\mathcal{S}}(\overline{X},\overline{f},\overline{X}%
_{0},\overline{X}_{\infty}).$ Therefore, we have:\vspace{-0.1in}%
\[
x\left(  t_{-}\right)  \in X_{-}\text{ , }X_{-}\subseteq\overline{X}_{-}\text{
}\Rightarrow\text{\ }x\left(  t_{-}\right)  \in\overline{X}_{-}.\vspace{-0.1in}
\]
However, $x\left(  t_{-}\right)  \in\overline{X}_{-}$ contradicts the fact
that $\overline{\mathcal{S}}(\overline{X},\overline{f},\overline{X}%
_{0},\overline{X}_{\infty})$ satisfies safety w.r.t. $\overline{X}_{-}$. Proof
of FTT is similar: let $\mathcal{X}\equiv x(.)$ be any solution of
$\mathcal{S}(X,f,X_{0},X_{\infty})$. Since $\mathcal{X}\equiv x(.)$ is also a
solution of $\overline{\mathcal{S}}(\overline{X},\overline{f},\overline{X}%
_{0},\overline{X}_{\infty}),$ it follows that there exists $t_{T}\in
\mathbb{Z}_{+}$ such that $x\left(  t_{T}\right)  \in\overline{X}_{\infty}.$
Since $x\left(  t_{T}\right)  $ is also an element of $X,$ it follows that
$x\left(  t_{T}\right)  \in\overline{X}_{\infty}\cap X.$ Since by definition
$\overline{X}_{\infty}\cap X\subset X_{\infty}$ holds$,$ we must have
$x\left(  t_{T}\right)  \in X_{\infty}.$
\end{proof}

\bigskip

\begin{proof}
[Proof of Proposition \ref{FTT2}]Note that (\ref{Softa2a1})$-$(\ref{Softa2a3}) imply
that $V$ is negative-definite along the trajectories of $\mathcal{S},$ except
possibly for $V\left(  x\left(  0\right)  \right)  $ which can be zero when
$\eta=0.$ Let $\mathcal{X}$ be any solution of $\mathcal{S}.$ Since $V$ is
uniformly bounded on $X$, we have:$\vspace*{-0.15in}$
\[
-\left\Vert V\right\Vert _{\infty}\leq V\left(  x\left(  t\right)  \right)
<0,\text{ }\forall x\left(  t\right)  \in\mathcal{X},~t>1.\vspace*{-0.15in}%
\]
Now, assume that there exists a sequence $\mathcal{X}\equiv(x(0),x(1),\dots
,x(t),\dots)$ of elements from $X$ satisfying (\ref{Softa1}), but not reaching
a terminal state in finite time. That is, $x\left(  t\right)  \notin
X_{\infty},$ $\forall t\in\mathbb{Z}_{+}.$ Then, it can be verified that if
$t>T_{u},$ where $T_{u}$ is given by (\ref{Bnd on No. Itrn.}), we must have:
$V\left(  x\left(  t\right)  \right)  <-\left\Vert V\right\Vert _{\infty}, $
which contradicts boundedness of $V.$
\end{proof}

\bigskip

\begin{proof}
[Proof of Theorem \ref{BddNess}]Assume that $\mathcal{S}$ has a solution
$\mathcal{X\hspace*{-0.03in}=\hspace*{-0.03in}}\left(  x\left(  0\right)
,...,x\left(  t_{-}\right)  ,...\right)  ,$ where $x\left(  0\right)  \in
X_{0}$ and $x\left(  t_{-}\right)  \in X_{-}.$ Let\vspace*{-0.2in}
\[
\gamma_{h}=\underset{x\in h^{-1}\left(  X_{-}\right)  }{\inf}V\left(
x\right)
\]
First, we claim that $\gamma_{h}\leq\max\left\{  V(x\left(  t_{-}\right)
),V(x\left(  t_{-}-1\right)  )\right\}  .$ If $h=I,$ we have $x\left(
t_{-}\right)  \in h^{-1}\left(  X_{-}\right)  $ and $\gamma_{h}\leq V(x\left(
t_{-}\right)  ).$ If $h=f,$ we have $x\left(  t_{-}-1\right)  \in
h^{-1}\left(  X_{-}\right)  $ and $\gamma_{h}\leq V(x\left(  t_{-}-1\right)
),$ hence the claim. Now, consider the $\theta=1$ case$.$ Since $V$ is
monotonically decreasing along solutions of $\mathcal{S},$ we must
have:$\vspace*{-0.1in}$%
\begin{equation}
\gamma_{h}=\underset{x\in h^{-1}\left(  X_{-}\right)  }{\inf}V\left(
x\right)  \leq\max\left\{  V(x\left(  t_{-}\right)  ),V(x\left(
t_{-}-1\right)  )\right\}  \leq V\left(  x\left(  0\right)  \right)
\leq~\underset{x\in X_{0}}{\sup}V(x)\vspace*{-0.1in}\label{Inf L than Sup}%
\end{equation}
which contradicts (\ref{Inf G than Sup})$.$ Note that if $\mu>0$ and $h=I,$
then (\ref{Inf L than Sup}) holds as a strict inequality and we can replace
(\ref{Inf G than Sup}) with its non-strict version. Next, consider case
$\left(  \text{I}\right)  ,$ for which, $V$ need not be monotonic along the
trajectories. Partition $X_{0}$ into two subsets $\overline{X}_{0}$ and
$\underline{X}_{0}$ such that $X_{0}=\overline{X}_{0}\cup\underline{X}_{0}$
and$\vspace*{-0.1in}$%
\[
V\left(  x\right)  \leq0\text{\quad}\forall x\in\underline{X}_{0},\text{\qquad
and\qquad}V\left(  x\right)  >0\text{\quad}\forall x\in\overline{X}_{0}%
\vspace*{-0.1in}%
\]
Now, assume that $\mathcal{S}$ has a solution $\overline{\mathcal{X}%
}\mathcal{=}\left(  \overline{x}\left(  0\right)  ,...,\overline{x}\left(
t_{-}\right)  ,...\right)  ,$ where $\overline{x}\left(  0\right)
\in\overline{X}_{0}$ and $\overline{x}\left(  t_{-}\right)  \in X_{-}.$ Since
$V\left(  x\left(  0\right)  \right)  >0$\ and $\theta<1,$ we have $V\left(
x\left(  t\right)  \right)  <V\left(  x\left(  0\right)  \right)
,\quad\forall t>0.$ Therefore,$\vspace*{-0.2in}$
\[
\gamma_{h}=\underset{x\in h^{-1}\left(  X_{-}\right)  }{\inf}V\left(
x\right)  \leq\max\left\{  V(x\left(  t_{-}\right)  ),V(x\left(
t_{-}-1\right)  )\right\}  \leq V\left(  \overline{x}\left(  0\right)
\right)  \leq~\underset{x\in X_{0}}{\sup}V(x)\vspace*{-0.1in}%
\]
which contradicts (\ref{Inf G than Sup})$.$ Next, assume that $\mathcal{S}$
has a solution $\underline{\mathcal{X}}\mathcal{=}\left(  \underline{x}\left(
0\right)  ,...,\underline{x}\left(  t_{-}\right)  ,...\right)  ,$ where
$\underline{x}\left(  0\right)  \in\underline{X}_{0}$ and $\underline
{x}\left(  t_{-}\right)  \in X_{-}.$ In this case, regardless of the value of
$\theta,$ we must have $V\left(  \underline{x}\left(  t\right)  \right)
\leq0,$ $\forall t,$ implying that $\gamma_{h}\leq0,$ and hence, contradicting
(\ref{Inf G than Zero})$.$ Note that if $h=I$ and either $\mu>0,$ or
$\theta>0,$ then (\ref{Inf G than Zero}) can be replaced with its non-strict
version. Finally, consider case $\left(  \text{II}\right)  $. Due to
(\ref{Sup L than Zero}), $V$ is strictly monotonically decreasing along the
solutions of $\mathcal{S}.$ The rest of the argument is similar to the
$\theta=1$ case.
\end{proof}

\vspace{0.15in}

\begin{proof}
[Proof of Corollary \ref{Bddness and FTT}]It follows from (\ref{Three3}) and the
definition of $X_{-}$ that:\vspace*{-0.15in}%
\begin{equation}
V\left(  x\right)  \geq\sup\left\{  \left\Vert \alpha^{-1}h\left(  x\right)
\right\Vert _{q}-1\right\}  \geq\sup\left\{  \left\Vert \alpha^{-1}h\left(
x\right)  \right\Vert _{\infty}-1\right\}  >0,\text{\qquad}\forall x\in
X.\vspace*{-0.15in}\label{o69}%
\end{equation}
It then follows from (\ref{o69}) and (\ref{One1}) that:\vspace*{-0.15in}%
\[
\underset{x\in h^{-1}\left(  X_{-}\right)  }{\inf}V\left(  x\right)
>0\geq~\underset{x\in X_{0}}{\sup}V(x)\vspace*{-0.15in}%
\]
Hence, the first statement of the Corollary follows from Theorem
\ref{BddNess}. The upperbound on the number of iterations follows from
Proposition \ref{FTT2} and the fact that $\sup_{x\in X\backslash\left\{
X_{-}\cup X_{\infty}\right\}  }\left\vert V\left(  x\right)  \right\vert
\leq1.$\vspace*{-0.1in}
\end{proof}

\vspace{0.15in}

\begin{proof}
[Proof of Corollary \ref{SafetyGraphCor1}]The unreachability property follows
directly from Theorem \ref{BddNess}. The finite time termination property
holds because it follows from (\ref{arcwiselyap}), (\ref{SGC1}) and
(\ref{MultiplicativeTheta}) along with Proposition \ref{FTT2}, that the
maximum number of iterations around every simple cycle $\mathcal{C}$ is
finite. The upperbound on the number of iterations is the sum of the maximum
number of iterations over every simple cycle.
\end{proof}

\vspace{0.1in}

\begin{proof}
[Proof of Lemma \ref{MIPL_Invariance_Lemma}]Define $x_{e}=\left(  x,w,v,1\right)
^{T},$ where $x\in\left[  -1,1\right]  ^{n},$ $w\in\left[  -1,1\right]
^{n_{w}},$ $v\in\left\{  -1,1\right\}  ^{n_{v}}.$ Recall that $\left(
x,1\right)  ^{T}=L_{2}x_{e},$ and that for all $x_{e}$ satisfying $Hx_{e}=0,$
there holds: $\left(  x_{+},1\right)  =\left(  Fx_{e},1\right)  =L_{1}x_{e}.$
It follows from Proposition \ref{prop:MILMLyap} that (\ref{Softa2}) holds
if:\vspace*{-0.15in}%
\begin{equation}
x_{e}^{T}L_{1}^{T}PL_{1}x_{e}-\theta x_{e}^{T}L_{2}^{T}PL_{2}x_{e}\leq
-\mu,\text{ s.t. }Hx_{e}=0,\text{ }L_{3}x_{e}\in\left[  -1,1\right]
^{n+n_{w}},\text{ }L_{4}x_{e}\in\left\{  -1,1\right\}  ^{n_{v}}.\vspace
*{-0.1in}\label{MILP_1}%
\end{equation}
Recall from the $\mathcal{S}$-Procedure ((\ref{Needs-S-Procedure}) and
(\ref{S-Procedure-sufficient})) that the assertion $\sigma\left(  y\right)
\leq0, $ $\forall y\in\left[  -1,1\right]  ^{n}$ holds if there exist
nonnegative constants $\tau_{i}\geq0,$ $i=1,...,n,$ such that $\sigma\left(
y\right)  \leq\sum\tau_{i}\left(  y_{i}^{2}-1\right)  =y^{T}\tau
y-\operatorname{Trace}\left(  \tau\right)  ,~$where $\tau=\operatorname{diag}%
\left\{  \tau_{i}\right\}  \succeq0.$ Similarly, the assertion $\sigma\left(
y\right)  \leq0,\forall y\in\left\{  -1,1\right\}  ^{n}$ holds if there exist
constants $\rho_{i}$ such that $\sigma\left(  y\right)  \leq\sum\rho
_{i}\left(  y_{i}^{2}-1\right)  =y^{T}\rho y-\operatorname{Trace}\left(
\rho\right)  ,$ where $\rho=\operatorname{diag}\left\{  \rho_{i}\right\}  .$
Applying these relaxations to (\ref{MILP_1}), we obtain sufficient conditions
for (\ref{MILP_1}) to hold:\vspace*{-0.15in}%
\[
x_{e}^{T}L_{1}^{T}PL_{1}x_{e}-\theta x_{e}^{T}L_{2}^{T}PL_{2}x_{e}\leq
x_{e}^{T}\left(  YH+H^{T}Y^{T}\right)  x_{e}+x_{e}^{T}L_{3}^{T}D_{xw}%
L_{3}x_{e}+x_{e}^{T}L_{4}^{T}D_{v}L_{4}x_{e}-\mu-\operatorname{Trace}%
(D_{xw}+D_{v})\vspace*{-0.1in}%
\]
Together with $0\preceq D_{xw},$ the above condition is equivalent to the LMIs
in Lemma \ref{MILP_Invariance_LMI}.\vspace*{-0.15in}
\end{proof}

\vspace{0.15in}

\begin{proof}
[Proof of Theorem \ref{Thm:GraphModelLMI}]Since $\theta_{j\emptyset}^{k}=0$ for all
$\left(  \emptyset,j,k\right)  \in\mathcal{E},$ we have $\sigma_{j}\left(
x\right)  \leq0,$ for all $j$ in the outgoing set of node $\emptyset.$ This
replaces condition (\ref{SGC1}) in Corollary \ref{SafetyGraphCor1}. The first
part of the theorem thus follows from Corollary \ref{SafetyGraphCor1} and an
application of the $\mathcal{S}$-Procedure relaxation technique in a similar
fashion to the proof of Lemma \ref{MIPL_Invariance_Lemma}. The second part of
the theorem also follows from Corollary \ref{SafetyGraphCor1} by noting that
due to (\ref{LGMLMI2.}), we have $\left\Vert \sigma\left(  \mathcal{C}\right)
\right\Vert _{\infty}\leq1$ for every simple cycle $\mathcal{C}\in G.$
\end{proof}


\newpage

\begin{thebibliography}{99}                                                                                               %
\bibitem {Adge11}A. Adge. Optimisation et jeux appliqu\'es \`a l'analyse
statique de programmes par interpr\'etation abstraite. Ph.D. Thesis, Ecole
Polytechnique, France, 2011.

\bibitem {Alur1995}R. Alur, C. Courcoubetis, N. Halbwachs, T. A. Henzinger,
P.-H. Ho X. Nicollin, A. Oliviero, J. Sifakis, and S. Yovine. The algorithmic
analysis of hybrid systems, \textit{Theoretical Computer Science}, vol. 138,
pp. 3--34, 1995.

\bibitem {Alur2002}R. Alur, T. Dang, and F. Ivancic. Reachability analysis of
hybrid systems via predicate abstraction. In \textit{Hybrid Systems:
Computation and Control}. LNCS v. 2289, pp. 35--48. Springer Verlag, 2002.

\bibitem {Baier}C. Baier, B. Haverkort, H. Hermanns, and J.-P. Katoen.
Model-checking algorithms for continuous-time Markov chains. \textit{IEEE
Trans. Soft. Eng.}, 29(6):524--541, 2003.

\bibitem {Bemporad Morari}A. Bemporad, and M. Morari. Control of systems
integrating logic, dynamics, and constraints. \textit{Automatica},
35(3):407--427, 1999.

\bibitem {Bemporad2000}A. Bemporad, F. D. Torrisi, and M. Morari.
Optimization-based verification and stability characterization of piecewise
affine and hybrid systems. LNCS v. 1790, pp. 45--58. Springer-Verlag, 2000.

\bibitem {Bertsimes1997}D. Bertsimas, and J. Tsitsikilis. \textit{Introduction
to Linear Optimization.} Athena Scientific, 1997.

\bibitem {ASTREE}B. Blanchet, P. Cousot, R. Cousot, J. Feret, L. Mauborgne, A.
Min\'{e}, D. Monniaux, and X. Rival. Design and implementation of a
special-purpose static program analyzer for safety-critical real-time embedded
software. LNCS v. 2566, pp. 85--108, Springer-Verlag, 2002.

\bibitem {Boshnack}J. Bochnak, M. Coste, and M. F. Roy. \textit{Real Algebraic
Geometry}. Springer, 1998.

\bibitem {Boyd1994}S. Boyd, L.E. Ghaoui, E. Feron, and V. Balakrishnan.
\textit{Linear Matrix Inequalities in Systems and Control Theory,} SIAM, 1994.

\bibitem {Branicky1998}M. S. Branicky. Multiple Lyapunov functions and other
analysis tools for switched and hybrid systems. \textit{IEEE Trans. Aut.
Ctrl.}, 43(4):475--482, 1998.

\bibitem {Branicky}M. S. Branicky, V. S. Borkar, and S. K. Mitter. A unified
framework for hybrid control: model and optimal control theory. \textit{IEEE
Trans. Aut. Ctrl.}, 43(1):31--45, 1998.

\bibitem {Brocket}R. W. Brockett. Hybrid models for motion control systems.
\textit{Essays in Control: Perspectives in the Theory and its Applications},
Birkhauser, 1994.

\bibitem {Clarkware}E. M. Clarke, O. Grumberg, H. Hiraishi, S. Jha, D.E. Long,
K.L. McMillan, and L.A. Ness. Verification of the Future-bus+cache coherence
protocol. In \textit{Formal Methods in System Design}, 6(2):217--232, 1995.

\bibitem {ClarkBook}E. M. Clarke, O. Grumberg, and D. A. Peled. \textit{Model
Checking}. MIT Press, 1999.

\bibitem {Cousot1977}P. Cousot, and R. Cousot. Abstract interpretation: a
unified lattice model for static analysis of programs by construction or
approximation of fixpoints. In \textit{4th ACM SIGPLAN-SIGACT Symposium on
Principles of Programming Languages}, pages 238--252, 1977.

\bibitem {Cousot2001}P. Cousot. Abstract interpretation based formal methods
and future challenges. LNCS, v. 2000:138--143, Springer, 2001.

\bibitem {FeronGiveAway}P. Cousot. Proving program invariance and termination
by parametric abstraction, Lagrangian relaxation and semidefinite programming.
In Verification, Model Checking, and Abstract Interpretation 6th International
Conference, VMCAI 2005. Proceedings in Lecture Notes in Computer Science v.
3385, 2005.

\bibitem {Dams1996}D. Dams. Abstract Interpretation and Partition Refinement
for Model Checking. Ph.D. Thesis, Eindhoven University of Technology, 1996.

\bibitem {FeronWorkshop}E. Feron. Abstraction mechanisms accross the
board:\ A\ short introduction. In \textit{A Workshop on Robustness,
Abstractions and Computations}, Philadelphia, March 18, 2004. Presentations
available at: http://web.mit.edu/feron/Public/wkshop\_pres

\bibitem{FeronRooz07} E. Feron, and M. Roozbehani. Certifying controls and
systems software. In \textit{Proc. of the AIAA Guidance, Navigation and
Control Conference}, Hilton Head, South Carolina, August 2007.

\bibitem {GahinetLMILAB}P. Gahinet, A. Nemirovskii, and A. Laub. LMILAB: A
Package for Manipulating and Solving LMIs. South Natick, MA: The Mathworks, 1994.

\bibitem {ILOG}ILOG Inc. ILOG CPLEX 9.0 User's guide. Mountain View, CA, 2003.

\bibitem {GirardPappsVerification}A. Girard, and G. J. Pappas. Verification
using simulation. LNCS, v. 3927, pp. 272--286 , Springer, 2006.

\bibitem {Goemans1995}M. X. Goemans, and D. P. Williamson, Improved
approximation algorithms for maximum cut and satisfiability problems using
semidefinite programming. \textit{Journal of the Association for Computing
Machinery (ACM)}, 42(6):1115--1145, 1995.

\bibitem {Gusev06}S. V. Gusev, and A. L. Likhtarnikov.
Kalman--Popov--Yakubovich Lemma and the $\mathcal{S}$-procedure: A historical
essay. \textit{Journal of Automation and Remote Control}, 67(11):1768--1810, 2006.

\bibitem {Heck2003}B. S. Heck, L. M. Wills, and G. J. Vachtsevanos. Software
technology for implementing reusable, distributed control systems.
\textit{IEEE Control Systems Magazine}, 23(1):21--35, 2003.

\bibitem {Heckt1977}M. S. Hecht. \textit{Flow Analysis of Computer Programs}.
Elsevier Science, 1977.

\bibitem {Johansson1998}M. Johansson, and A. Rantzer. Computation of piecewise
quadratic Lyapunov functions for hybrid systems. \textit{IEEE Tran. Aut.
Ctrl}. 43(4):555--559, 1998.

\bibitem {Khalil}H. K. Khalil. \textit{Nonlinear Systems}. Prentice Hall, 2002.

\bibitem {Kopetz2001}H. Kopetz. \textit{Real-Time Systems Design Principles
for Distributed Embedded Applications}. Kluwer, 2001.

\bibitem {Krzhanski1996}A. B. Kurzhanski, and I. Valyi. \textit{Ellipsoidal
Calculus for Estimation and Control}. Birkhauser, 1996.

\bibitem {Laferriere1999}G. Lafferriere, G. J. Pappas, and S. Sastry. Hybrid
systems with finite bisimulations. LNCS, v. 1567, pp. 186--203, Springer, 1999.

\bibitem {Lafferriere2001}G. Lafferriere, G. J. Pappas, and S. Yovine.
Symbolic reachability computations for families of linear vector fields.
\textit{Journal of Symbolic Computation}, 32(3):231--253, 2001.

\bibitem {Lofberg2004}J. L\"{o}fberg. YALMIP : A Toolbox for Modeling and
Optimization in MATLAB. In Proc. of the CACSD Conference, 2004. URL:
http://control.ee.ethz.ch/\symbol{126}joloef/yalmip.php

\bibitem {Lovasz1991}L. Lovasz, and A. Schrijver. Cones of matrices and
set-functions and 0-1 optimization. \textit{SIAM Journal on Optimization},
1(2):166--190, 1991.

\bibitem {Manna1995}Z. Manna, and A. Pnueli. \textit{Temporal Verification of
Reactive Systems: Safety}. Springer-Verlag, 1995.

\bibitem {Marrero}W. Marrero, E. Clarke, and S. Jha. Model checking for
security protocols. \textit{In Proc. DIMACS Workshop on Design and Formal
Verification of Security Protocols}, 1997.

\bibitem {Meg01}A. Megretski. Relaxations of quadratic programs in operator
theory and system analysis. Operator Theory: Advances and Applications, v.
129, pp. 365--392. Birkhauser -Verlag, 2001.

\bibitem {Megretski2003}A. Megretski. Positivity of trigonometric polynomials.
\textit{In Proc. 42nd IEEE Conference on Decision and Control}, pages
3814--3817, 2003.

\bibitem {MitraThesis}S. Mitra. \textit{A Verification Framework for Hybrid
Systems}. Ph.D. Thesis. Massachusetts Institute of Technology, September 2007.

\bibitem {MineThesis}A. Mine. \textit{Weakly Relational Numerical Abstract
Domains}. Ph.D. Thesis.
\'{}%
Ecole Normale Sup%
\'{}%
erieure, December 2004.

\bibitem {MinePaper}A. Mine. Relational abstract domains for detection of
floating point run-time errors, European symposium on programming (ESOP), LNCS
2986, pages 3-17, 2004.

\bibitem {Murthy2001}C. S. R. Murthy, and G. Manimaran. \textit{Resource
Management in Real-Time Systems and Networks}. MIT Press, 2001.

\bibitem {Naumovich}G. Naumovich, L. A. Clarke, and L. J. Osterweil.
Verification of communication protocols using data flow analysis. In
\textit{Proc. 4-th ACM SIGSOFT Symposium on the Foundation of Software
Engineering}, pages 93--105, 1996.

\bibitem {nem}G. L. Nemhauser and L. A. Wolsey. Integer and Combinatorial
Optimization. Wiley-Interscience, 1988.

\bibitem {Nesterov 2000}Y.E. Nesterov, H. Wolkowicz, and Y. Ye. Semidefinite
programming relaxations of nonconvex quadratic optimization. In
\textit{Handbook of Semidefinite Programming: Theory, Algorithms, and
Applications}. Dordrecht, Kluwer Academic Press, pp. 361--419, 2000.

\bibitem {Nielson}F. Nielson, H. Nielson, and C. Hank. \textit{Principles of
Program Analysis}. Springer, 2004.

\bibitem {Parrilo2001}P. A. Parrilo. Minimizing polynomial functions. In
\textit{Algorithmic and Quantitative Real Algebraic Geometry}. DIMACS Series
in Discrete Mathematics and Theoretical Computer Science, v. 60, pp. 83-100, 2003.

\bibitem {ParriloThesis}P. A. Parrilo. \textit{Structured Semidefinite
Programs and Semialgebraic Geometry Methods in Robustness and Optimization}.
Ph.D. Thesis, California Institute of Technology, 2000.

\bibitem {Pel01}D. A. Peled. \textit{Software Reliability Methods}.
Springer-Verlag, 2001.

\bibitem {Pierece}B. C. Pierce. \textit{Types and Programming Languages}. MIT
Press, 2002.

\bibitem {Prajna2005}S. Prajna. \textit{Optimization-Based Methods for
Nonlinear and Hybrid Systems Verification}. Ph.D. Thesis, California Institute
of Technology, 2005.

\bibitem {Prajna}S. Prajna, A. Papachristodoulou, P. Seiler, and P. A.
Parrilo, SOSTOOLS: Sum of squares optimization toolbox for MATLAB, 2004.
http://www.mit.edu/\symbol{126}parrilo/sostools.

\bibitem {Prajna2007}S. Prajna, and A. Rantzer, Convex programs for temporal
verification of nonlinear dynamical systems, \textit{SIAM Journal on Control
and Opt.}, 46(3):999--1021, 2007.

\bibitem {RMF-ACC05}M. Roozbehani, A. Megretski, E. Feron. Convex optimization
proves software correctness. In \textit{Proc. American Control Conference},
pages 1395--1400, 2005.

\bibitem {RFM-HSCC05}M. Roozbehani, E. Feron, and A. Megretski. Modeling,
optimization and computation for software verification. In \textit{Hybrid
Systems: Computation and Control}. Lecture Notes in Computer Science, v. 3414,
pp.\ 606--622, Springer-Verlag 2005.

\bibitem {RoozMegFer06}M. Roozbehani, A. Megretski, E. Feron. Safety
verification of iterative algorithms over polynomial vector fields. In
\textit{Proc. 45th IEEE Conference on Decision and Control}, pages 6061--6067, 2006.

\bibitem {RoozbehaniHSCC}M. Roozbehani, A. Megretski, E. Frazzoli, and E.
Feron. Distributed Lyapunov functions in analysis of graph models of software.
In Hybrid Systems: Computation and Control, Springer LNCS 4981, pp 443-456, 2008.

\bibitem {Sherali1994}H. D. Sherali, and W. P. Adams. A hierarchy of
relaxations and convex hull characterizations for mixed-integer zero-one
programming problems. \textit{Discrete Applied Mathematics}, 52(1):83--106, 1994.

\bibitem {Strum1999}J. F. Sturm. Using SeDuMi 1.02, a MATLAB toolbox for
optimization over symmetric cones. \textit{Optimization Methods and Software},
11--12:625--653, 1999. URL: http://sedumi.mcmaster.ca

\bibitem {Tiwari2002}A. Tiwari, and G. Khanna. Series of abstractions for
hybrid automata. In \textit{Hybrid Systems: Computation and Control}, LNCS, v.
2289, pp. 465--478. Springer, 2002.

\bibitem {VB}L. Vandenberghe, and S. Boyd. Semidefinite programming. SIAM
Review, 38(1):49--95, 1996.

\bibitem {Yakubovic}V. A. Yakubovic. S-procedure in nonlinear control theory.
Vestnik Leningrad University, 4(1):73--93, 1977.

\bibitem {MVichSite}http://web.mit.edu/mardavij/www/Software

\bibitem {ASTREE-Web} http://www.astree.ens.fr/

\end{thebibliography}
\end{document}